\documentclass[conference]{IEEEtran}
\IEEEoverridecommandlockouts
\usepackage{cite}
\usepackage{url}

\newif\ifcountdistinct

\usepackage[normalem]{ulem}
\newcommand{\fullversion}{}
\ifdefined\fullversion
    \countdistincttrue
\else
    \newcommand{\sigmodSubmission}{}
    \countdistinctfalse
\fi

\usepackage{etoolbox}

\ifcountdistinct
	\newcommand{\countdistinct}[1]{#1}
	\newcommand{\countdistinctapp}[1]{}
\else
    \newcommand{\countdistinct}[1]{}
    \newcommand{\countdistinctapp}[1]{#1}
    \newcommand{\icdeSubmission}{}
\fi

\newif\ifcomm
\commfalse

\ifcomm
	\newcommand{\mycomm}[3]{{\footnotesize{{\color{#2} \textbf{[#1: #3]}}}}}
	\newcommand{\CRdel}[1]{\textcolor{red}{\sout{#1}}}
\else
    \newcommand{\mycomm}[3]{}
    \newcommand{\CRdel}[1]{}
\fi

\newcommand{\MM}[1]{\mycomm{MM}{red}{#1}} 
 
\newcommand{\RBB}[1]{\mycomm{Ran}{purple}{#1}} 
 
\newcommand{\ran}[1]{\RBB{#1}}

\usepackage{caption}

\let\subparagraph\relax

\ifdefined\sigmodSubmission
\usepackage[compact]{titlesec}

\def\BibTeX{{\rm B\kern-.05em{\sc i\kern-.025em b}\kern-.08em T\kern-.1667em\lower.7ex\hbox{E}\kern-.125emX}}
\fi
\usepackage{wrapfig}

\usepackage{graphicx}
\usepackage{comment}

\usepackage{amsthm}
\usepackage{amsmath}
\usepackage{amssymb}
\usepackage{algorithm}
\usepackage{subfig}
\usepackage{xcolor}
\usepackage{thm-restate}
\newcommand{\Var}{\mathrm{Var}}

\newcommand{\eps}{\epsilon}
\newcommand{\set}[1]{\left\{#1\right\}}
\newcommand{\brackets}[1]{\left[#1\right]}
\newcommand{\angles}[1]{\left\langle#1\right\rangle}
\newcommand{\ceil}[1]{ \left\lceil{#1}\right\rceil}
\newcommand{\floor}[1]{ \left\lfloor{#1}\right\rfloor}

\newcommand{\parentheses}[1]{ \left({#1}\right)}

\newcommand{\norm}[1]{\left\lVert#1\right\rVert}

\usepackage{balance}

\title{SALSA: Self-Adjusting Lean Streaming Analytics}

\author{\IEEEauthorblockN{Ran Ben Basat}
	\IEEEauthorblockA{University College London
	}
	\and
	\IEEEauthorblockN{Gil Einziger}
	\IEEEauthorblockA{Ben Gurion University
	}
	\and
	\IEEEauthorblockN{Michael Mitzenmacher }
	\IEEEauthorblockA{Harvard University	
	}
	\and
	\IEEEauthorblockN{Shay Vargaftik}
	\IEEEauthorblockA{
		VMware Research
	}
}

\newtheorem{theorem}{Theorem}[section]
\newtheorem{lemma}[theorem]{Lemma}

\newtheorem*{observation*}{Observation}
\renewcommand{\textcolor}[2]{#2}
\newcommand{\rev}[1]{\textcolor{black}{#1}}
\newcommand{\CR}[1]{\textcolor{blue}{#1}}
\newcommand{\del}[1]{}

\begin{document}
\maketitle
\begin{abstract}
\rev{Counters are the fundamental building block of many data sketching schemes, which hash items to a small number of counters and account for collisions to provide good approximations for frequencies and other measures.}
Most existing methods rely on fixed-size counters, which may be wasteful in terms of space, as counters must be large enough to eliminate any risk of overflow.
Instead, some solutions use small, fixed-size counters that may overflow into secondary structures. 

This paper takes a different approach. We propose a simple and general method called SALSA for \emph{dynamic re-sizing} of counters, and show its effectiveness.
\rev{SALSA starts with small counters, and overflowing counters simply merge with their neighbors.}
\rev{SALSA can thereby allow more counters for a given space, expanding them as necessary to represent large numbers.}
Our evaluation demonstrates that, at the cost of a small overhead for its merging logic, SALSA significantly improves the accuracy of popular schemes (such as Count-Min Sketch and \mbox{Count Sketch) over a variety of tasks.}
\mbox{Our code is released as open source~\cite{SALSACode}.}

\end{abstract}


\section{Introduction}
Analysis of large data streams is essential in many domains, including natural language processing~\cite{NLPsketches}, load balancing~\cite{LoadBalancing}, and forensic analysis~\cite{Forensic2}.
Typically, the data volume renders exact analysis algorithms too expensive. 
However, often it is sufficient to {estimate} measurements such as per-item frequency~\cite{ConextPaper}, item distribution entropy~\cite{DBLP:conf/sigmetrics/LallSOXZ06}, or top-$k$/heavy hitters~\cite{RAP} by using {approximation algorithms} {often referred to as sketches. }
\CR{Sketching schemes reduce the space requirements by sharing counters that keep frequency counts of the (potentially multiple) associated items~\cite{AugmentedSketch,ColdFilter}.}
  
That is, rather than use a counter for each item, which may be space-prohibitive, sketches bound the effect of collisions to \mbox{guarantee good approximations.}

\CR{
A common approach for sketch design is to consider counters as the basic building block. 
Namely, the goal is to optimize the accuracy for a given number of counters (e.g.,~\cite{CountMinSketch,RAP}).
However, these works do not discuss how many bits each counter should have, a quantity whose optimal value depends on the workload and optimization metric.
}

For fixed-size counters, if they are too large, space is wasted. 
Conversely, if they are too small, there are risks of overflow. 
Instead, some solutions use small fixed-size counters that may overflow into secondary structures (e.g.,~\cite{Brick,PyramidSketch}).

{
\CR{\textbf{Our Contributions:}~We present Self-Adjusting Lean Streaming Analytics (SALSA), a  simple and  general framework for dynamic re-sizing of counters.
In a nutshell, SALSA starts with small (e.g., $8$-bit) counters and merges overflowing ones with their neighbors to represent larger numbers. This way, more counters {fit in a given space without limiting the counting range.}}
}
\CR{\noindent
To do so efficiently, we employ novel methods for representing merges with low memory and computation overheads.
These methods also respects byte boundaries making them readily {implementable in software and some hardware platforms.}
}
%
%

\CR{SALSA integrates with popular sketches and probabilistic counter-compression techniques to improve their precision to memory tradeoff.} 
\CR{We prove that SALSA stochastically improves the accuracy of standard schemes, including the Count Min Sketch~\cite{CountMinSketch}, the Conservative Update Sketch~\cite{CUSketch}, \mbox{and the Count Sketch~\cite{CountSketch}.}
Using different workloads, metrics, and tasks, we also show significant accuracy improvements for the above schemes as well as for state of the art solutions like Univmon~\cite{univmon}, Cold Filter~\cite{ColdFilter}, and AEE~\cite{Compsketch}.
\rev{We also compare against Pyramid Sketch~\cite{PyramidSketch} and ABC~\cite{gong2017abc}, recent \CR{variable-counter-size} solutions, and show that SALSA is more accurate than both.}
\mbox{Finally, we release our code as open source~\cite{SALSACode}.
}
}
\vspace*{-1.0mm}

\section{Related Work}
\ifdefined\sigmodSubmission\vspace*{-1mm}\fi
\CR{
The term \emph{sketch} here informally describes an algorithm that uses shared counters, such that each item is associated with a subset of the counters via hash functions~\cite{CountMinSketch,CountSketch,CUSketch,univmon}.
\CRdel{
Examples include the Count Min, the Conservative Update, and the Count sketches.
{
Some modern sketches, such as the Universal Sketch~\cite{univmon,UnivMonTheory}, can estimate \emph{all} tractable functions of the frequency vector, and the Cold Filter~\cite{ColdFilter}  which improves the accuracy and speed by using an aggregation buffer and filtering cold elements before updating the sketch.}
}
Sketches offer tradeoffs between update speed, accuracy, and space, where each of these parameters is important in some scenarios. For example, in software-based network measurement, we are primarily concerned about update speed~\cite{Nitro}. Conversely, in hardware-based measurements, \mbox{space is often the bottleneck~\cite{PRECISION,CASE}.}
}

\CR{
Some sketches optimize the update speed at the expense of space. 
For example, Randomized Counter Sharing~\cite{RandomizedCounterSharing} uses multiple hash functions but only updates a random one.
NitroSketch~\cite{Nitro} extends this idea and only performs updates for sampled packets using a novel sampling technique that asymptotically improves over uniform sampling.
}
\CR{
Other solutions aim to maximize the accuracy for a given space allocation.
For example, Counter Braids~\cite{CounterBraids} and Counter Tree~\cite{CounterTree} aim to fit into the fast static RAM (SRAM) while optimizing the precision.
These solutions estimate element sizes using complex offline procedures that, while being highly accurate, may be too slow for online applications. 
}

\CR{
Most relevant to our setting are ABC~\cite{gong2017abc} and Pyramid 
Sketch~\cite{PyramidSketch}, which vary the size of counters on the fly. 
In ABC, an overflowing counter is allowed to ``borrow'' bits from the next counter. If there are not enough bits to represent both values, the counters ``combine'' to create a larger counter.
However, the encoding of ABC is cumbersome. It requires three bits to mark combined counters (e.g., when starting with $8$-bits, combined counters can count to $2^{13}-1$) and slows the sketch down significantly (see Section~\ref{sec:eval}). Moreover, it does not allow counters to combine more than once.
Pyramid Sketch~\cite{PyramidSketch} has several layers for extending overflowing counters. 
An overflowing counter increases a counter at the next layer. 
Each pair of same-layer counters are associated with a single counter at the next layer. If both overflow, they will share their most significant bits in that counter while keeping the least significant bits separately.
Critically, the counters of all layers are pre-allocated regardless of the access patterns. 
This results in inferior memory utilization since many of the upper layers' counters may never be used. 
Further, when reading a counter, Pyramid may make multiple non-sequential memory accesses, thus slowing the processing down.
SALSA improves over these solutions due to its efficient encoding and the fact that its counting range is not limited by the initial configuration (e.g., counter size).
}

\rev{
\CR{An orthogonal line of} works reduces the size of counters by using  probabilistic estimators that  only increment their value with a certain probability on an update~\cite{CEDAR,ANLSUpscaling,Compsketch,ApproximateCounting}.}
Such an approach saves space as estimators can represent large numbers with fewer bits\CR{, at the cost of a higher error. }

\ifdefined\sigmodSubmission\vspace*{-1mm}\fi
\section{Preliminaries}
\ifdefined\sigmodSubmission\vspace*{-1mm}\fi
\CRdel{This section gives the relevant vocabulary used in this paper. }
\CR{
We consider a data stream $S$  consisting of \emph{updates} in the form of $\angles{x,v}$, where $x\in U$ is an element (or item) and $v\in\mathbb Z$
\ran{was $v\in\set{-M,\ldots,M}$ }
is a value.
Here, $U\triangleq\set{1,\ldots,u}$ is the \emph{universe} and $u$ is the \emph{universe size}.\ran{, and $M$ bounds the update value.}}
For $x\in\set{1,\ldots,u}$, $f_x\triangleq\sum_{\angles{x,v}\in S} v$ denotes the \emph{frequency} of $x$. \CR{Additionally}, $f\triangleq\angles{f_1,\ldots,f_u}$ is the \emph{frequency vector} of $S$.
We denote by $N\triangleq\sum_{x\in U} |f_x|$ the {volume} of the stream. 
\CR{
The above is called the \emph{Turnstile} model. 
Other models include the \emph{Strict Turnstile} model, where frequencies are non-negative at all times, and the \emph{Cash Register} model, where updates are strictly positive.}

The $p$'th moment of the frequency vector is defined as $F_p\triangleq\sum_{x\in U} |f_i|^p$ (e.g., $F_1=N$) and the $p$'th norm (defined for $p\ge 1$) is $L_p\triangleq\sqrt[p]F_p$.
We say that an algorithm estimates frequencies with an $(\epsilon,\delta)$
$L_p$ guarantee if for any element $x\in U$ it produces an estimate $\widehat{f_x}$ \mbox{that satisfies $\Pr\brackets{|\widehat{f_x}-f_x|\le \epsilon L_p}\ge 1-\delta$.} 
Throughout the paper, we assume the standard RAM model and that each counter value fits into $O(1)$ machine words.

We survey several popular sketches that SALSA extends.
\textbf{Count Min Sketch (CMS)~\cite{CountMinSketch}: }\label{sec:CMS}
CMS is arguably the simplest and most popular sketch.  It provides an $L_1$ guarantee in the Strict Turnstile model.
The sketch consists of a $d\times w$ matrix $C$ of \emph{counters} and $d$ random hash functions $h_1,\ldots,h_d:U\to[w]$ that map elements into counters. 
Each element $x$ is associated with one counter in each row: $C[1,h_1(x)],\ldots,C[d,h_d(x)]$. When processing the update $\angles{x,v}$, CMS adds $v$ to all of $x$'s counters.
Since CMS operates in the Strict Turnstile model where \rev{all frequencies are non-negative}, each of $x$'s counters provides an over-estimation for its true frequency (i.e., $\forall i\in[d]:C[i,h_i(x)]\ge f_x$). Therefore, CMS uses the \rev{minimum} of $x$'s counters to estimate $f_x$. That is,  $\widehat{f_x}\triangleq\min_{i\in[d]} C[i,h_i(x)]$.

For its analysis, denote by $\mathfrak E_i \triangleq C[i,h_i(x)] - f_x\ge 0$ the \emph{estimation error} of the $i$'th counter of $x$.

Notice that $\mathbb E[\mathfrak E_i]=\frac{N-f_x}{w}\le \frac{N}{w}$, and according to Markov's \mbox{inequality we have that}
{\small
\begin{align}
    \forall c>0, i\in[d]:\Pr[\mathfrak E_i\ge N\cdot c/w]\le 1/c.\label{eq:cms}
\end{align}}
We note that CMS, like all the algorithms below, provides a curve of guarantees, in that setting $\delta$ determines the $\epsilon$ value for which we have an $(\epsilon,\delta)$ guarantee with the $d\times w$ configuration.
Setting $\epsilon=\delta^{-1/d}/w$ and $c=\delta^{-1/d}$, equation \eqref{eq:cms} gives that $\Pr[\mathfrak E_i\ge N\epsilon]\le \delta^{1/d}$, and as the $d$ rows are independent we get that $\Pr[\forall i:\mathfrak E_i\ge N\epsilon] = \parentheses{\Pr[\mathfrak E_i\ge N\epsilon]}^d\le \delta$.
For fixed $(\epsilon,\delta)$ values, setting $w=e/\epsilon$ and $d=\ln\delta^{-1}$ minimizes the space required by the sketch, but CMS is often configured with a smaller \mbox{number of rows $d$ since its update and query time are $O(d)$.}

\textbf{Conservative Update Sketch (CUS)~\cite{CUSketch}:}
CUS improves the accuracy of CMS but is restricted to the Cash Register model. 
Intuitively, when all the update values are positive, we may not need to increase all the counters of the current element. For example, assume that $C[1,h_1(x)]=7$ and $C[2,h_2(x)]=4$, and the update $\angles{x,2}$ arrives. In such a scenario, we know that $f_x\le 4$ before the update, and thus should not increase $C[1,h_1(x)]$.
In general, given an update $\angles{x,v}$, CUS sets each counter $C[i,h_i(x)]$ to $\max\set{C[i,h_i(x)],v+\widehat{f_x}}$, where $\widehat{f_x}=\min_{i\in[d]} C[i,h_i(x)]$ is the estimate for $x$ \emph{before the update}. While CUS improves the accuracy of CMS, its updates are slower due to the need to compute $\widehat{f_x}$ before increasing the counters. Since an estimate of CUS is always bounded by CMS's estimates from above (and by $f_x$ from below), the analysis of CMS holds for CUS as well. We refer the reader to~\cite{CUAnal} for a refined analysis.

\textbf{Count Sketch (CS)~\cite{CountSketch}:}
CS works in the more general Turnstile model and provides the stronger $L_2$ guarantee.
As with CMS and CUS, each element $x$ is associated with a set of counters $\set{C[i,h_i(x)]}_{i\in[d]}$. However, the update process is slightly different. Each row $i\in[d]$ in CS has another pairwise independent hash function $g_i:U\to\set{+1,-1}$ that associates each element with a \emph{sign}. When processing an update $\angles{x,v}$, CS increases each counter $C[i,h_i(x)]$ by $v\cdot g_i(x)$. 
Intuitively, this ``unbiases'' the noise from all other elements as they increase or decrease the counters with equal probabilities. 
As a result, each counter now gives an unbiased estimate and therefore CS estimates the size as
$\widehat{f_x}\triangleq\text{median} \set{C[i,h_i(x)]\cdot g_i(x)}_{i\in[d]}$.

Assuming without loss of generality that $g_i(x)=1$, the standard CS analysis bounds the error of the $i$'th row, $\mathfrak E_i\triangleq C[i,h_i(x)] - f_x$, by showing that $\Var[\mathfrak E_i]\le
    F_2 / w$.
Therefore, using Chebyshev's inequality we get that 
$\Pr[|\mathfrak E_i|\ge c\sqrt{\Var[\mathfrak E_i]}] \le \Pr[|\mathfrak E_i|\ge cL_2/\sqrt{w}]\le 1/c^2$. By setting $w=\Theta(\epsilon^{-2})$, we can get $\Pr[|\mathfrak E_i|\ge L_2\cdot \epsilon]\le 1/2-\Omega(1)$, and then use a Chernoff bound to show that $d=O(\log \delta^{-1})$ rows are enough for an $(\epsilon,\delta)$ guarantee.

\textbf{Universal Sketch (UnivMon)~\cite{univmon,UnivMonTheory}: }
UnivMon summarizes the data once and supports \emph{many} functions of the frequency vectors (e.g., its entropy or number of non-zero entries) in the Cash Register model.
Importantly, when using UnivMon, we provide a function $G:\mathbb Z\to \mathbb R$ as an input, and estimate the \emph{G-sum}, \mbox{given by $\sum_{x\in U} G(f_x)$. }
Not all functions of the frequency vector can be computed in poly-log space in a one-pass streaming setting (a class called \emph{Stream-PolyLog}). The surprising result of~\cite{UnivMonTheory} is that any function $G$ in Stream-PolyLog is supported by UnivMon.

UnivMon leverages $O(\log u)$ sketches with an $L_2$ guarantee (e.g., Count Sketch), which are applied to different subsets of the universe. We refer the reader to~\cite{univmon,UnivMonTheory} for details. \CRdel{about the construction of UnivMon and the precise definition of Stream-PolyLog.}

\textbf{Cold Filter~\cite{ColdFilter}: }
\rev{
A recent framework for fast and accurate stream processing. It consists of two \emph{stages}, where the first stage is designed to filter cold items and the second measures heavy hitters accurately. To accelerate the computation, it \mbox{uses an aggregation buffer and employs SIMD parallelism.}
}


\textbf{Finding Heavy Hitters:}
Often, we care about finding the most significant elements in a data stream, which has applications for load balancing~\cite{LoadBalancing}, accounting, and security. That is, in addition to estimating the frequency of elements, we wish to track the most frequent elements without needing to query each $x\in U$. For $p\ge 1$, the $L_p$-heavy hitter problem asks to return all elements with frequency larger than $\theta L_p$ and no element smaller than $(\theta-\epsilon)L_p$, where $\theta\in[0,1]$ is given at query time.
In the Cash Register model, we can store a min-heap with the $1/\epsilon$ elements with the highest estimates. Whenever an update arrives, we query the item and update the heap if necessary. As a result, we can find the $L_1$ heavy hitters using CMS and CUS, or the $L_2$ heavy \mbox{hitters using CS.}

\ifdefined\fullversion
{\textbf{{Counting Distinct Items:}}
Estimating the number of distinct items in a data stream (defined as $F_0\equiv\norm f_0$) is a fundamental primitive for applications such as discovering denial of service attacks~\cite{IntrusionDetection2}.
While UnivMon can natively support such a function, we can also estimate it from CMS and CUS.
By observing the fraction of zero-valued counters in a sketch's row $p$, we can estimate the number of distinct elements (as additional occurrences of the same element would not change this quantity). 
Specifically, a common approach (e.g.,~\cite{elasticsketch}) is use the Linear Counting algorithm~\cite{LinearCounting} that estimates the distinct count as $\frac{\log p}{\log(1-1/w)}\approx -w\log p$. Such an estimate has a standard error of $\frac{\sqrt {w\cdot (e^{\frac{F_0}{w}}-\frac{F_0}{w}-1)}}{F_0}$~\cite{LinearCounting} that \mbox{improves when $w$ grows.}} 
\else
\rev{
{\textbf{{Counting Distinct Items:}}}
Sketches can also estimate the number of distinct elements in the stream (e.g., as suggested in~\cite{elasticsketch}). In the full version~\cite{fullVersion}, we demonstrate how SALSA can improve their count distinct performance as well.
}
\fi

\section{Techniques}
The description of the above sketches does not address the fundamental question of sizing the counters.  A common practice is to assume some upper bound on the maximal frequency (e.g., $\overline N$) and allocate each counter with $\overline n=O(\log \overline N)$ bits. For performance, this upper bound is often rounded up to be a multiple of the word size. 
For example, practitioners often allocate $32$-bit counters when estimating the \emph{unit-count} of {elements}, and $64$-bit counters for measuring their weighted-frequency (e.g.,~\cite{FAST,CormodeCode}). 
When space is tight, estimators are sometimes integrated into sketches to allow smaller (e.g., $16$-bit) per-counter overhead at the cost of additional error~\cite{Compsketch}.
However, these solutions miss the potential of allowing counters' bit sizes to vary and adjust dynamically.  Intuitively, the \CR{largest} counter value is often considerably larger than the average value, especially in highly skewed workloads where many counter values remain small as most of the volume belongs to a small set of \emph{heavy hitters}. 

Alternatively, one can use address-calculation coding (e.g., see~\cite{teuhola2011interpolative,elmasry2012improved}) to encode a variable length counter array in near-optimal space (compared to the information theoretic lower bound). \CR{Such schemes require an upper bound $N_{\max}$ on the volume, and use $w\log_2 (1+N_{\max}/w) + O(w)$. However, the update time of such encoding is  $\Omega(\log^2 N_{\max})$ which may be  prohibitive \mbox{for high-performance applications.} }
To the best of our knowledge, no implementation that combines such encoding with sketches has been proposed.
\CR{In comparison, SALSA allows for dynamic counter sizing by \emph{merging} overflowing counters with their neighbors, and optimizes for performance by respecting word alignments. }
A simple SALSA encoding requires one bit per counter, and an optimized encoding requires less than $0.6$ bits per counter while still allowing for constant-time read and update operations. 
\CR{Importantly, SALSA resolves overflows without dynamic memory allocations (e.g.,~\cite{SpectralBloom}), without relying on additional data structures  (as in~\cite{CASE}), and without requiring global \mbox{rescaling operations for all the counters (e.g.,~\cite{Compsketch}).} }
\CRdel{Thus, SALSA poses a significant advantage \mbox{over the alternatives. }}



\textbf{The SALSA encoding:}\label{sec:SALSA_encoding}
SALSA starts with all counters having $s$ bits (e.g., $s=8$), where $s$ may be significantly smaller than the intended counting range (e.g., $\overline N=2^{32}$). 
Here, we describe an encoding that requires one bit of overhead per counter (e.g., 12.5\% for $s=8$ bit counters); we later explain how to reduce it to less than \mbox{$0.6$ bits (7.5\% for $s=8$).}

Each counter $i$ is associated with a \emph{merge bit} $m_i$. 
Once a counter needs to represent a value of $2^s$, we say that the counter \emph{overflows}. In principle, an overflowing counter can merge with its left-neighbor or right-neighbor. 
In SALSA, we select the merge direction to maximize byte and word alignment, which improves performance. We also make counters grow in powers of two (e.g., from $s$ bits to $2s$, then to $4s$, etc.). 
\ifdefined\fullversion
In Section~\ref{sec:tango}, we explore a slower but more fine-grained approach.\ran{This talks about tango, combine with the smarter encoding part}
\else
\fi
Specifically, when an $s$-bit counter $i$ overflows, it merges with $i+(1-2\cdot (i\mod 2))$. For example, if counter $6$ overflows, it merges with $7$, \rev{while if counter 7 overflows, it merges with $6$}.
More generally, when an $s\cdot 2^{\ell}$-bit counter with indices $\angles{i\cdot 2^{\ell}, i\cdot 2^{\ell}+1\ldots (i+1)\cdot 2^{\ell}-1}$ overflows, it merges with the counter-set at indices $\angles{j\cdot 2^{\ell}, j\cdot 2^{\ell}+1\ldots (j+1)\cdot 2^{\ell}-1}$, for $j=(1-2\cdot (i\mod 2))$.
As an example, if we started from $s=8$ bit counters and counter $6$ overflows, it right-merges with $7$ to create a $16$ bit counter with indices $\angles{6,7}$. If this counter overflows, it left-merges into a $32$ bit counter with indices $\angles{4,5,6,7}$, and if this overflows, it left-merges into a $64$ bit counter \mbox{with indices $\angles{0,\ldots,7}$.}

\rev{To encode that $\angles{i\cdot 2^{\ell}, i\cdot 2^{\ell}+1,\ldots, (i+1)\cdot 2^{\ell}-1}$ are merged into a single $s\cdot 2^{\ell}$-bit counter, SALSA sets $m_{i\cdot 2^{\ell}+2^{\ell-1}-1}=1$. For example, to encode that $\angles{6,7}$ are merged, we have $(i=3,\ell=1)$ and thus set $m_{3\cdot 2^1 + 2^{1-1}-1}=m_6=1$; when $\angles{4,5,6,7}$ are merged, we have $(i=1,\ell=2)$ and thus we set $m_{1\cdot 2^2 + 2^{2-1}-1}=m_5=1$; and when $\angles{0,\ldots,7}$ are merged we have ($i=0,\ell=3)$ and thus we set $m_{0\cdot 2^3 + 2^{3-1}-1}=m_3=1$.}
We can compute the counter size by testing the $\ell$ relevant \rev{bits.}
We demonstrate this encoding in Figure~\ref{fig:encoding}.
All the computations involved in determining the counter size and offset can be efficiently implemented using bit operations, \mbox{especially if $s$ is a power of two.}\\

\begin{figure}[t]
\centering
\ifdefined\sigmodSubmission\vspace*{2mm}\fi
\includegraphics[width = 1\columnwidth]
		{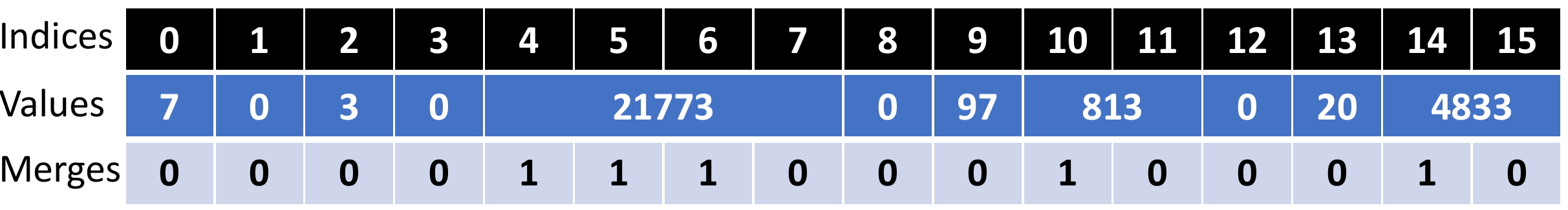}
    
	\ifdefined\submissionVersion
	\ifdefined\sigmodSubmission\vskip -0.1cm\fi
	\fi
    \caption{\label{fig:encoding} SALSA encoding for an array with a basic counter size of $s=8$ bits, notice that large counters consume more indices than small counters due to merge operations.
    }
\end{figure}

\textbf{{Reducing the Encoding Overhead}:}\label{sec:smartEncoding}
The encoding we used in SALSA so far is efficient as well as amenable for simple implementation. The cost of this encoding is a single merge bit per counter. 
\ifdefined\fullversion
This is, in fact, within a factor of 2 of the optimal encoding, as we show in  Appendix~\ref{app:improvedEncoding}. 
That is, we prove that any encoding for SALSA must use at least $\log_2 1.5\approx 0.585$ overhead bits per counter and show a somewhat more complex  $O(1)$-time encoding with at most $0.594$ overhead bits per counter. 
For a given memory allocation, this encoding provides improved accuracy as the lower overhead allows fitting more counters, but may be somewhat slower.
\else
\rev{In the full version~\cite{fullVersion}, we show that any encoding must use at least $\log_2 1.5\approx 0.585$ bits per counter and present a near-optimal, albeit slower, constant time encoding with $0.594$-bits per counter.}
\fi
\ifdefined\fullversion\else
We also explore a slower but more fine-grained approach \rev{where counters can merge to e.g., $3s$ bits. Our analysis and evaluation indicate that the benefit of such an approach is limited, \mbox{given the computational overheads.}}
\fi

\ifdefined\icdeSubmission
\else
\textbf{Fine-grained Counter Merges:}\label{sec:tango}
The SALSA encoding we presented in Section~\ref{sec:SALSA_encoding} doubles the counter size upon an overflow, which may be wasteful when the overflowing counter could benefit from a smaller increase in size. 
Thus, we suggest the more refined Tango algorithms to explore the benefits of a more fine-grained merging strategy. 
In Tango, counters can be merged into sizes that are arbitrary multiples of $s$. For example, if we start from $s=8$ bit counters, Tango can merge a $16$ bit counter into a $24$ bit counter while SALSA would merge from $16$ bits to $32$. The encoding of Tango is simple: each counter $j$ is associated with a merge bit $m_j$ that denotes whether the counter is merged with its right-neighbor.
To compute the counter size and offset in Tango of $j=h(x)$, we scan the number of set bits to the left and right of $m_j$ until we hit a zero at both sides. 
For example, if $j = 5$ and $m_4=m_5=m_6=m_7=1$ while $m_3=m_8=0$ then the counter consists of $s\cdot 5$ bits, spanning $\angles{4,5,6,7,8}$.
In general, one can use complex logic to decide whether to merge with the left or right neighbor once a counter overflows. However, we design Tango to evaluate the potential benefits of fine-grained merging and therefore enforce a merging logic that mimics SALSA. 
Specifically, Tango always tries to be aligned to the smallest possible power of two. For example, if counter $9$ overflows, it merges with $8$ to be aligned with the $2$-block $\angles{8,9}$. If it overflows again, it merges with $10$ (creating a $s\cdot 3$ bits sized counter) and then with $11$. If more bits are needed it will merge with $12$ then with $13,14$ and $15$ (being aligned to the $8$-block $\angles{8,\ldots,15}$). Then it merges with $7, 6,\ldots$, etc. Notice that at every point in time, the Tango counters are contained in the corresponding SALSA counters. In particular, this allows us to produce an estimate that is at least as accurate as SALSA.
We note that Tango poses a tradeoff -- while it allows more accurate sketches (e.g., as a counter may not exceed $2^{24}-1$ and thus it could be wasteful to merge it into $32$ bits), it also has slower decoding time and cannot use the efficient encoding of the previous section.
\fi



\section{SALSA-fying Sketches}
We now describe how SALSA integrates with existing sketches, and specifically how to set the value of merged counters in each sketch. We also state and prove accuracy guarantees for the resulting SALSA sketches.
We employ hash functions $h_i:U\to[w]$ similarly to the original sketches. Given a merged counter with indices $\angles{L,L+1,\ldots,R}$, we consider all elements $x$ with $L\le h_i(x)\le R$ to be mapped into it.
Hereafter, we often refer to the \emph{underlying sketch} as following: If the largest merged counter size is $s\cdot 2^{\ell}$, the underlying sketch is a vanilla (fixed counter size) sketch where \emph{each counter} is of size $s\cdot 2^{\ell}$ and its hashes are $\set{\widetilde{h}_i(x) \triangleq \floor{h_i(x)/2^\ell}\mid i\in[d]}$.

\begin{figure}[t]
    \ifdefined\sigmodSubmission\centering\ifdefined\sigmodSubmission\vskip -0.15cm\fi\fi
    \ifdefined\sigmodSubmission\vspace*{-2mm}\fi
    \subfloat[\rev{Sum merging of counters}]
    {\label{fig:sum-merge}\includegraphics[width =0.48\columnwidth]
    {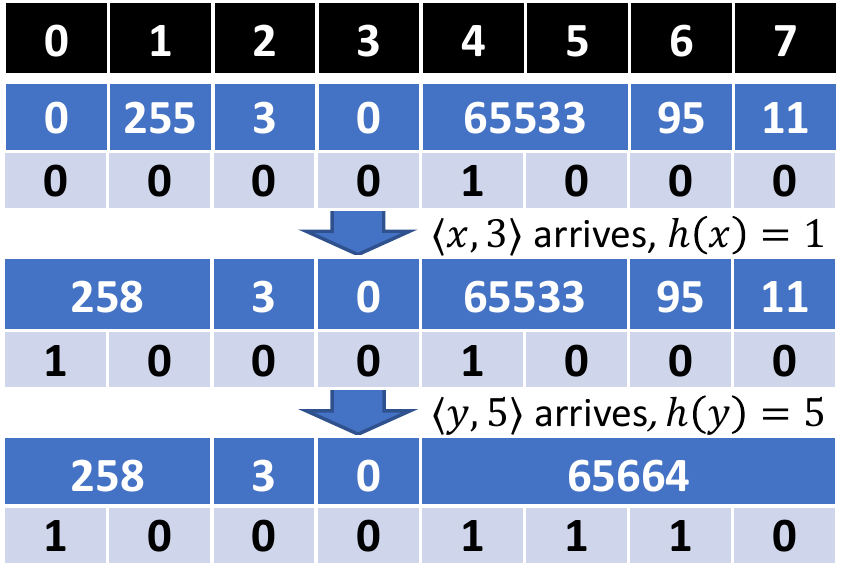}}    
    \hspace{1mm}
    \subfloat[\rev{Max merging of counters}]
    {\label{fig:max-merge}\includegraphics[width =0.48\columnwidth]
    {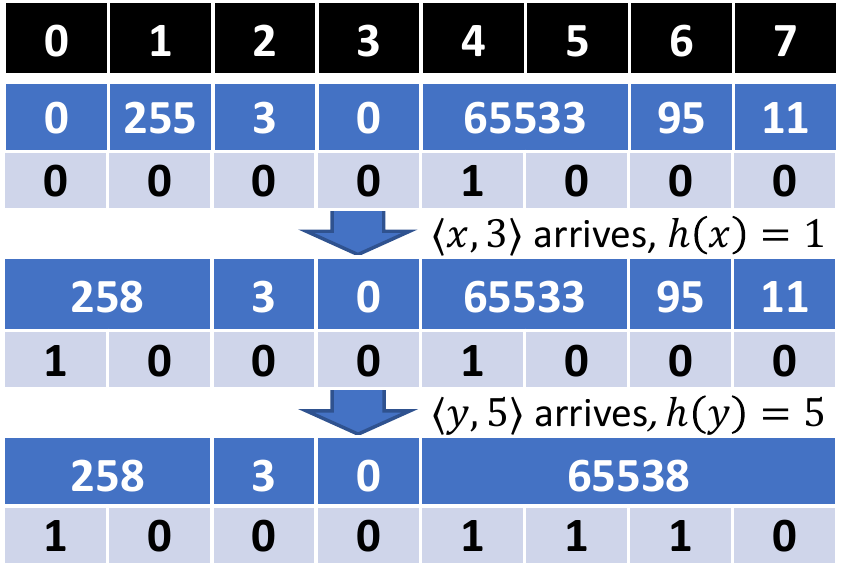}}
    \ifdefined\sigmodSubmission\vspace*{-2mm}\fi
    \caption{Sum and Max merge in SALSA CMS with $s=8$.
    \label{fig:cache-based}}\ifdefined\sigmodSubmission\vspace*{-4mm}\fi
\end{figure} 

\textbf{Count Min Sketch (CMS):}
SALSA CMS 
\ifdefined\fullversion
and Tango CMS are
\else
is
\fi 
identical to CMS as long as no counter overflows. We have already defined the merge operation with regard to encoding (Section~\ref{sec:SALSA_encoding}), and with regard to hash mapping in the previous section. However, we still need to define how we determine the value of a merged counter, which provides a degree of freedom we leverage to increase measurement accuracy according to the specific model requirements. A natural merging operation is to \emph{sum} the merged counter values (illustrated in Figure~\ref{fig:sum-merge}). We formalize the correctness of this approach for the Strict Turnstile model via the following theorem.

\ifdefined\sigmodSubmission
Due to lack of space, we defer the proofs of the simpler theorems (\ref{thm:CMS_sum}-\ref{thm:cus}), also marked by (*), to the full paper~\cite{fullVersion}.
\fi

\ifdefined\sigmodSubmission
\ifdefined\fullversion
\begin{theorem}\label{thm:CMS_sum}(*)
 Assume that SALSA and Tango use sum merge to unify counters. Let $2^\ell\cdot s$ be the maximal bit-size of any counter in SALSA CMS, and $\forall i\in[d]$ let $\widetilde{h}_i(x) = \floor{h_i(x)/2^\ell}$ be hash functions that map items into a standard CMS with $(2^\ell\cdot s)$-sized counters. Then for any $x\in U:f_x\le \widehat{f_x^{\text{Tango}}}\le \widehat{f_x^{\text{SALSA}}}\le \widehat{f_x^{\text{CMS}}}$, where $\widehat{f_x^{\text{Tango}}}, \widehat{f_x^{\text{SALSA}}}$ and $\widehat{f_x^{\text{CMS}}}$ are the estimates of Tango, SALSA, {and the underlying CMS (with functions $\widetilde{h}_i(x)$).}
\end{theorem}
\else
\begin{theorem}\label{thm:CMS_sum}(*)
 \rev{Assume that SALSA uses sum merge to unify counters. Let $2^\ell\cdot s$ be the maximal bit-size of any counter in SALSA CMS, and $\forall i\in[d]$ let $\widetilde{h}_i(x) = \floor{h_i(x)/2^\ell}$ be hash functions that map items into a standard CMS with $(2^\ell\cdot s)$-sized counters. Then for any $x\in U:f_x\le \widehat{f_x^{\text{SALSA}}}\le \widehat{f_x^{\text{CMS}}}$, where $\widehat{f_x^{\text{SALSA}}}$ and $\widehat{f_x^{\text{CMS}}}$ are the estimates of SALSA {and the underlying CMS (with functions $\widetilde{h}_i(x)$).}}
\end{theorem}
\fi 
\else
\begin{theorem}\label{thm:CMS_sum}
 Assume that SALSA and Tango use sum merge to unify counters. Let $2^\ell\cdot s$ be the maximal bit-size of any counter in SALSA CMS, and $\forall i\in[d]$ let $\widetilde{h}_i(x) = \floor{h_i(x)/2^\ell}$ be hash functions that map items into a standard CMS with $(2^\ell\cdot s)$-sized counters. Then for any $x\in U:f_x\le \widehat{f_x^{\text{Tango}}}\le \widehat{f_x^{\text{SALSA}}}\le \widehat{f_x^{\text{CMS}}}$, where $\widehat{f_x^{\text{Tango}}}, \widehat{f_x^{\text{SALSA}}}$ and $\widehat{f_x^{\text{CMS}}}$ are the estimates of Tango, SALSA, {and the underlying CMS (with functions $\widetilde{h}_i(x)$).}
\end{theorem}
\begin{proof}
 The sum merge maintains an \textit{invariant} where the value of each merged counter is the total frequency of \textit{all} elements mapped to it. In the worst case, a merge results in a counter of size equal to that of the corresponding counter in the underlying CMS. In this case, the values of the counters are identical. 
  \ifdefined\fullversion
Otherwise, the value of a Tango counter is upper bounded by a SALSA counter which, in turn, is upper bounded by the corresponding \mbox{value in the underlying CMS.\hspace*{-1mm} \qedhere}
\else
\rev{Otherwise, the value of a SALSA counter is upper bounded by the corresponding \mbox{value in the underlying CMS.\hspace*{-1mm} \qedhere}}\fi
\end{proof}
\fi

For Cash Register streams (with only positive updates), \rev{rather than sum the counters when merging, we can take the maximum value of the merged counters to gain more accuracy (exemplified in Figure~\ref{fig:max-merge}) while maintaining guarantees, as formalized in the following theorem.}
\ifdefined\sigmodSubmission
\ifdefined\fullversion
\begin{theorem}\label{thm:CMS_max}(*)
 Assume that SALSA and Tango use max merge to unify counters. Let $2^\ell\cdot s$ be the maximal bit-size of any counter in SALSA CMS, and $\forall i\in[d]$ let $\widetilde{h}_i(x) = \floor{h_i(x)/2^\ell}$ be hash functions that map items into a standard CMS with $(2^\ell\cdot s)$-sized counters. Then for any $x\in U:f_x\le \widehat{f_x^{\text{Tango}}}\le \widehat{f_x^{\text{SALSA}}}\le \widehat{f_x^{\text{CMS}}}$, where $\widehat{f_x^{\text{Tango}}}, \widehat{f_x^{\text{SALSA}}}$ and $\widehat{f_x^{\text{CMS}}}$ are the estimates of Tango, SALSA, {and the underlying CMS (with functions $\widetilde{h}_i(x)$).}
\end{theorem}
\else
\begin{theorem}\label{thm:CMS_max}(*)
 \rev{Assume that SALSA uses max merge to unify counters. Let $2^\ell\cdot s$ be the maximal bit-size of any counter in SALSA CMS, and $\forall i\in[d]$ let $\widetilde{h}_i(x) = \floor{h_i(x)/2^\ell}$ be hash functions that map items into a standard CMS with $(2^\ell\cdot s)$-sized counters. Then for any $x\in U:f_x\le \widehat{f_x^{\text{SALSA}}}\le \widehat{f_x^{\text{CMS}}}$, where $\widehat{f_x^{\text{SALSA}}}$ and $\widehat{f_x^{\text{CMS}}}$ are the estimates of SALSA {and the underlying CMS (with functions $\widetilde{h}_i(x)$).}}
\end{theorem}
\fi 
\else
\begin{theorem}\label{thm:CMS_max}
 Assume that SALSA and Tango use max merge to unify counters. Let $2^\ell\cdot s$ be the maximal bit-size of any counter in SALSA CMS, and $\forall i\in[d]$ let $\widetilde{h}_i(x) = \floor{h_i(x)/2^\ell}$ be hash functions that map items into a standard CMS with $(2^\ell\cdot s)$-sized counters. Then for any $x\in U:f_x\le \widehat{f_x^{\text{Tango}}}\le \widehat{f_x^{\text{SALSA}}}\le \widehat{f_x^{\text{CMS}}}$, where $\widehat{f_x^{\text{Tango}}}, \widehat{f_x^{\text{SALSA}}}$ and $\widehat{f_x^{\text{CMS}}}$ are the estimates of Tango, SALSA, {and the underlying CMS (with functions $\widetilde{h}_i(x)$).}
\end{theorem}

\begin{proof}
After each merge, the counter value upper bounds the frequency of any element mapped to the hash range of the merged counter. In addition, the value of SALSA and Tango counters when using the max merge are upper bounded by the corresponding value of SALSA and Tango counters when using the sum merge.
\end{proof}
\fi


\ifdefined\fullversion
Theorems~\ref{thm:CMS_sum} and \ref{thm:CMS_max} show that SALSA CMS and Tango CMS are at least as accurate as the underlying CMS for both merge operations. 
\else
\rev{Theorems~\ref{thm:CMS_sum} and \ref{thm:CMS_max} show that SALSA CMS is at least as accurate as the underlying CMS for both merge operations.}
\fi 
Intuitively, by sum-merging every consecutive $\overline n$ bits, we obtain estimates that are identical to a CMS sketch that uses $\overline n$ bit counters. 
Therefore, sum-merging SALSA's estimates are upper bounded by the CMS estimates. 
In Cash Register streams, max-merging estimates are upper bounded by the sum-merging ones. 
\ifdefined\fullversion
Finally, for any given element, the estimates of SALSA CMS and Tango CMS are lower bounded by its true frequency, which implies that our approach provides the same error guarantee as the underlying sketch. 
\else
\rev{Finally, for any given element, the estimates of SALSA CMS are lower bounded by its true frequency, which implies that our approach provides the same error guarantee as the underlying sketch.}
\fi 

\ifdefined\fullversion
SALSA CMS also improves the performance of count distinct queries for Linear Counting~\cite{LC} using CMS.\ran{Remove anything else about linear counting and add a ref to the full version}  Recall Linear Counting estimates the number of distinct queries using the fraction $p$ of zero counters.
We consider running Linear Counting using SALSA CMS staring with $s=8$ bit counters, compared to a standard CMS implementation using 32-bit counters.  Unlike standard CMS, SALSA may be unable to determine the exact number of ($s$-bit) counters that remain zero, as some are merged into other counters. Instead, we compute the fraction $\mathfrak f$ of $s$-bit counters that remained zero \emph{from the overall number of counters that did not merge}.  For every counter that is the result of one or more merges, we know that at least one of its sub-counters is not zero;  we optimistically assume that a fraction $\mathfrak f$ of its remaining sub-counters are zero.  So, for example, our estimate of the number of counters that are 0 is the number of $s$-bit counters that remained zero, plus $\mathfrak f$ times the number of $2s$-bit counters, plus $3\mathfrak f$ times the number of $4s$-bit counters, and so on if there are larger counters.  Note that this approach is heuristic and its accuracy guarantees are left as future work.
\else
\fi
\textbf{Conservative Update Sketch (CUS):}
SALSA CUS is similar to the standard CUS -- whenever an update $\angles{x,v}$ arrives, each counter $C[i,h_i(x)]$ is set to $\max\set{C[i,h_i(x)],v+\widehat{f_x}}$ with  $\widehat{f_x}=\min_{i\in[d]} C[i,h_i(x)]$ being the previous frequency estimate for $x$. 
Unlike the CMS variant, the correctness of SALSA CUS is not immediate as not all counters are increased for each packet. Theorem~\ref{thm:cus} shows that SALSA CUS is correct in the Cash Register model when working with the max-merge method. 
     
\ifdefined\sigmodSubmission
\begin{theorem}
\label{thm:cus}(*)
 Let $2^\ell\cdot s$ be the maximal bit-size of any counter in \emph{max-merge} SALSA CUS, and $\forall i\in[d]$ let $\widetilde{h}_i(x) = \floor{h_i(x)/2^\ell}$ be hash functions that map items into a standard CUS with $(2^\ell\cdot s)$-sized counters. Then for any $x\in U:f_x\le \widehat{f_x^{\text{SALSA}}}\le \widehat{f_x^{\text{CUS}}}$, where $\widehat{f_x^{\text{SALSA}}}$ and $\widehat{f_x^{\text{CUS}}}$ are the estimates of SALSA and the underlying CUS (with functions $\widetilde{h}_i(x)$).
\end{theorem}
\else
\begin{theorem}
\label{thm:cus}
 Let $2^\ell\cdot s$ be the maximal bit-size of any counter in \emph{max-merge} SALSA CUS, and $\forall i\in[d]$ let $\widetilde{h}_i(x) = \floor{h_i(x)/2^\ell}$ be hash functions that map items into a standard CUS with $(2^\ell\cdot s)$-sized counters. Then for any $x\in U:f_x\le \widehat{f_x^{\text{SALSA}}}\le \widehat{f_x^{\text{CUS}}}$, where $\widehat{f_x^{\text{SALSA}}}$ and $\widehat{f_x^{\text{CUS}}}$ are the estimates of SALSA and the underlying CUS (with functions $\widetilde{h}_i(x)$).
\end{theorem}

\begin{proof}
It is sufficient to consider only updates with $v=1$ since each $\angles{x,v}$ update is identical to $v$ consecutive $\angles{x,1}$ updates.
The proof is by induction on the number of updates. Specifically, we show that after each update it holds that
\begin{equation}\label{eq:cuSALSA_proof}
\forall x, i \in [d]: C_{SALSA}[i,h_i(x)] \le C_{CUS}[i,\tilde{h}_i(x)] \ ,    
\end{equation}
where we denote by $C_{SALSA}$ and $C_{CUS}$ the counters of SALSA and the underlying CUS, respectively.

As a base case, initially $C_{SALSA}[i,h_i(x)] = C_{CUS}[i,\tilde{h}_i(x)]=0 \,\, \forall i \in [d]$. 
We show that if Equation \eqref{eq:cuSALSA_proof} holds, it continues to hold after an additional update.   

Case 1: $C_{SALSA}[i,h_i(x)] = C_{CUS}[i,\tilde{h}_i(x)]$. 
In this case, on update $\angles{x,1}$, $C_{CUS}[i,\tilde{h}_i(x)]$ is increased by CUS. Therefore the claim trivially holds if there is no overflow in SALSA. If there is an overflow, the claim holds by the virtue of the max-merge. That is, the value of the merged counter grows by exactly 1. This also means that the inequality holds for all counters involved in this merge since they are all upper bounded by $C_{CUS}[i,h_i(x)]$ prior to the update.  

Case 2: $C_{SALSA}[i,h_i(x)] < C_{CUS}[i,\tilde{h}_i(x)]$. In this case, on update $\angles{x,1}$, the claim trivially holds if there is no overflow in SALSA. If there is an overflow, by the virtue of the max-merge, the value of the merged counter still only grows by 1. This also means that the inequality holds for all counters involved in this merge since they are all upper bounded by $C_{CUS}[i,h_i(x)]$ prior to the update, and therefore are upper bounded by  $C_{CUS}[i,\tilde{h}_i(x)]$ after it.  \qedhere

\end{proof}
\fi

\textbf{Count Sketch (CS):}
SALSA can also extend the CS, with a minor modification. Unlike most existing implementations, which use the standard Two's Complement encoding, SALSA CS uses a \emph{sign-magnitude} representation of counters (as counters can be negative), with the most significant bit for the sign and the rest as magnitude.  
While Two's Complement represents values in the range $\set{-2^{s-1},\ldots,2^{s-1}-1}$, sign-magnitude does not allow a representation of $-2^{s-1}$.  However, our use of sign-magnitude
is critical for us to ensure that the overflow event is sign-symmetric, which allow us to prove that our sketch is unbiased. 
When an $s\cdot 2^\ell$ bits counter exceeds an absolute value of $2^{s\cdot 2^\ell - 1}- 1$, we merge to double its size. 
\rev{
When merging counters in SALSA CS, we use sum-merge;  note max-merge may not be correct as counters may have opposite signs.
}
We prove the correctness of SALSA 
CS.
For simplicity, we focus on the main variant. That is, a counter merges \emph{at most twice}, starting from $s=8$ bits and assuming that no counter reaches an absolute value of $2^{31}$, which is the common implementation assumption.





\ifdefined\sigmodSubmission\vspace*{-7mm}\fi
 
\ \\
\begin{wrapfigure}{r}[0.cm]{0.07\textwidth}
    \begin{center}
        \ifdefined\sigmodSubmission\vspace*{-3mm}\fi
        \hspace*{-4mm}
        \includegraphics[width=0.09\textwidth]{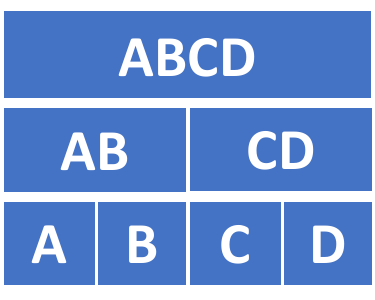}
    \end{center}
    \ifdefined\sigmodSubmission\vspace*{-4mm}\fi
\end{wrapfigure} \mbox{             } \mbox{    } \rev{
Let $x\in U$ be an element mapped to counter $A$, which may be merged with counter $B$ to create $\angles{A,B}$, which in turn may later merge with $\angles{C,D}$ to make the $4s$-bit counter $\angles{A,B,C,D}$. This setting is illustrated to the right:
}


\rev{
  We wish to show that the estimates of each row in SALSA CS are unbiased, and further that each estimate has variance with SALSA CS that is no larger than the corresponding variance with CS. As mentioned, here give a full analysis for starting with $s$ bit counters and allowing counters to grow to $4s$ bits, as this is the focus in our implementation, but the approach generalizes readily to additional levels.  we now introduce some notation to analyze SALSA CS.
}

We use $O_{AB}$ to denote the event that $A$ and $B$ \rev{have been} merged at query time (either into $\angles{A,B,C,D}$ or just as $\angles{A,B}$) and $O_{ABCD}$ for the event that $A, B, C,$ and $D$ \rev{have been} merged into $\angles{A,B,C,D}$. We also denote the value of $A$ by (the random variable) $X_A$, the value of $\angles{A,B}$ by $X_{AB}$, and similarly for $X_B$, $X_{CD}$, and $X_{ABCD}$. We emphasize that $X_S$ represents the value of the count mapped to $S$, regardless of whether the counter overflows (e.g., $X_A\triangleq\sum_{y\in U: h(y)=A}f_yg(y)$ even if $X_A\ge 2^s$).
Without loss of generality, we also assume that the sign of $x$ is $g(x)=1$ (and thus $\mathbb E[X_A]=\mathbb E[X_{AB}]=\mathbb E[X_{ABCD}]=f_x$).
This allows us to express the estimate given \mbox{in a row for SALSA CS as:}
$    \widehat{f_x}= X_A (1-O_{AB}) + X_{AB} \cdot O_{AB} \cdot (1-O_{ABCD}) + X_{ABCD}\cdot O_{ABCD} 
    = X_A (1-O_{AB}) + X_{AB} \cdot O_{AB} - X_{AB} \cdot O_{AB}\cdot O_{ABCD} + X_{ABCD}\cdot O_{ABCD}.
$
 Observe that $O_{ABCD}\subseteq O_{AB}$ and thus $O_{AB}\cdot O_{ABCD}=O_{ABCD}$.  \rev{This implies} that (since $X_{AB}-X_A=X_B$ and $X_{ABCD}-X_{AB}=X_{CD}$):
\begin{align}
    \widehat{f_x}
    = X_A + X_{B} \cdot O_{AB} + X_{CD}\cdot O_{ABCD}
    .\label{eq:CS_est}
\end{align}
We continue by proving that the estimate is unbiased.
\begin{lemma}
SALSA is unbiased, i.e.,
$\mathbb E[\widehat{f_x}] = f_x.$
\end{lemma}
\begin{proof}
Due to the sign-symmetry of the sign function $g$, we have that\footnote{This assumes that the variables remain symmetric conditioned on the overflow events, which is correct for SALSA CS due to \mbox{our sign-magnitude representation.}} $\mathbb E[X_B|O_{AB}]=0$ and $\mathbb E[X_{CD}|O_{ABCD}]=0$ (as $\forall\mathfrak i: \Pr[X_B=\mathfrak i|O_{AB}]=\Pr[X_B=-\mathfrak i|O_{AB}]$ and $\Pr[X_{CD}=\mathfrak i|O_{ABCD}]=\Pr[X_{CD}=-\mathfrak i|O_{ABCD}]$).
Thus, according to~\eqref{eq:CS_est}:
{\small
\begin{multline*}
\mathbb E[\widehat{f_x}]
    = \mathbb E\brackets{X_A + X_{B} O_{AB} + X_{CD}O_{ABCD}}
    = \mathbb E\brackets{X_A} \\+ \mathbb E\brackets{X_{B}| O_{AB}}  \Pr[O_{AB}] + \mathbb E\brackets{X_{CD} | O_{ABCD}} \Pr[O_{ABCD}] = f_x.\qedhere
\end{multline*}
}
\end{proof}
\mbox{We show that SALSA reduces the variance in each row.}
\begin{lemma}\label{lem:CS_var}
$\Var[\widehat{f_x}] \le \Var[CS],$ where $\Var[CS]\triangleq \Var[X_{ABCD}]$ is the variance of the underlying Count Sketch.
\end{lemma}
\begin{proof}
Let us prove that $\Var[CS]-\Var[\widehat{f_x}]\ge 0$.
Observe that since CS and SALSA CS are unbiased, we have that:
\begin{align*}
    \Var[&CS]-\Var[\widehat{f_x}] \\&=\mathbb E\brackets{\parentheses{X_{ABCD}-f_x}^2} - \mathbb E\brackets{\parentheses{\widehat{f_x}-f_x}^2}\\
    &= \mathbb E\brackets{{X_{ABCD}^2-\widehat{f_x}^2-2f_xX_{ABCD}+2f_x\widehat{f_x}}}\\
    &=
    \mathbb E\brackets{X_{ABCD}^2}-\mathbb E\brackets{\widehat{f_x}^2}+2f_x\mathbb E\brackets{\widehat{f_x}-X_{ABCD}}.
\end{align*}
Due to unbiasedness, $\mathbb E\brackets{\widehat{f_x}}=\mathbb E\brackets{X_{ABCD}}=f_x$ and thus
\begin{align}
    \Var[CS]-\Var[\widehat{f_x}]=\mathbb E\brackets{X_{ABCD}^2}-\mathbb E\brackets{\widehat{f_x}^2}.\label{eq:CSeq1}
\end{align}
We continue by simplifying the expression for $\mathbb E\brackets{\widehat{f_x}^2}$:
{\small
\begin{align*}
    \mathbb E&\brackets{\widehat{f_x}^2} = \mathbb E\brackets{\parentheses{X_A+X_BO_{AB}+X_{CD}O_{ABCD}}^2}\\
     &= \mathbb E\brackets{X_A^2} + \mathbb E\brackets{X_B^2O_{AB}}+\mathbb E\brackets{X_{CD}^2O_{ABCD}}\\
     &+2\Bigg(\mathbb E\brackets{X_{A}X_BO_{AB}}+\mathbb E\brackets{X_AX_{CD}O_{ABCD}}\\&\qquad+\mathbb E\brackets{X_BX_{CD}O_{ABCD}}\Bigg)\\
    &= \mathbb E\brackets{X_A^2} + \mathbb E\brackets{X_B^2|O_{AB}}\Pr[O_{AB}]\\&+\mathbb E\brackets{X_{CD}^2|O_{ABCD}}\Pr[O_{ABCD}]+2\Big(\mathbb E\brackets{X_{A}X_B|O_{AB}}\Pr[O_{AB}]\\&+\mathbb E\brackets{X_AX_{CD}|O_{ABCD}}\Pr[O_{ABCD}]\\&+\mathbb E\brackets{X_BX_{CD}|O_{ABCD}}\Pr[O_{ABCD}]\Big).     
\end{align*}
}
Since $x$ is mapped to $A$, we can use the sign-symmetry of $X_B$ and $X_{CD}$ to get\footnotemark[1]  $\mathbb E\brackets{X_{A}X_B|O_{AB}}=\mathbb E\brackets{X_AX_{CD}|O_{ABCD}}=\mathbb E\brackets{X_BX_{CD}|O_{ABCD}}=0$, which gives
{\small
\begin{multline}
\mathbb E\brackets{\widehat{f_x}^2} = \mathbb E\brackets{X_A^2} + \mathbb E\brackets{X_B^2|O_{AB}}\Pr[O_{AB}]\\+\mathbb E\brackets{X_{CD}^2|O_{ABCD}}\Pr[O_{ABCD}]
= \mathbb E\brackets{X_A^2} + \mathbb E\brackets{X_B^2}+\mathbb E\brackets{X_{CD}^2} \\- 
\parentheses{\mathbb E\brackets{X_B^2|\neg O_{AB}}\Pr[\neg O_{AB}]+\mathbb E\brackets{X_{CD}^2|\neg O_{ABCD}}\Pr[\neg O_{ABCD}]}\\\le \mathbb E\brackets{X_A^2} + \mathbb E\brackets{X_B^2}+\mathbb E\brackets{X_{CD}^2}.\label{eq:CSeq2}
\end{multline}
}
Now, notice that $X_{ABCD}^2 = \parentheses{X_A+X_B+X_{CD}}^2= X_A^2+X_B^2+X_{CD}^2+2(X_AX_B+X_AX_{CD}+X_BX_{CD})$. Due to the sign-symmetry of $g$, we have $\mathbb E\brackets{X_AX_B}=\mathbb E\brackets{X_AX_{CD}}=\mathbb E\brackets{X_BX_{CD}}=0$ and thus:
$
\mathbb E\brackets{X_{ABCD}^2} = \mathbb E\brackets{X_A^2}+\mathbb E\brackets{X_B^2}+\mathbb E\brackets{X_{CD}^2}\ge \mathbb E\brackets{\widehat{f_x}^2},
$
where the inequality follows from~\eqref{eq:CSeq2}. Together with~\eqref{eq:CSeq1}, \mbox{this concludes the proof. \qedhere}
\end{proof}

Because the theorem shows the error variance is no larger \emph{for each row}, following the same analysis as for CS (using Chebyshev's inequality to bound the error of the row and then Chernoff's inequality to bound the error of the median) yields the same error bounds for SALSA CS.  Indeed, we expect better estimates using SALSA as the inequality from the proof of the theorem ($\Var[CS] \geq \Var[\widehat{f_x}]$) is usually a strict inequality.  In our experimental evaluation, we show that SALSA CS obtains better estimates than CS. 

\ifdefined\sigmodSubmission\vspace*{-1mm}\fi
\begin{theorem}
Let $2^\ell\cdot s$ be the maximal bit-size of any counter in \emph{sum-merge} SALSA CS, and $\forall i\in[d]$ let $\widetilde{h}_i(x) = \floor{h_i(x)/2^\ell}$ be hash functions that map items into a standard CS with $(2^\ell\cdot s)$-sized counters. Then for any $x\in U, i\le d:
\mathbb E \brackets{C_{SALSA}[i,h_i(x)]\cdot g_i(x)} = f_x$ and
$
\Var\brackets{C_{SALSA}[i,h_i(x)]\cdot g_i(x)-f_x}\le  \Var\brackets{C_{CS}[i,\tilde{h}_i(x)]\cdot g_i(x)-f_x}$, where $C_{SALSA}[i,h_i(x)]$ and $C_{CS}[i,\tilde{h}_i(x)]$ are the counters of SALSA and the underlying CS (with functions $\widetilde{h}_i(x)$).
\end{theorem}
\ifdefined\sigmodSubmission\vspace*{-1mm}\fi

\rev{
We note that SALSA CS can also provide other derived results, similarly to CS.} For example, by using a heap, we can find the $L_2$-heavy hitters in Cash Register \mbox{streams similarly to the original version. }


\textbf{Universal Sketch (UnivMon):}
The universal monitoring sketch (UnivMon) uses several $L_2$ sketches that are applied on different subsets of the universe. By improving the accuracy of CS, we can also improve the performance of UnivMon.
\rev{We note that since SALSA CS provides an accuracy guarantee that is at least as good as the underlying sketch, SALSA Univmon provides the same accuracy guarantee as the vanilla Univmon.}

\textbf{Merging and Subtracting SALSA Sketches:}\label{sec:merging}
\rev{Given streams $A,B$ and their sketches $s(A),s(B)$,
we may then wish to derive statistics on $A\cup B$ (for example, we can parallelize the sketching of $A$ and $B$ and then \emph{merge} them), or on $A\setminus B$ (for example, to detect changes in our network traffic compared to the previous epoch).
By $A\setminus B$, we refer to computing the frequency difference; e.g., if $x$ appeared twice in $A$ and three times in $B$, its frequency in $A\setminus B$ is -1. }
Most standard sketches are linear, and can be naturally summed/subtracted counter-wise to obtain a sketches $s(A\cup B)\equiv s(A) + s(B)$ and $s(A\setminus B)\equiv s(A)-s(B)$ if they share the \mbox{same hash functions, and work in the Turnstile model.} 

SALSA can also merge and subtract sketches. For merging $s(A)$ and $s(B)$, SALSA traverses the counters and merges them according sum-merging.
Specifically, each counter in the merged sketches has a size at least as large as its size in $s(A)$ and its size in $s(B)$. Additionally, when summing or subtracting counters an overflow may occur, triggering another merge to make sure we have enough bits to encode the resulting values.
CS, as a Turnstile sketch, also supports general subtracting that is done similarly to merging, while CMS (which works in the Strict Turnstile model) can compute $s(A\setminus B)$ \rev{given a  guarantee} that $B\subseteq A$. These operations are illustrated in Figure~\ref{fig:merge_subtract}. 

\begin{figure}[t]
\centering
\includegraphics[width = 1\columnwidth]
		{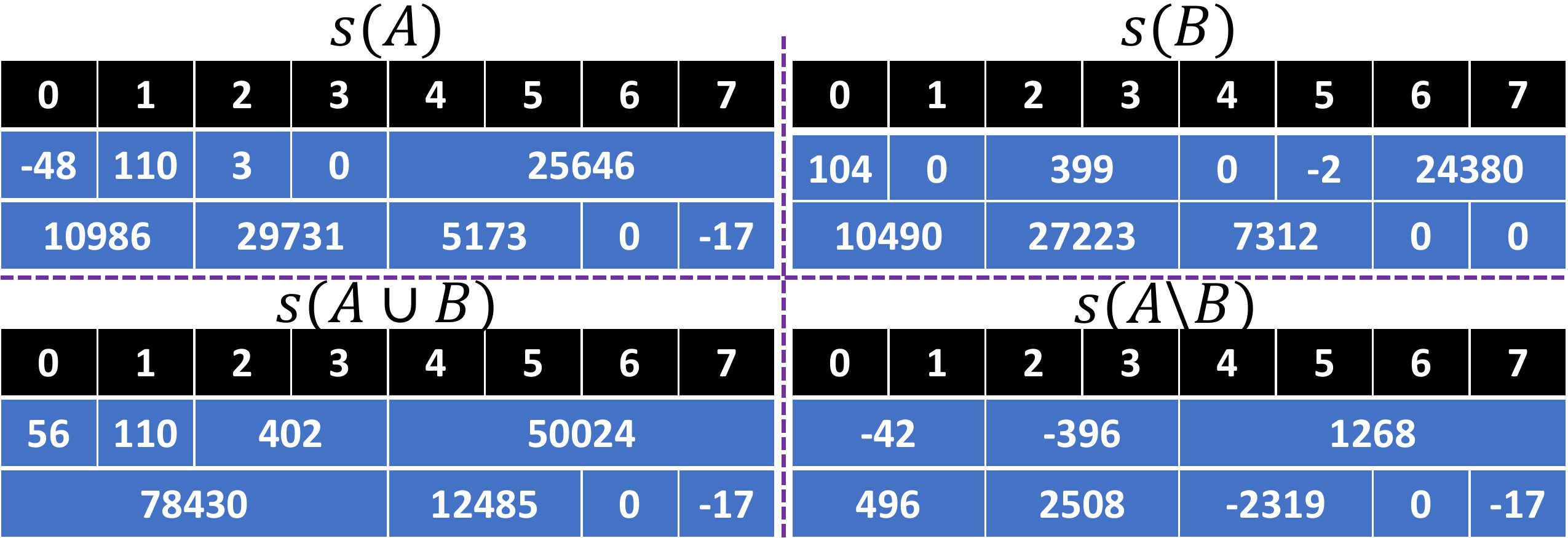}
	\ifdefined\sigmodSubmission\vskip -0.15cm\fi
	\ifdefined\submissionVersion
	\ifdefined\sigmodSubmission\vskip -0.1cm\fi
	\fi
	\vspace*{-0mm}
    \caption{\label{fig:merge_subtract} An example of $s=8$ bit SALSA CS merging and subtracting Turnstile sketches $s(A)$ \mbox{and $s(B)$.}
    }
    \ifdefined\sigmodSubmission\vspace*{-3mm}\fi
\end{figure}

\textbf{Integrating Estimators into SALSA: }\label{sec:estimators}
\rev{Thus far, we have described a single strategy to handle overflows: when a counter reaches a value that can not be represented with the current number of bits, it merges with a neighbor.}
However, there are alternatives that allow one to increase the counting range. 
Specifically, \emph{estimators} can represent large numbers using a smaller number of bits, \mbox{at the cost of introducing an error.}

The state of the art Additive Error Estimators (AEE)~\cite{Compsketch} offer a simple and efficient technique to increase the counting range. 
For simplicity, we describe the technique for CMS and unit-weight streams (where all updates are of the form $\angles{x,1}$), although AEE can support weighted updates and other L1 sketches as well.
Throughout the execution, incoming updates are \emph{sampled} with probability $p$. If an update is sampled, it increases the sketch, and otherwise, it is ignored.
Whenever a counter overflows, a \emph{downsampling} event happens. When downsampling, $p$ is halved and any counter $C[i,j]$ is replaced by either $Bin(C[i,j],1/2)$ (called probabilistic downsampling) or by  $\floor{C[i,j]/2}$ (deterministic downsampling). 
Since the counter values are reduced as a result of the downsampling, new updates can be processed, and no additional counter bits are needed.  
\rev{For any $\delta_{est}>0$, we have an implied estimation error for AEE given by
$\epsilon_{est}\triangleq\sqrt{\frac{2p^{-1}\ln(2/\delta_{est})}{N}}$, such that
{\small
\ifdefined\sigmodSubmission\vspace*{-1mm}\fi
$\Pr\brackets{\Big|\widehat{C[i,j]}-C[i,j]\Big|\ge N\epsilon_{est}}= 
\Pr\Big[|\widehat{C[i,j]}-C[i,j]|\ge \sqrt{2Np^{-1}\ln(2/\delta_{est})}\Big]\le\delta_{est} . 
$}
}
Another motivation for AEE comes from the processing speed. Since the sampling probability is independent of the value of the current counter, one can compute the hash functions $h_i(x)$ \emph{only if a packet is sampled}. Since hash functions are a major bottleneck for sketches~\cite{Nitro}, AEE is faster \mbox{than the baseline sketches.}
Another version of the estimator, called AEE MaxSpeed, aims to maximize the processing speed while bounding the error. Therefore, instead of waiting for a counter to overflow, it downsamples all counters once enough updates have been processed. In comparison with the original variant (called AEE MaxAccuracy), MaxSpeed is faster \mbox{but less accurate~\cite{Compsketch}.}

Intuitively, downsampling and merging increase the error in different ways. While downsampling increases the inherent error of a counter, merging adds noise from other elements that previously have not collided with the counter. 
SALSA selects how to handle overflows in a way that minimizes the theoretical error increase by either downsampling or merging. 
Specifically, as our accuracy theorems suggest, the sketch error in SALSA depends on the size of the largest counter. Therefore, unless a largest counter overflows, SALSA opts for merging as its overflow strategy.
When a largest counter overflows, SALSA computes the estimator error difference $\Delta_{\mathit{est}}=\sqrt 2 \cdot \epsilon_{est}$, which is the increase in error if we downsample.
Similarly, if the currently largest counter is of size $s\cdot 2^\ell$, SALSA computes $\eps_{CMS}\triangleq \delta^{-1/d}\cdot 2^\ell/w$, which is the current accuracy guarantee (see Theorem~\ref{thm:CMS_sum} and Section~\ref{sec:CMS}) and $\Delta_{\mathit{CMS}}=\eps_{CMS}$ is then the difference in error guarantee that results from merging. We pick $\delta_{est}=\delta/d$ to allow \emph{all} counters of the current element to be estimated within $\epsilon_{est}$ with probability $1-\delta$.
Finally, SALSA chooses to merge if $\Delta_{\mathit{CMS}}\le \Delta_{\mathit{est}}$, and otherwise it downsamples.
As a result, SALSA estimates element sizes to within $N\cdot(\epsilon_{est}+\epsilon_{CMS})$ with a probability of at least $1-2\delta$. 

As an optimization, when downsampling, SALSA may be able to \emph{split} counters if the resulting values can be represented using fewer bits. For example, if $s=8$ and a value of $300$ was represented in the 16-bit counter $\angles{4,5}$, and if it is then downsampled to a value of $150$, we can split the counter and set both counter $4$ and counter $5$ to $150$. We note that this only works for max-merging, where the accuracy guarantees~seamlessly follow.   




\section{Evaluation}\label{sec:eval}
In this section, we extensively evaluate SALSA's \rev{performance} on real and synthetic datasets and compare it to that of the underlying sketches. 
We first document the methodology.

\textbf{Sketch Configuration Parameters:}
Unless specified otherwise, all CMS and CUS sketches are configured with $d=4$ rows, as is used e.g., in the Caffeine caching library~\cite{caffein}. Since CS requires taking median over the rows, all CS experiments are configured with $d=5$ rows as done, e.g., in~\cite{L2IQ}.
We configure UnivMon with $16$ CS instances, each configured with $d=5$ and a heap of size $100$, following the implementation of~\cite{univmon}.
Such settings are standard for applications that aim for speed rather than being memory-optimal.
%
For the ABC~\cite{gong2017abc} and Pyramid~\cite{PyramidSketch} sketches, as well as the Cold Filter~\cite{ColdFilter} framework, we use the configurations recommended by the authors.
We pick $s=8$ bit counters as the default configuration of SALSA, motivated by the synthetic results. 
We use the simple encoding (1 bit of overhead per counter) of Section~\ref{sec:SALSA_encoding}, which uses slightly more space but is faster.  
The Baseline implementations use 32-bit counters, \rev{a choice we justify later in Figure~\ref{fig:smallCounters}, and that is also} common \mbox{in existing implementations~\cite{Nitro,CormodeCode}.}
Nonetheless, our SALSA implementation allows counters to grow further, up to $64$ bits.
For implementation efficiency, all row widths $w$ are powers of two.
When we give figures where an $x$-axis is allocated memory, we include the encoding overheads. 
For the integration with AEE, we configure SALSA AEE with $\delta=4\cdot\delta_{est}=0.001$ (see Section~\ref{sec:estimators}).

\textbf{Datasets: }
\rev{ We evaluate our algorithms using four real datasets and several synthetic ones. In particular, we use three network packet traces: two from major backbone routers in the US, denoted NY18~\cite{CAIDA2018} and CH16~\cite{CAIDA2016}, and a data center network trace denoted Univ2~\cite{UWISC}.
In these traces, we define items using the ``5-tuples'' of the packets (srcip, dstip, srcport, dstport, proto).
Additionally, we use a YouTube video trace~\cite[US category]{kaggleYouTubeDataset}. As the video data does not have a recorded order (just view-count), we use a random order where each item is a video independently sampled according to the view-count distribution. Finally, we use random order Zipfian traces. All traces have 98M elements for consistency with the shortest real dataset.
}
%
In our evaluation, we use unit-weight Cash Register streams (i.e., all updates are of the form $\angles{x,1}$). \rev{We also experiment with the task of evaluating change detection, which requires \mbox{a SALSA sketch under the Turnstile model.}
}

\textbf{Metrics: }
For frequency estimates, we use the On-arrival model that asks for an estimate of the size of \emph{each arriving element} (e.g.,~\cite{Compsketch,ConextPaper,HeavyHitters,RAP}). Intuitively, this model is motivated by the need to take per-packet actions in networking, e.g., to restrict the allowed bandwidth to prevent denial of service attacks.
Given a stream with $n$ updates, we obtain errors $e_1,e_2,\ldots,e_n$; the \emph{Mean Square Error} is defined as $MSE\triangleq n^{-1}\cdot\sum_i e_i^2$, the \emph{Root Mean Square Error} is then $RMSE\triangleq\sqrt{MSE}$, while the \emph{Normalized RMSE} is $\mathit{NRMSE}\triangleq n^{-1}\cdot {RMSE}$. Similar metrics are used, e.g., in~\cite{Compsketch,HeavyHitters,RAP,ConextPaper}.
Notice that NRMSE is a \mbox{unitless quantity in the interval $[0,1]$.}
For fairness, we also evaluate using the error metrics used in Pyramid and ABC: Average Absolute Error (AAE) and Average Relative Error (ARE).
AAE averages the error over all the elements with non-zero frequency i.e., $AAE\triangleq\frac{1}{|U_{> 0}|}\sum_{x\in U_{> 0}}|\widehat{f_x}-f_x|$, where $U_{> 0}\triangleq \set{x\in U:f_x>0}$. \mbox{Similarly, ARE is defined as $\frac{1}{|U_{> 0}|}\sum_{x\in U_{> 0}}\frac{|\widehat{f_x}-f_x|}{f_x}$.}

\ifdefined\fullversion
For tasks such as Count Distinct, Entropy, and Frequency Moments estimation,
\else
\rev{For tasks such as Entropy and Frequency Moments estimation,}
\fi
we use the Average Relative Error (ARE) metric that averages over the relative error of the ten runs.\footnote{This is different from the ARE used for frequency estimation where the averaging is done over all elements with positive frequency.}
For turnstile evaluation, we evaluate the capability of SALSA to improve sketches for the Change Detection task (e.g., see~\cite{univmon,krishnamurthy2003sketch}) in which we partition the workload into two equal-length parts $A$ and $B$, sketch each, and test the NRMSE of the estimates of the frequency changes between $A$ and $B$. 
Each data point is the result of ten trials; we report the mean and 95\% confidence \mbox{intervals according to Student's t-test~\cite{student1908probable}.}

\textbf{Implementation:}
We leverage existing CMS and CUS implementations from~\cite{Compsketch} and extend them to implement SALSA. We also extend these to create a Baseline and SALSA implementation of CS. 
We also used the authors' code for the Pyramid~\cite{PyramidSketch}, ABC~\cite{gong2017abc}, and Cold Filter~\cite{ColdFilter} algorithms. Particularly, for error measurements, we used the code as-is, while for speed measurements, we applied our optimizations for a fair comparison. 
All sketches use the same hash functions (BobHash) and index computation methods.
When evaluating against the AEE estimators~\cite{Compsketch}, we use the provided open-source code.
Similarly, we obtained the UnivMon code from~\cite{L2IQ} and replaced its CS sketches with SALSA CS to create 'SALSA UnivMon'. 

All speed measurements were performed using a single core on a PC with an Intel Core i7-7700 CPU @3.60GHz (256KB L1 cache, 1MB L2 cache, and 8MB L3 cache) and 32GB DDR3 2133MHz RAM.




\begin{figure}[t]
    \centering
    \subfloat[Error, Count Min Sketch (2MB)]
    { \includegraphics[width =0.5\columnwidth]
    {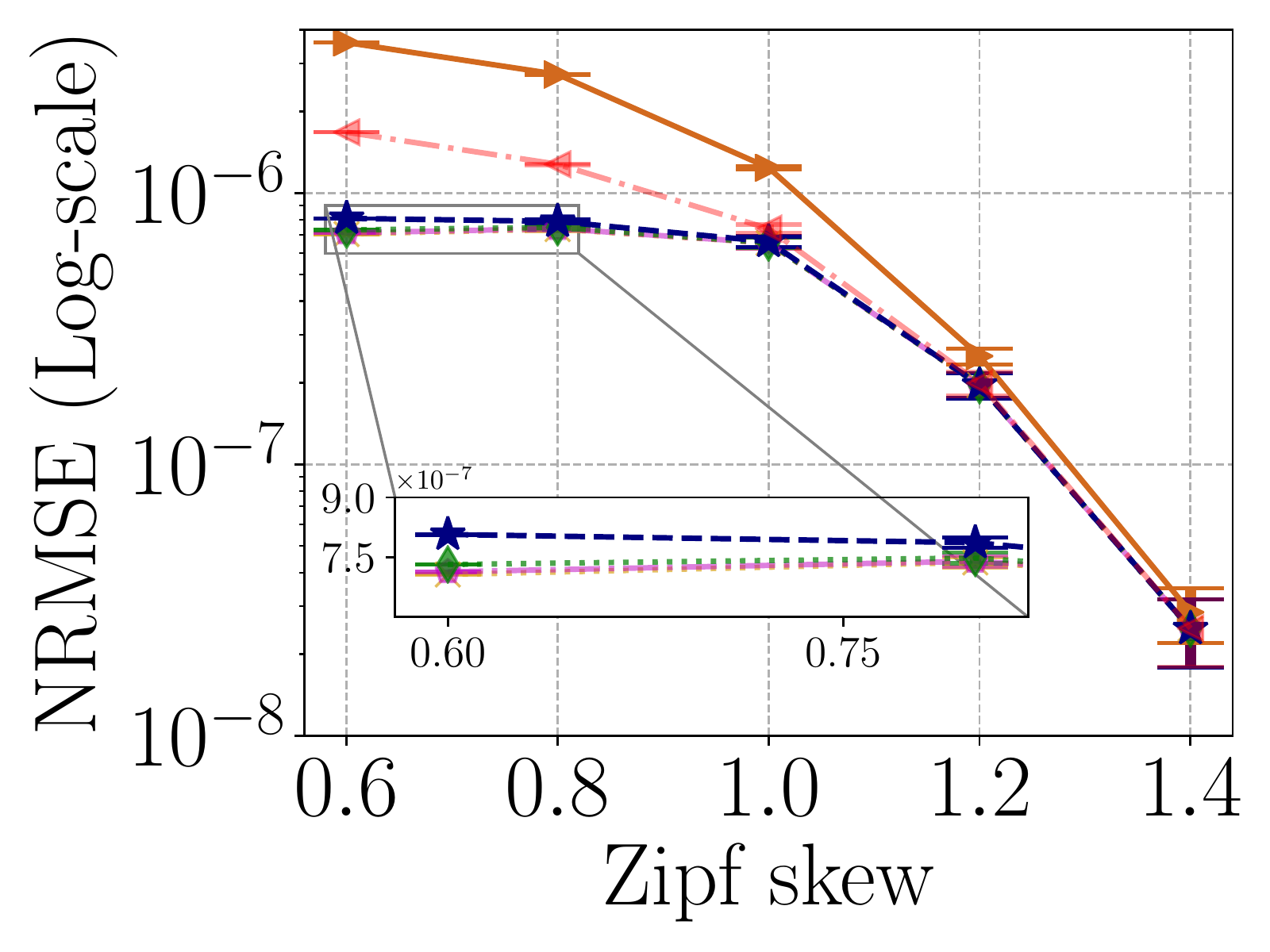}}
    \subfloat[Error, Count Sketch (2.5MB)]
    {\includegraphics[width =0.5\columnwidth]
    {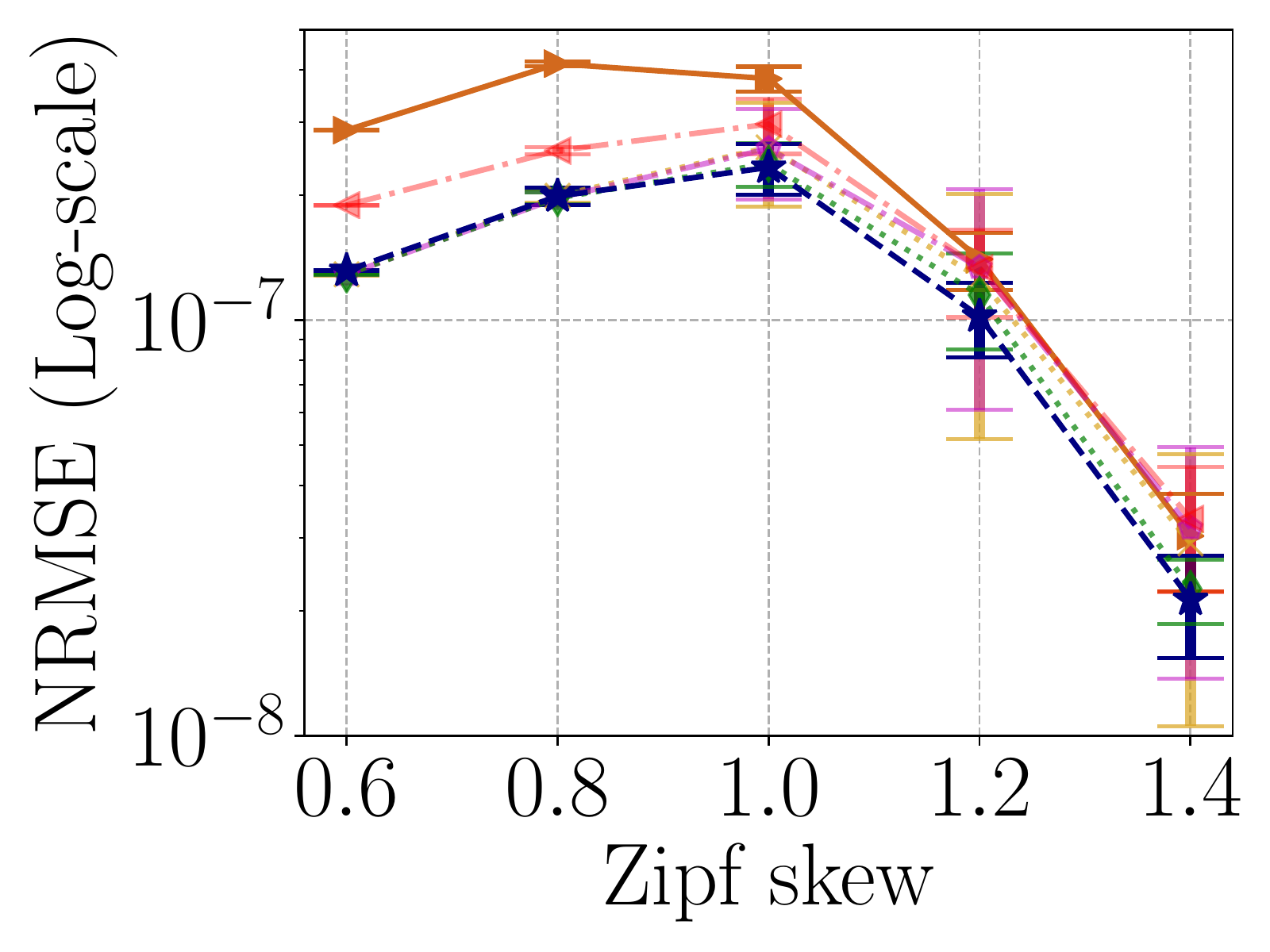}}\\
    {\includegraphics[width =1.00\columnwidth]
    {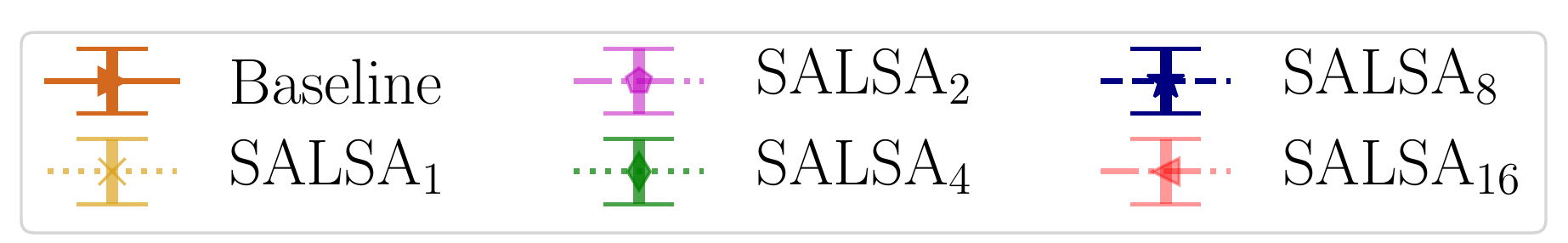}\ifdefined\sigmodSubmission\vspace*{-3mm}\fi}
    \ifdefined\sigmodSubmission\vspace*{-2mm}\fi
    \caption{\small Speed and accuracy of SALSA CMS and SALSA CS for the synthetic datasets. The Baseline uses $w=2^{17}$ counters in each row for a total of 2MB of space in CMS and 2.5MB in CS. Here, SALSA{\Large$_s$} is using $w=(2^{17}\cdot 32/s)$ sized rows for a total of $2(1+1/s)$MB space for CMS and $2.5(1+1/s)$MB for CS.
    }\label{fig:Zipf}
    \ifdefined\sigmodSubmission\vspace*{-3mm}\fi
\end{figure} 

\begin{figure}[t]
    \centering
    \ifdefined\sigmodSubmission\vspace*{-5mm}\fi
    \subfloat[Error, NY18]
    {\includegraphics[width =0.5\columnwidth]
    {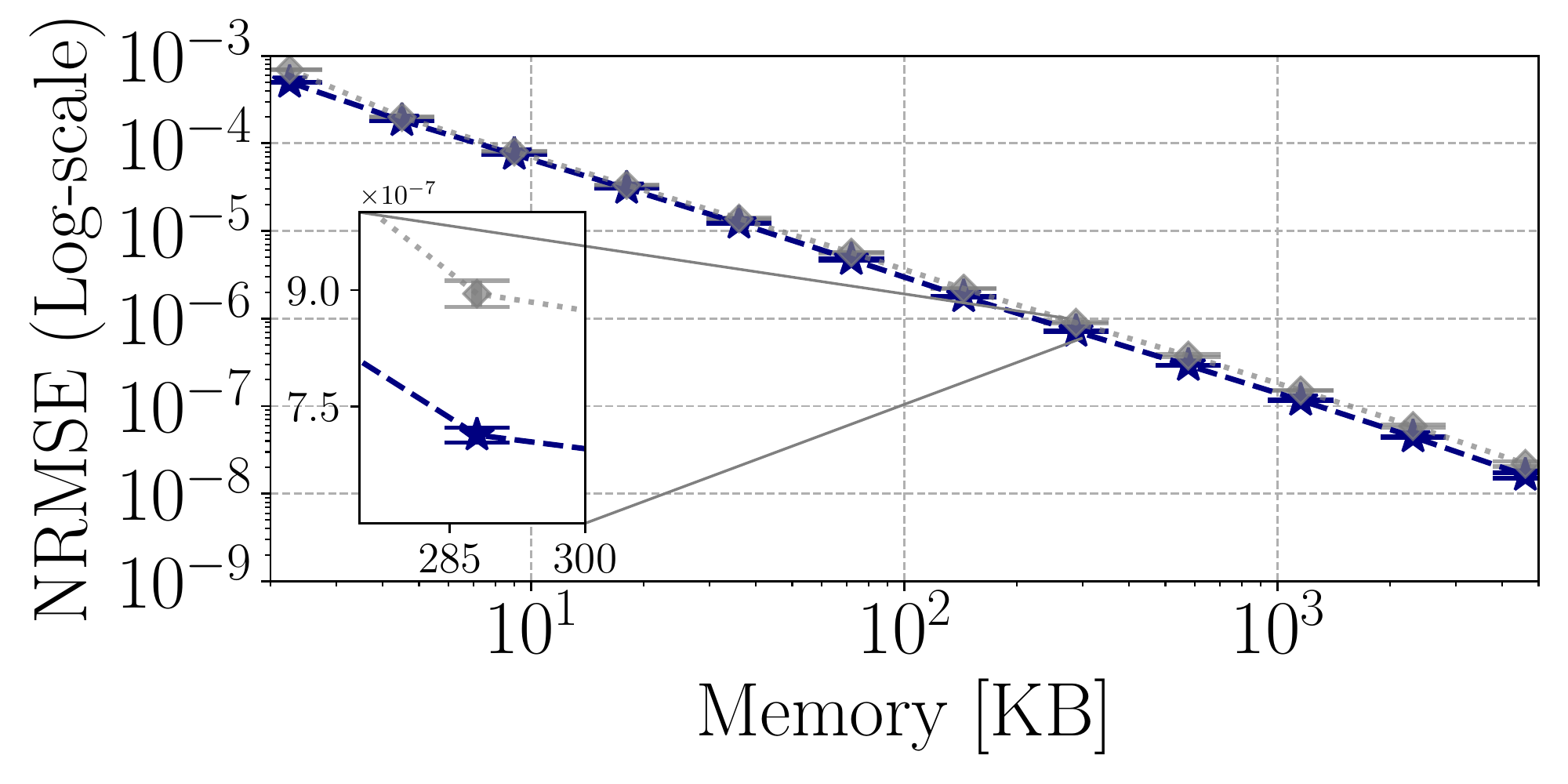}}
    \subfloat[Error, Zipf (2MB)]
    {\includegraphics[width =0.5\columnwidth]
    {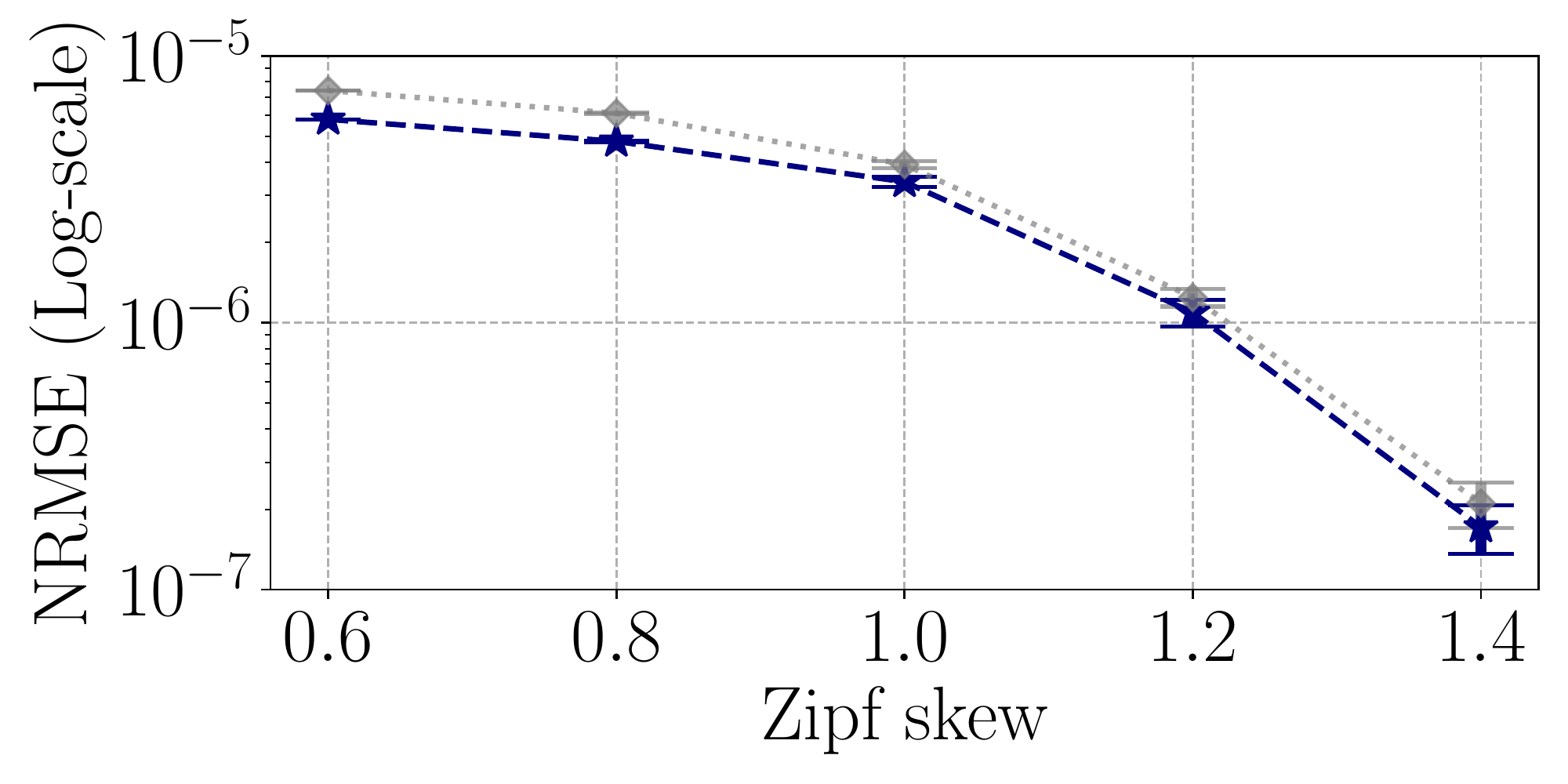}}   
    \hspace*{-1mm}\\
    {\includegraphics[width =0.6\columnwidth]
    {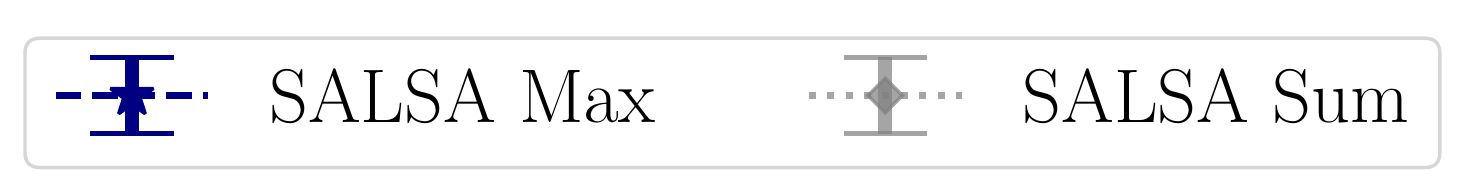}\ifdefined\sigmodSubmission\vspace*{-1mm}\fi}
    \caption{\small Accuracy of SALSA CMS with Sum merge vs. Max merge.
    }\label{fig:merging-eval}
    \ifdefined\sigmodSubmission\vspace*{-3mm}\fi
\end{figure}

\begin{figure}[t]
    \centering
    \hspace*{-2mm}
    
    \subfloat[\rev{Varying the threshold $\phi$}]
    {\label{fig:smallCountersVsPhi} \includegraphics[width =0.51\columnwidth]
    {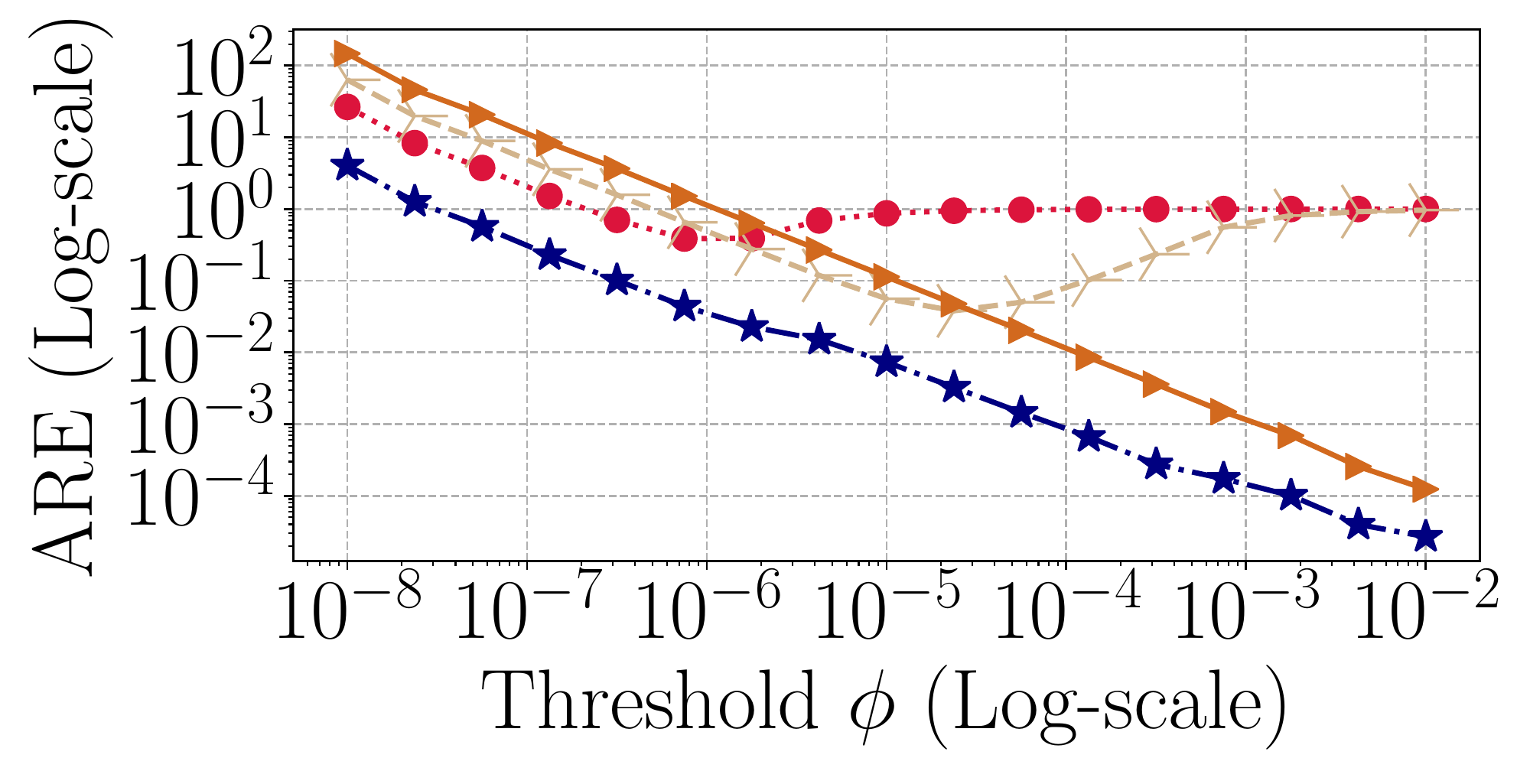}}
    \hspace*{-3mm}
    \subfloat[\rev{Varying stream length ($\phi=10^{-4})$}]
    {\label{fig:smallCountersVsLength} \includegraphics[width =0.51\columnwidth]
    {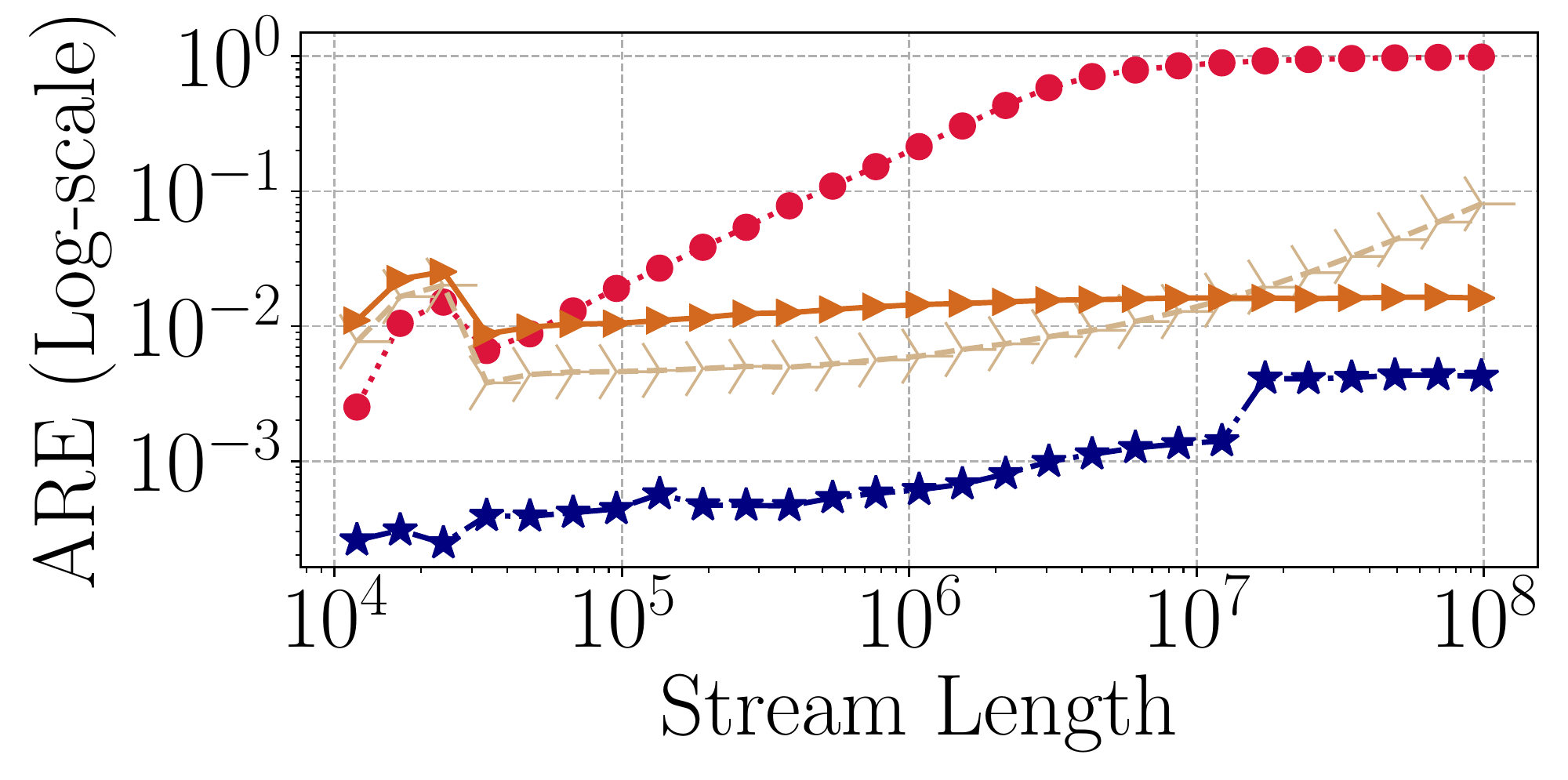}}   
    \hspace*{-1mm}\\
    {\includegraphics[width =0.96\columnwidth]
    {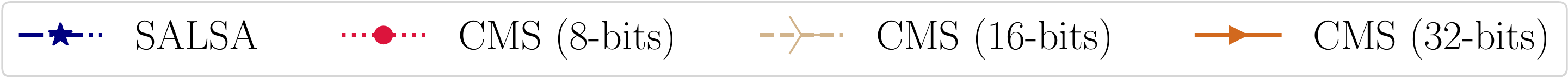}\ifdefined\sigmodSubmission\vspace*{-1mm}\fi}
    \hspace*{-1mm}
    \vspace*{-0mm}
    \caption{\mbox{\rev{SALSA CMS vs. CMS with small counters (2MB).}}
    }\label{fig:smallCounters}
    \ifdefined\sigmodSubmission\vspace*{-3mm}\fi
\end{figure}

\ifdefined\icdeSubmission
\else
\begin{figure}[]
    \centering
    \subfloat[Error, NY18]
    { \includegraphics[width =0.49\columnwidth]
    {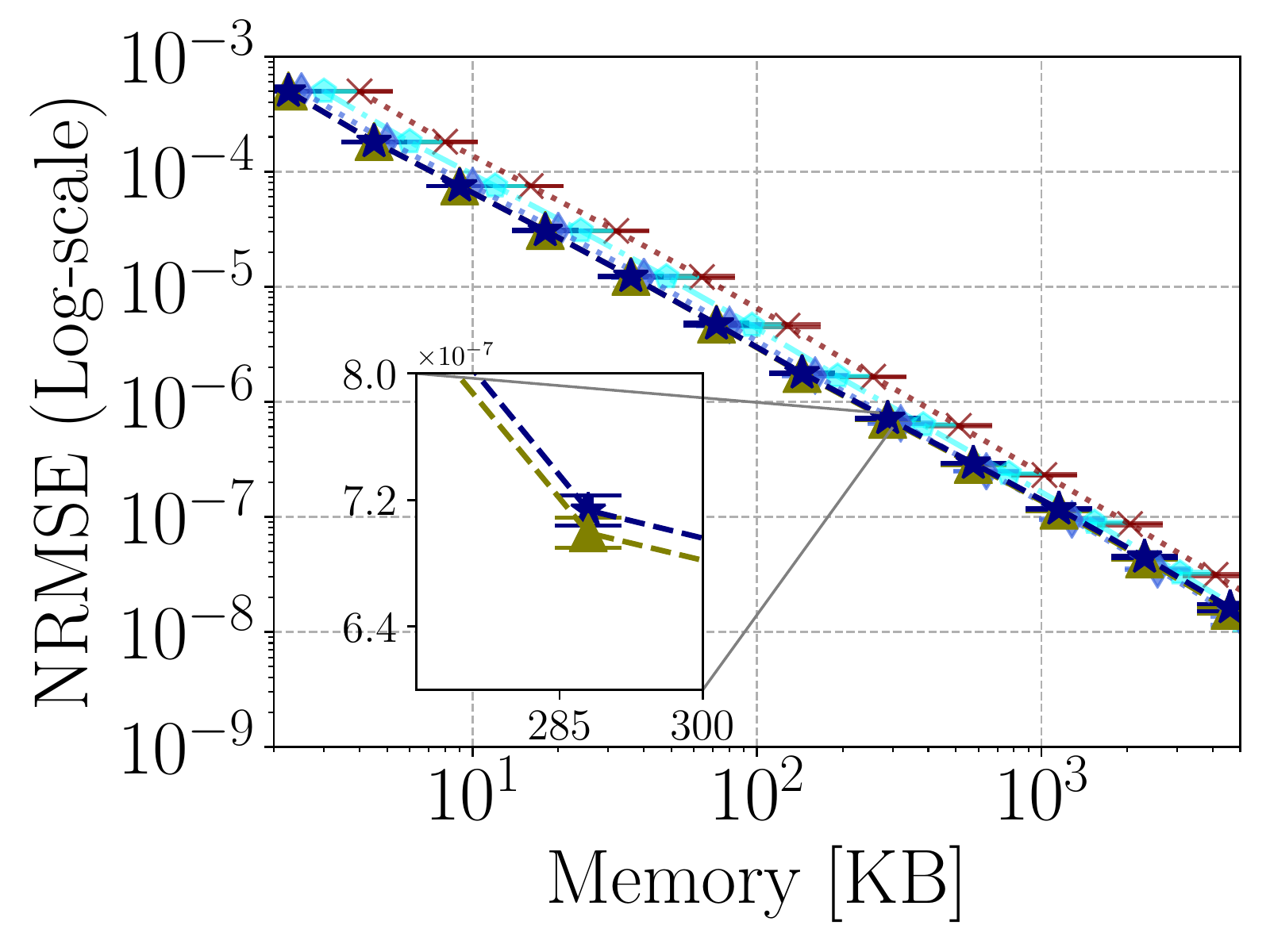}}
    \subfloat[Error, Zipf]
    { \includegraphics[width =0.49\columnwidth]
    {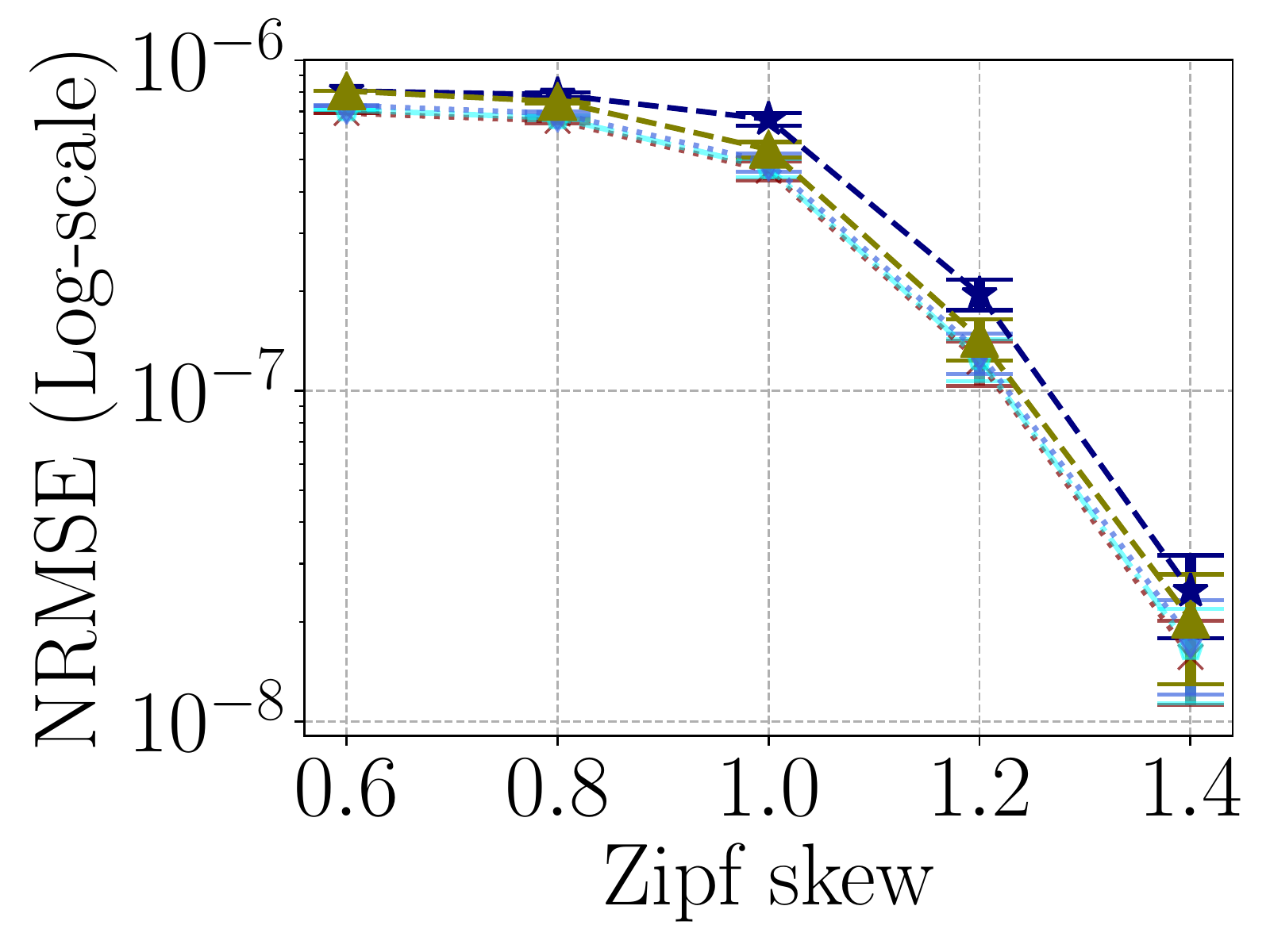}}   
    \\
    {\includegraphics[width =1.00\columnwidth]
    {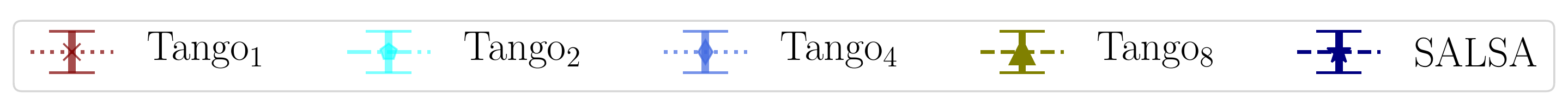}\ifdefined\sigmodSubmission\vspace*{-1mm}\fi}
    \caption{\small Accuracy of SALSA CMS (with $s=8$ bits) vs. Tango CMS. In (b), Tango{\Large$_s$} is allocated with $2(1+1/s)$MB of space while SALSA uses $2(1+1/8)=2.25$MB.
    }\label{fig:tango}
    \ifdefined\sigmodSubmission\vspace*{-3mm}\fi
\end{figure}
\fi

\begin{figure*}[t]
    \centering
     \hspace*{-2mm}
    \subfloat[Speed, NY18\label{fig:pspeed1}]
    { \includegraphics[width =0.5\columnwidth]
    {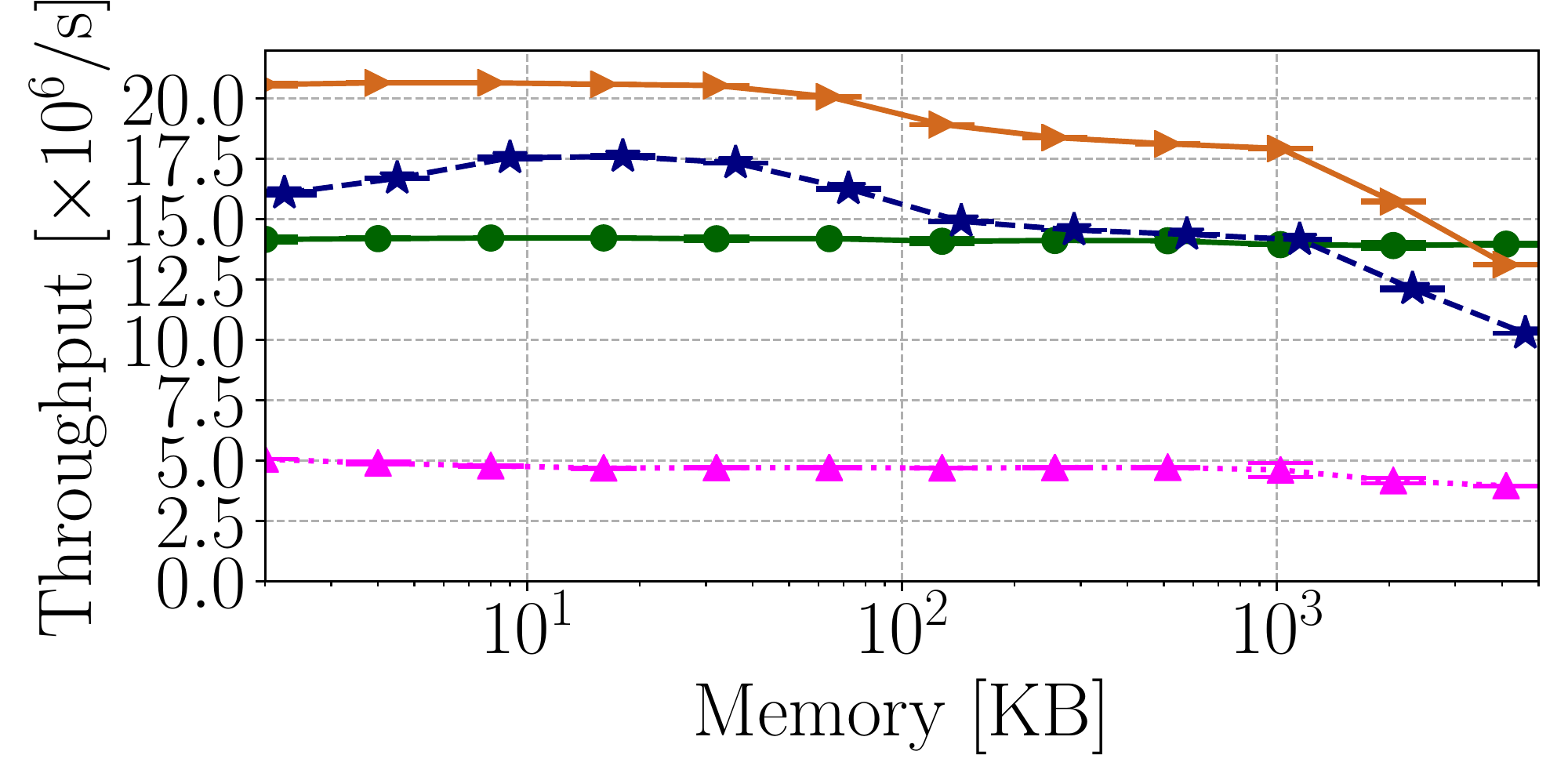}}
    \hspace*{-2mm}
    \subfloat[Speed, CH16\label{fig:pspeed2}]
    { \includegraphics[width =0.5\columnwidth]
    {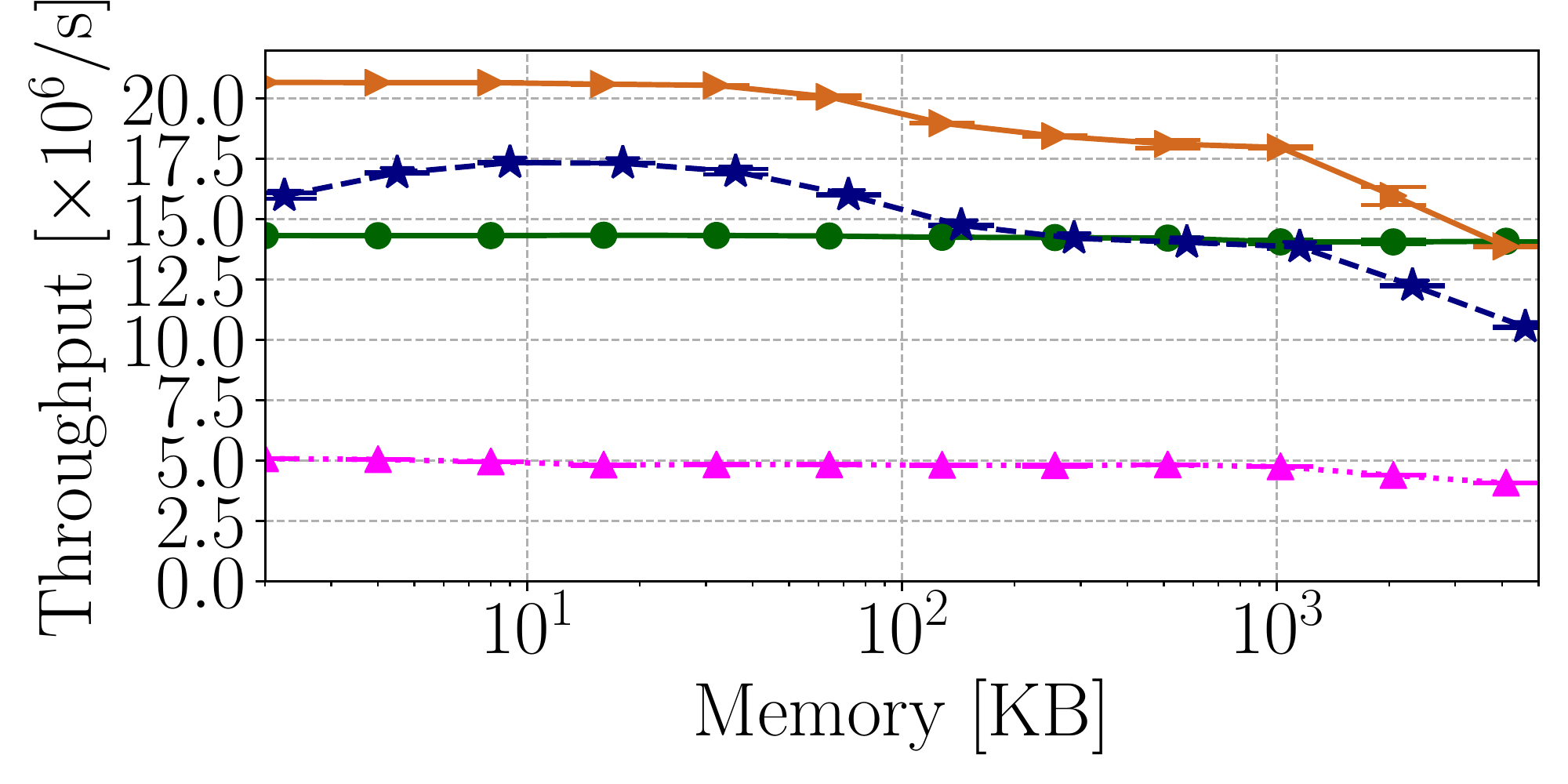}}    
    \hspace*{-2mm}
    \subfloat[NRMSE Error, NY18\label{fig:nrmsep1}]
    { \includegraphics[width =0.5\columnwidth]
    {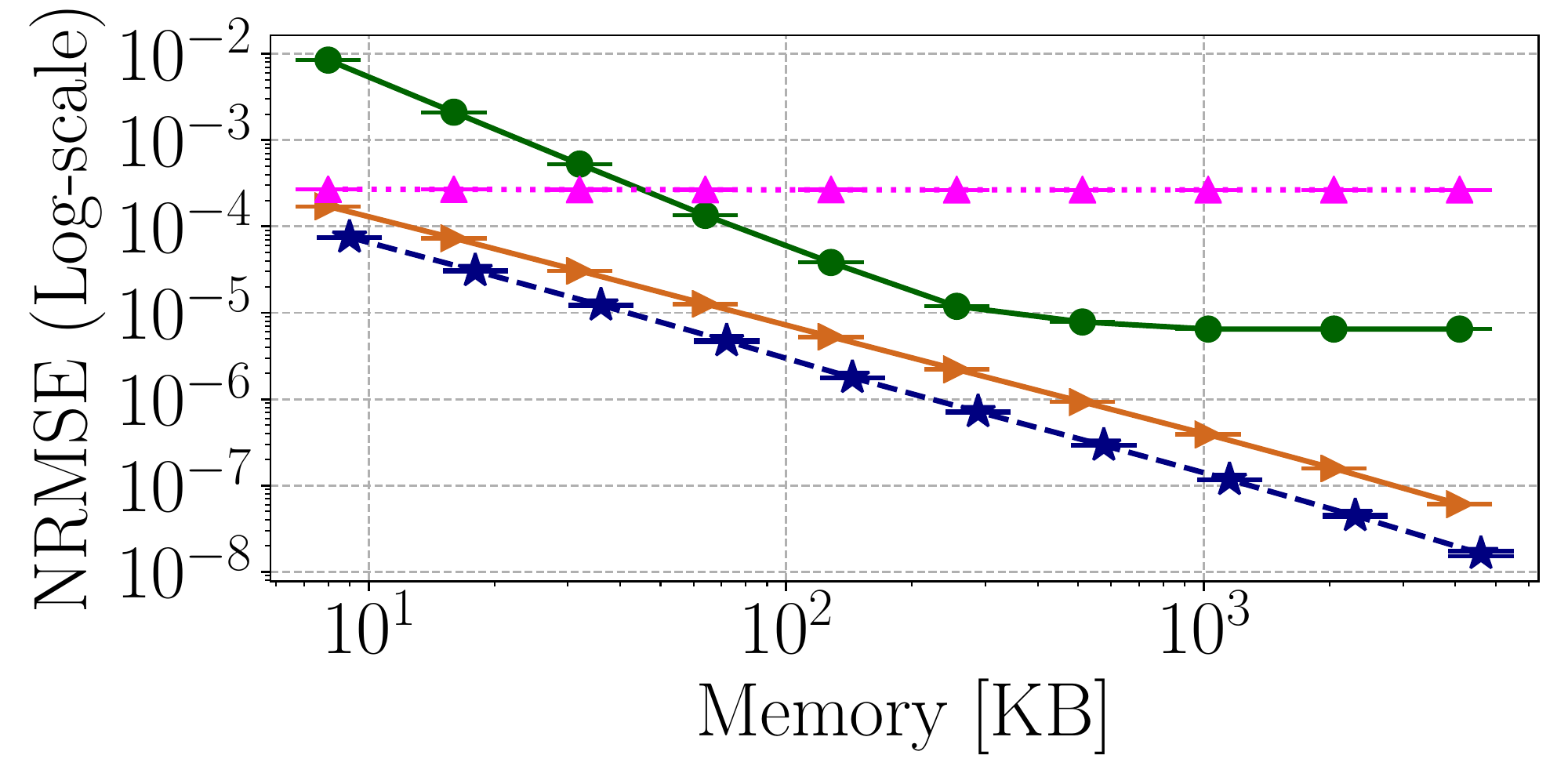}}
    \hspace*{-2mm}
    \subfloat[NRMSE Error, CH16\label{fig:nrmsep2}]
    { \includegraphics[width =0.5\columnwidth]
    {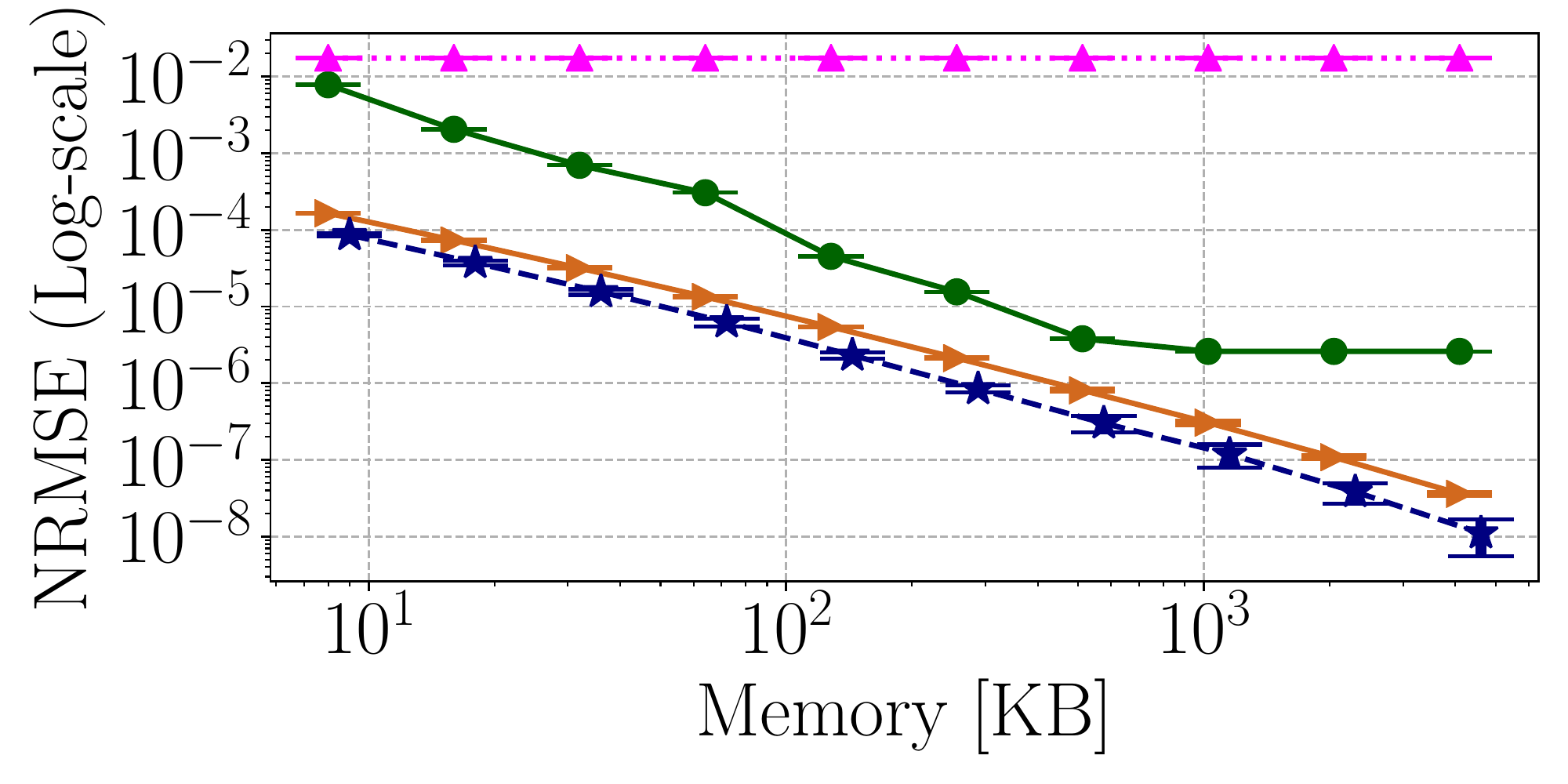}}
    \hspace*{-2mm}    
    \\
    {\includegraphics[width =0.95\columnwidth]
    {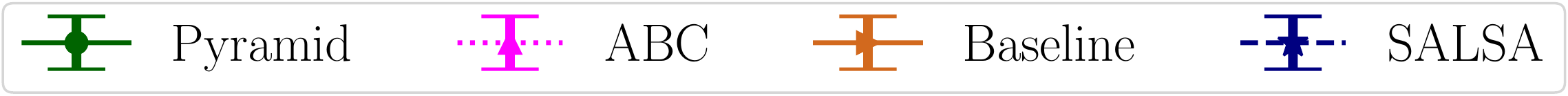}\ifdefined\sigmodSubmission\vspace*{-2mm}\fi}
    \\\ifdefined\sigmodSubmission\vspace*{-1mm}\fi 
    \subfloat[AAE Error, NY18\label{fig:aaeNY18}]
    { \includegraphics[width =0.5\columnwidth]
    {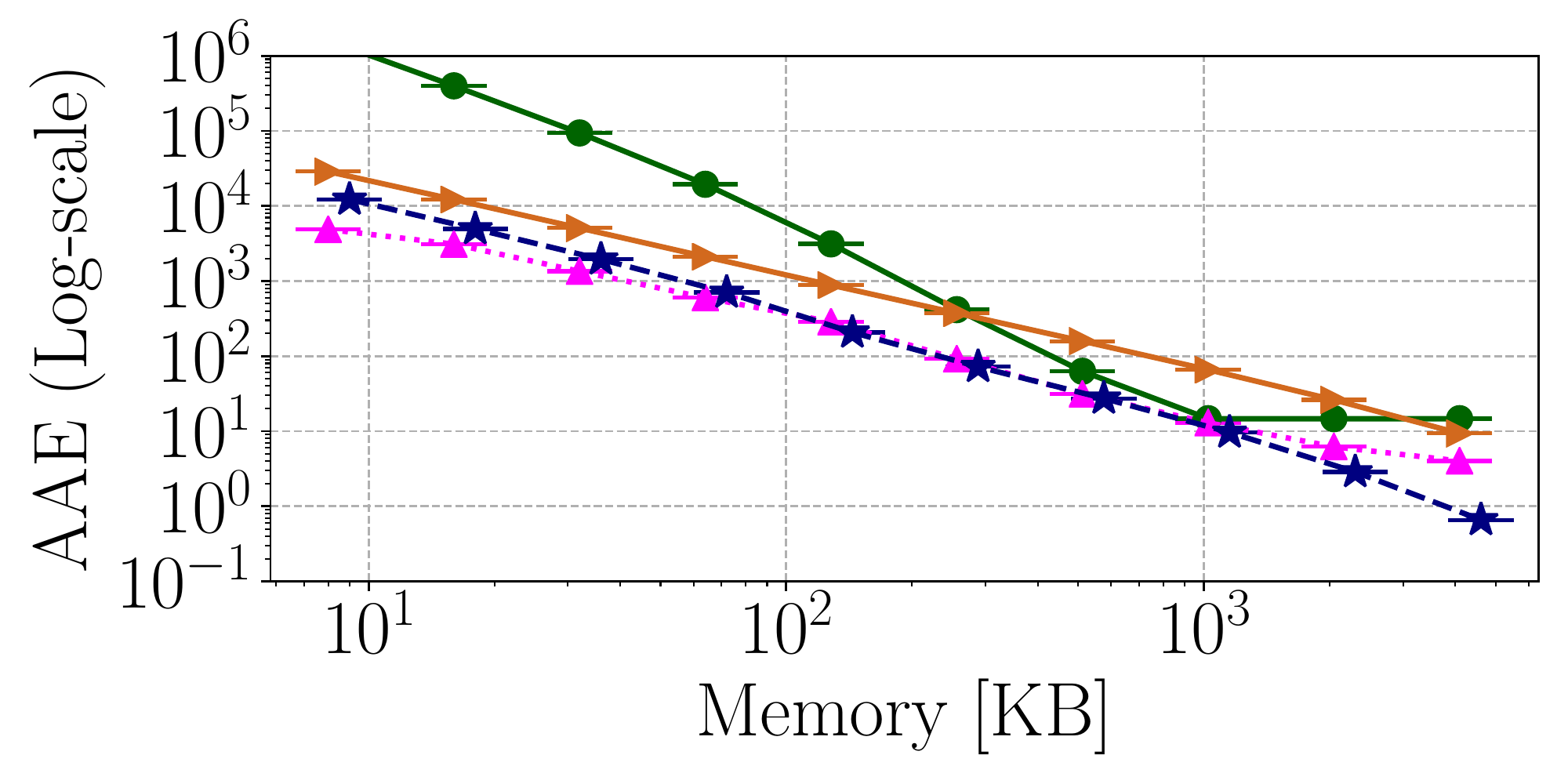}}
    \hspace*{-2mm}
    \subfloat[AAE Error, CH16\label{fig:aaeCH16}]
    { \includegraphics[width =0.5\columnwidth]
    {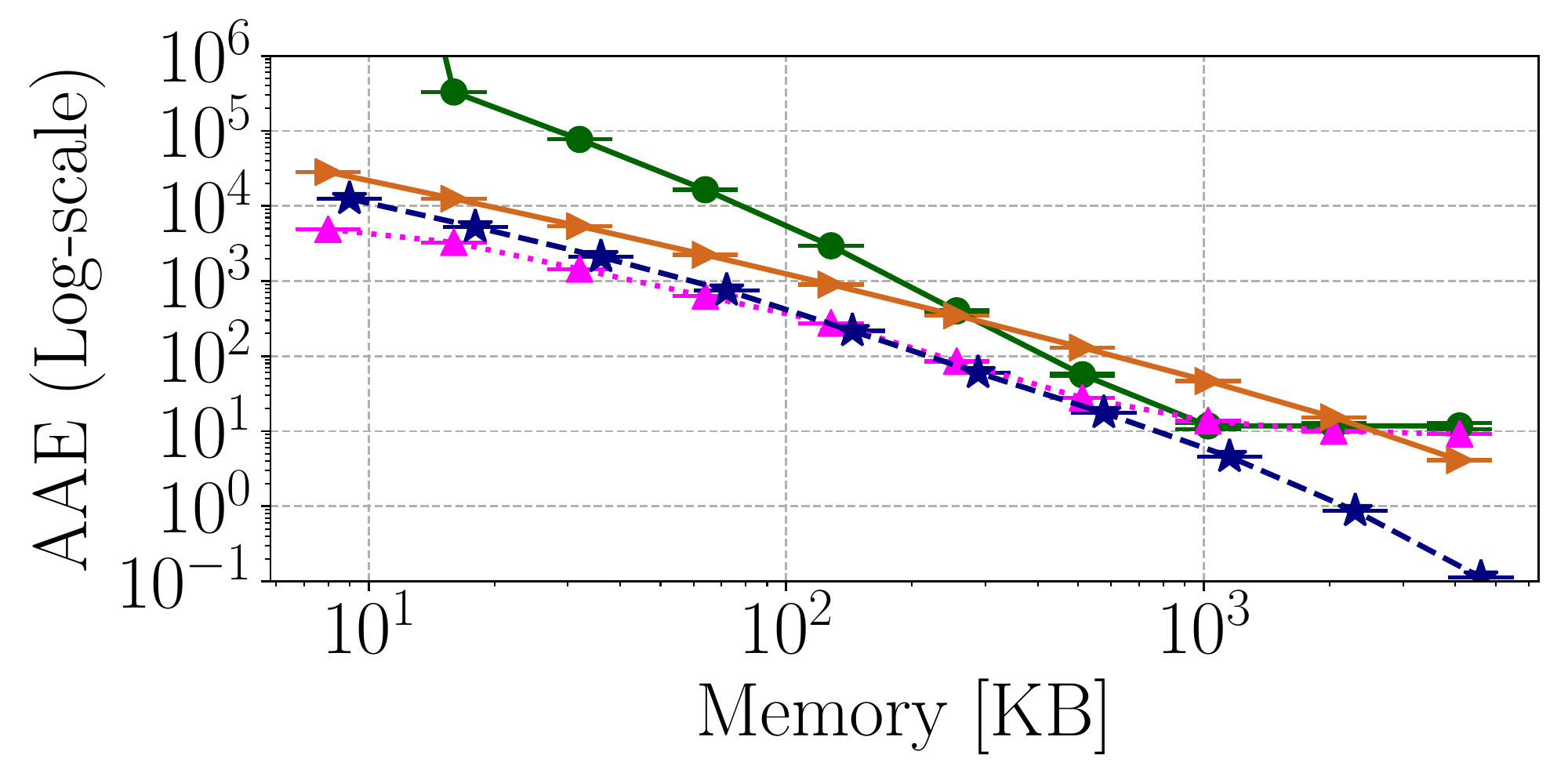}}
    \hspace*{-2mm}     
    \subfloat[ARE Error, NY18\label{fig:areNY18}]
    { \includegraphics[width =0.5\columnwidth]
    {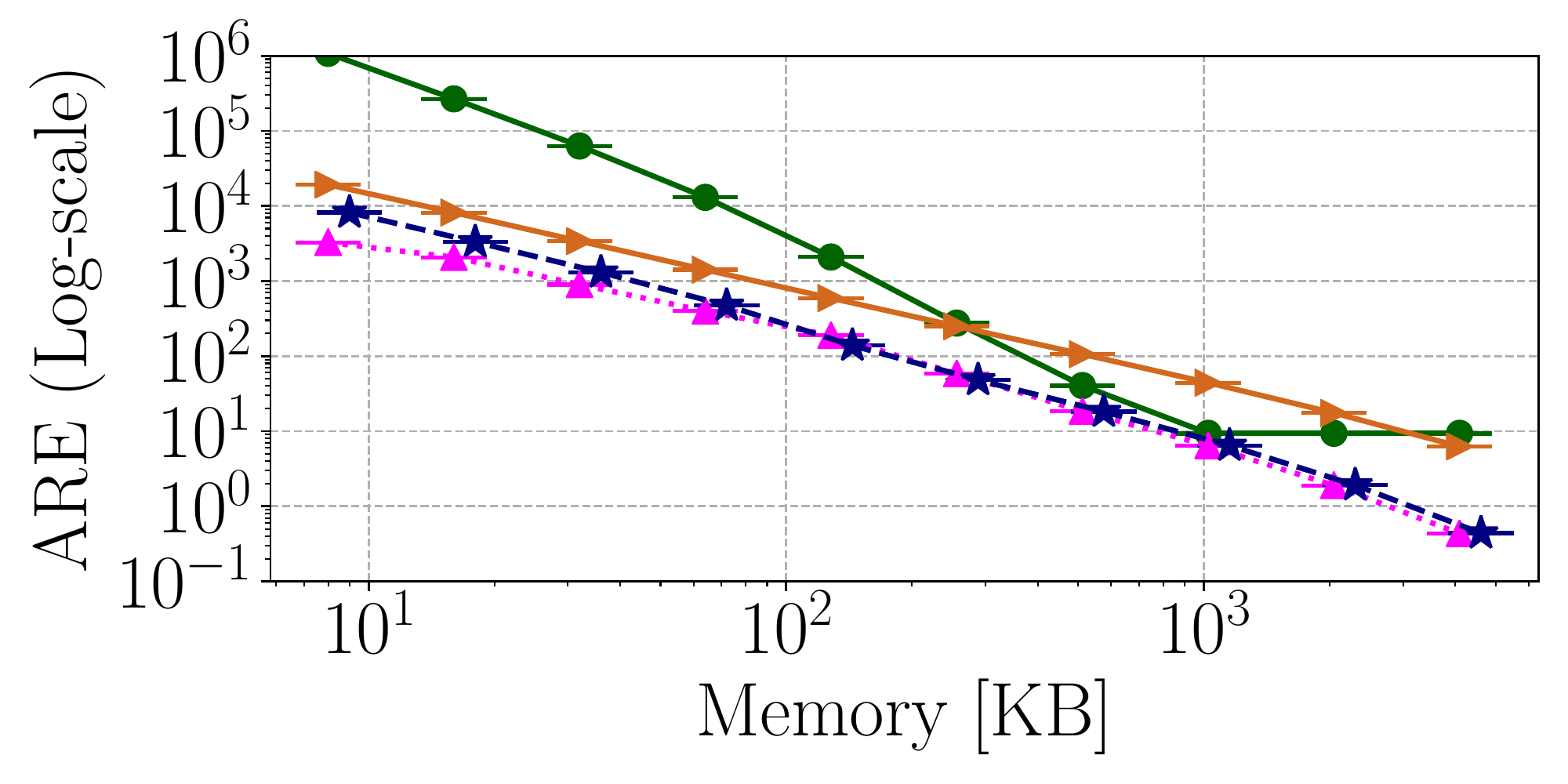}}
    \hspace*{-2mm}
    \subfloat[ARE Error, CH16\label{fig:areCH16}]
    { \includegraphics[width =0.5\columnwidth]
    {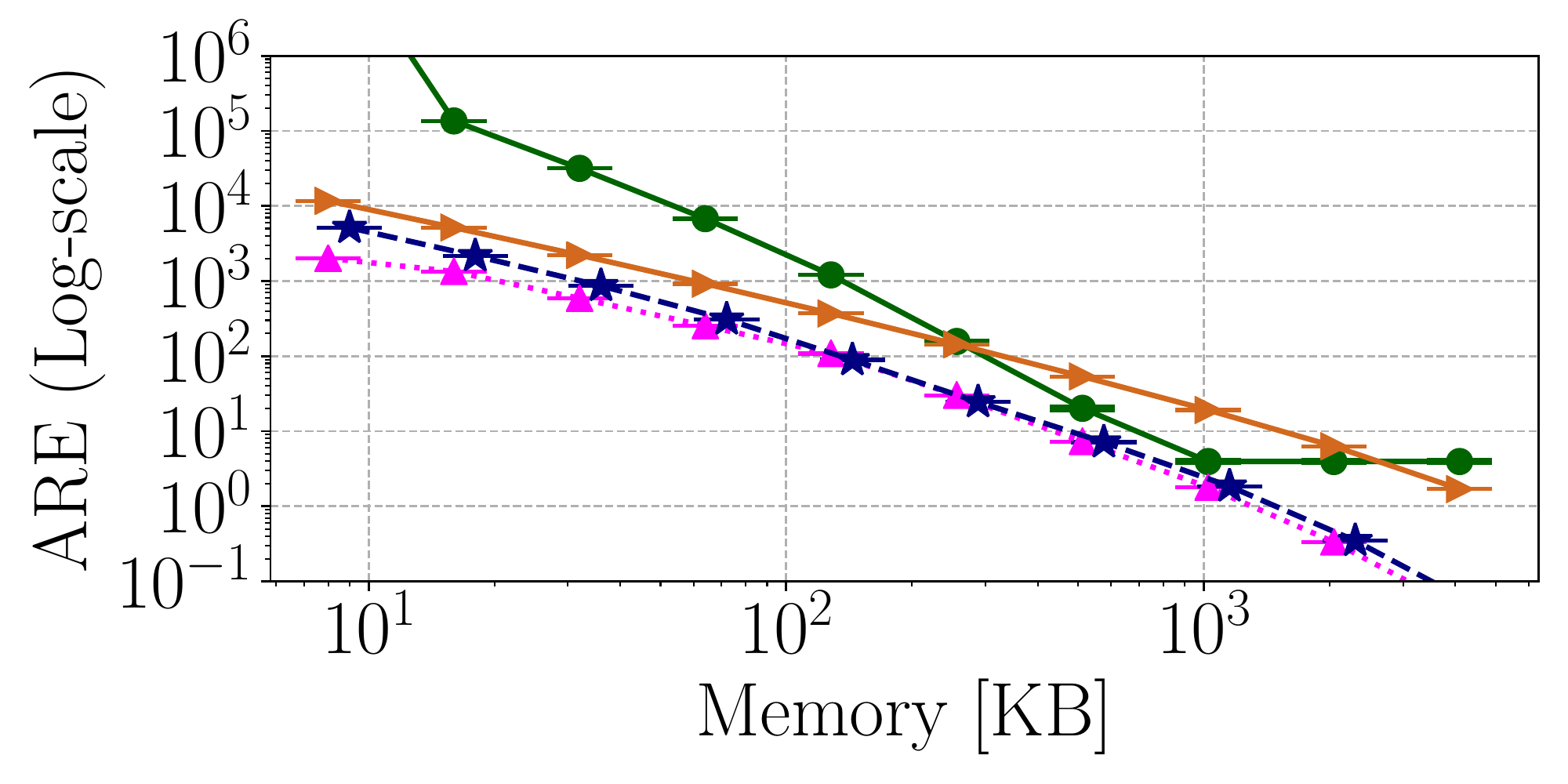}}
    \hspace*{-2mm}         
    \ifdefined\sigmodSubmission\vspace*{-2mm}\fi
    \caption{\small \mbox{Comparing the performance of the SALSA, Pyramid~\cite{PyramidSketch}, ABC~\cite{gong2017abc}, and Baseline versions of CMS.
    }\label{fig:Pyramid}}
    \ifdefined\sigmodSubmission\vspace*{-7mm}\fi 
\end{figure*}
\begin{figure}[t]
    \centering
    \ifdefined\sigmodSubmission
    \ifdefined\sigmodSubmission\vspace*{-4mm}\fi
    \fi
    \hspace*{-2mm}
    
    \subfloat[NY18]
    { \includegraphics[width =0.505\columnwidth]
    {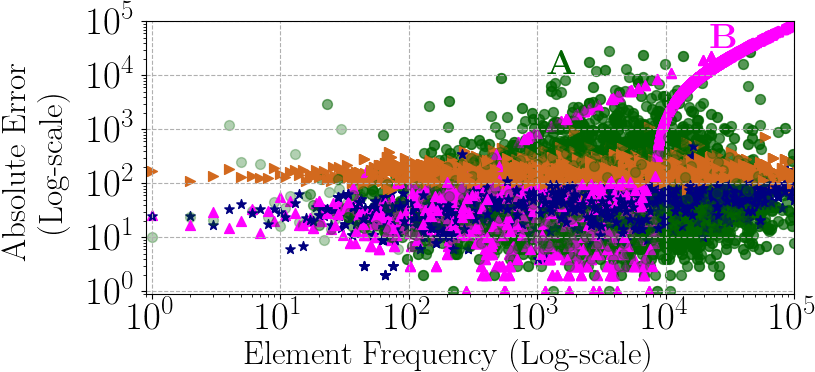}}
    \hspace*{-2mm}
    \subfloat[CH16
    ]
    { \includegraphics[width =0.505\columnwidth]
    {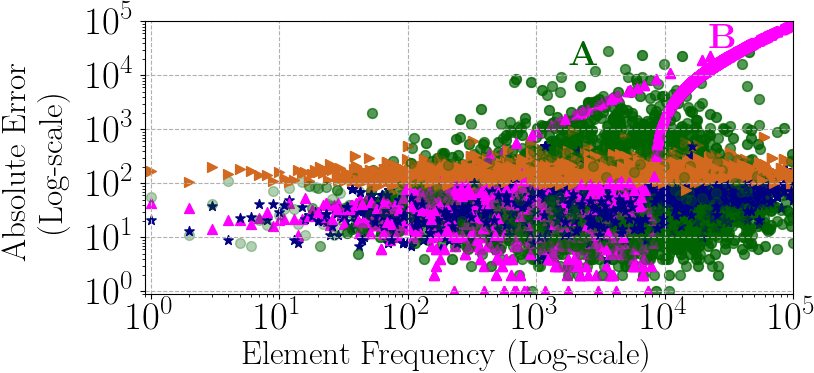}}   
    \hspace*{-2mm}\\
    {\includegraphics[width =0.96\columnwidth]
    {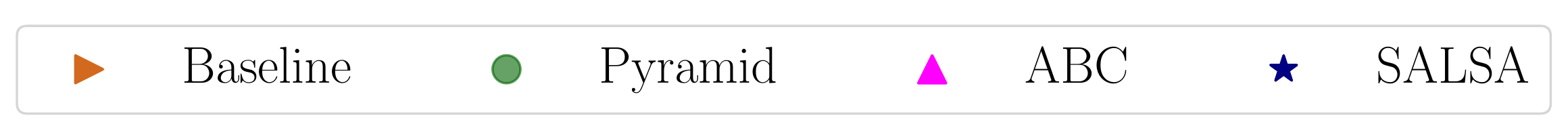}\ifdefined\sigmodSubmission\vspace*{-1mm}\fi}
    \hspace*{-1mm}
    \vspace*{-0mm}
    \caption{\mbox{The error distribution of the algorithms (2MB).}
    }\label{fig:scatter}
    \ifdefined\sigmodSubmission\vspace*{-2mm}\fi
\end{figure}

\begin{figure*}[]
    \centering
    \hspace*{-3mm}
    \subfloat[Error, NY18]
    {\label{3a}\includegraphics[width =0.5\columnwidth]
    {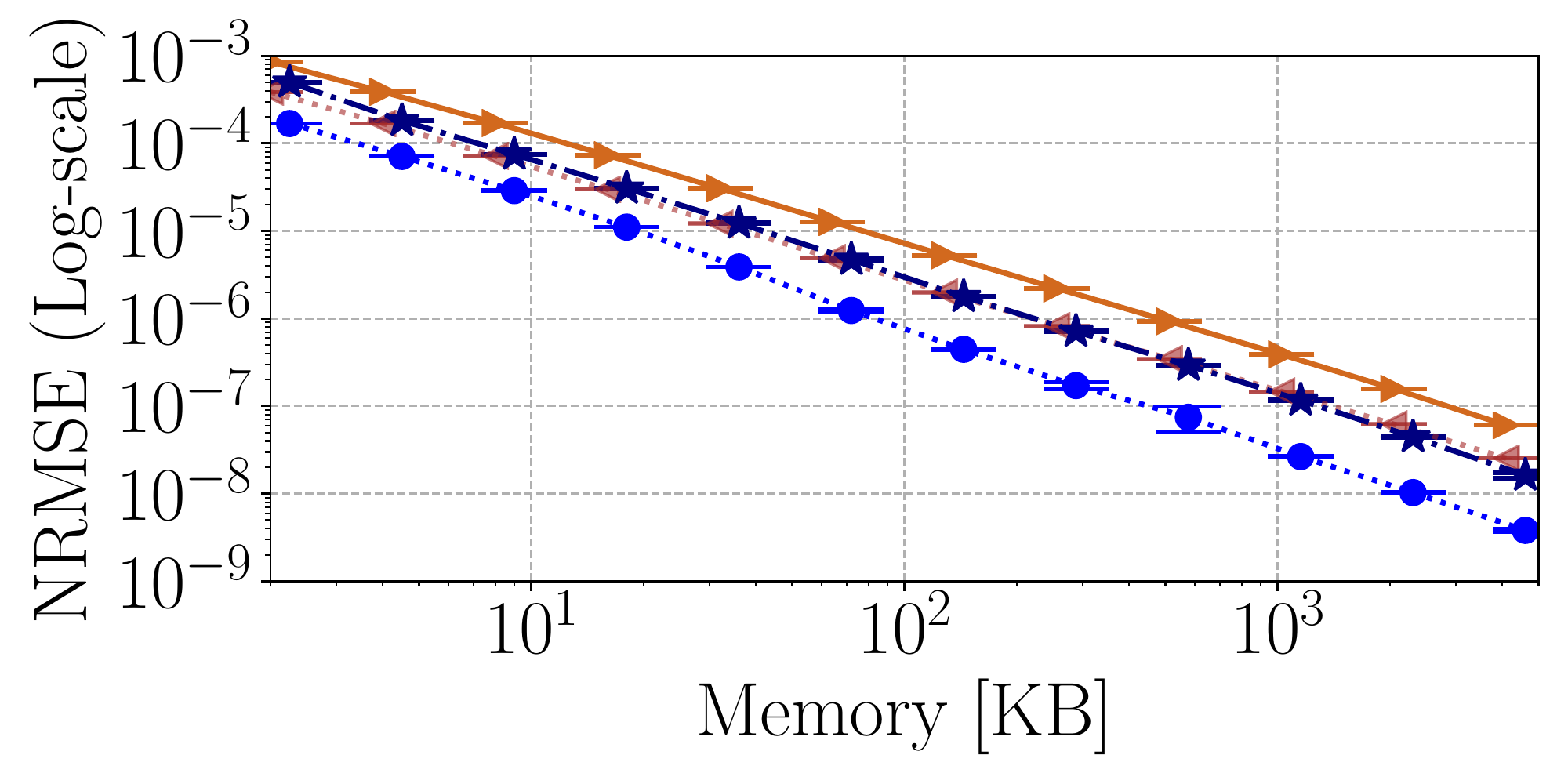}}
    \subfloat[Error, CH16]
    {\label{3b}\includegraphics[width =0.5\columnwidth]
    {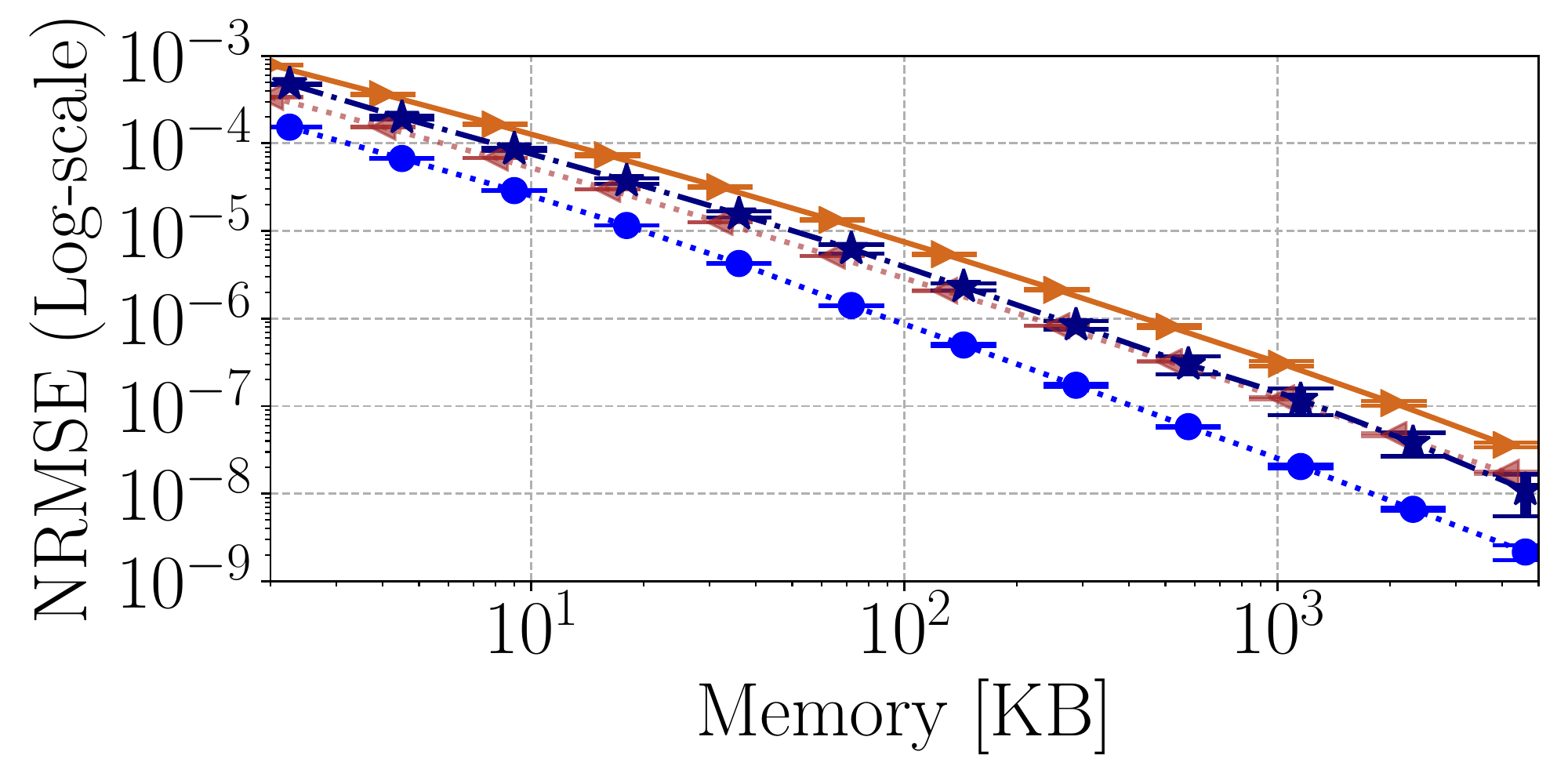}}
    \subfloat[Error, Univ2]
    {\label{3c}\includegraphics[width =0.5\columnwidth]
    {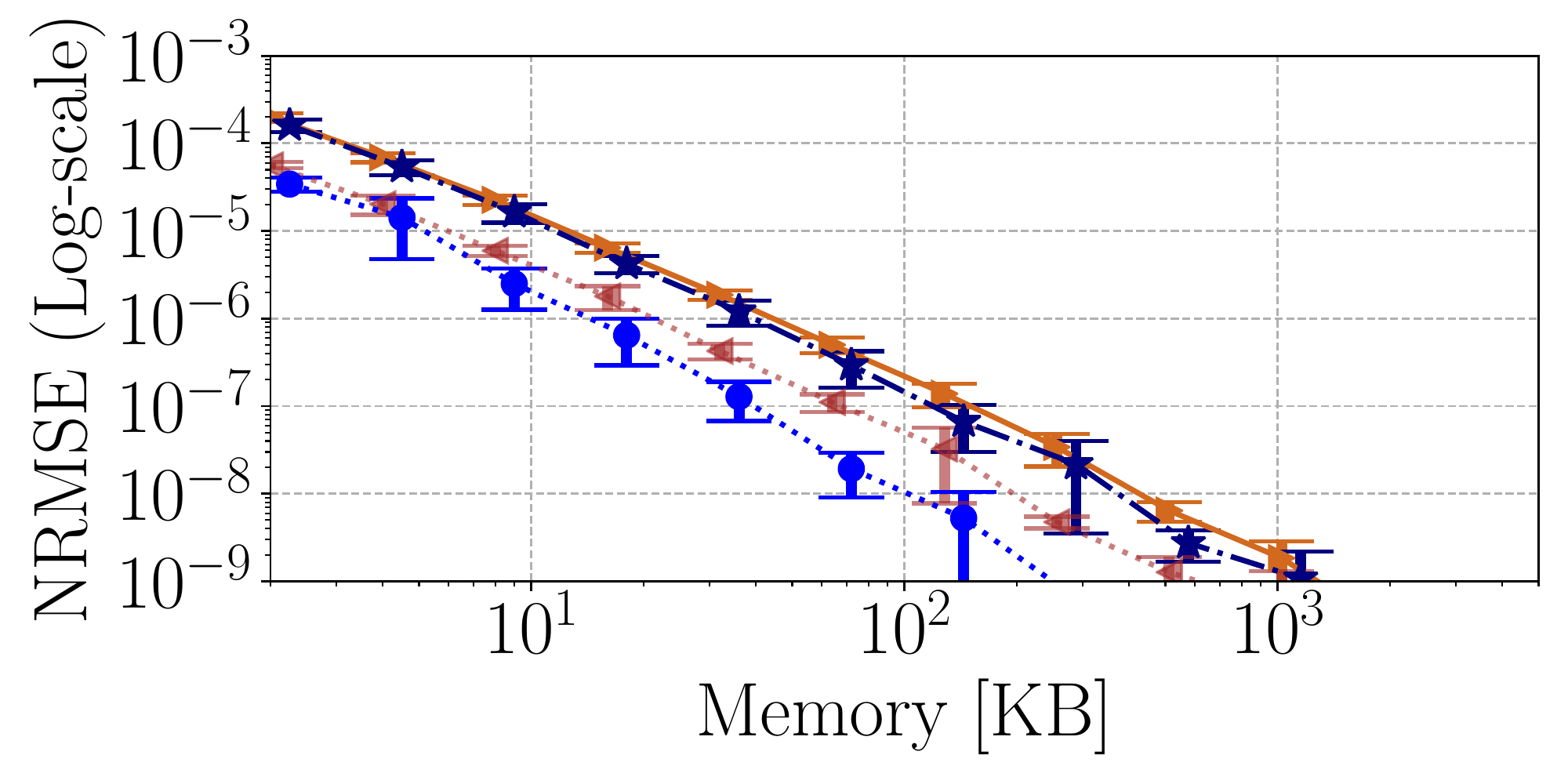}}
    \subfloat[Error, YouTube]
    {\label{3d}\includegraphics[width =0.5\columnwidth]
    {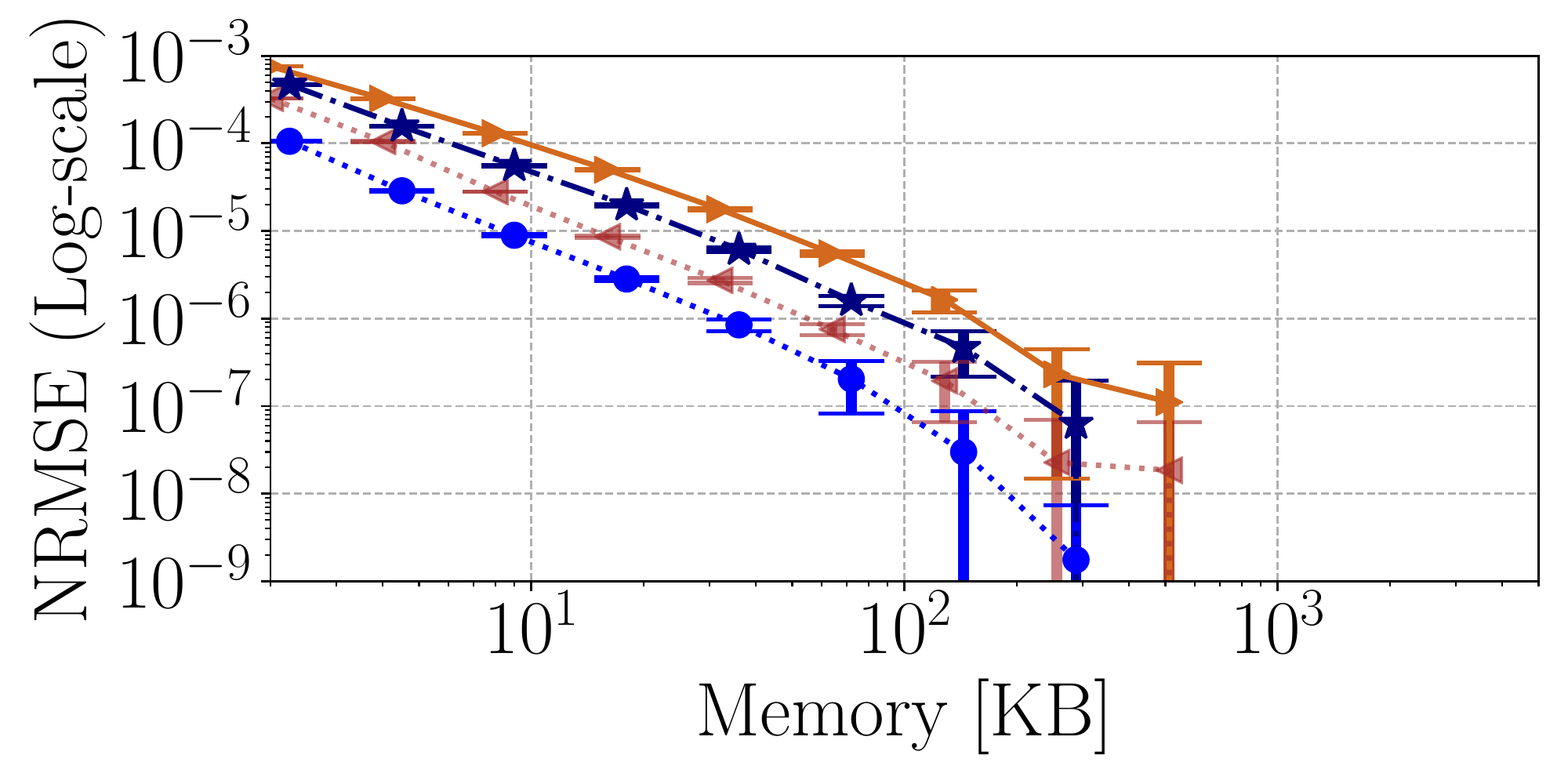}}    \\ 
    {\includegraphics[width=1.42\columnwidth]
    {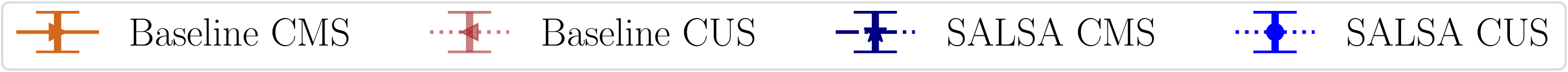}\ifdefined\sigmodSubmission\vspace*{-3mm}\fi}\\
    \hspace*{-3mm}
    \subfloat[Speed, NY18]
    {\label{3e}\includegraphics[width =0.5\columnwidth]
    {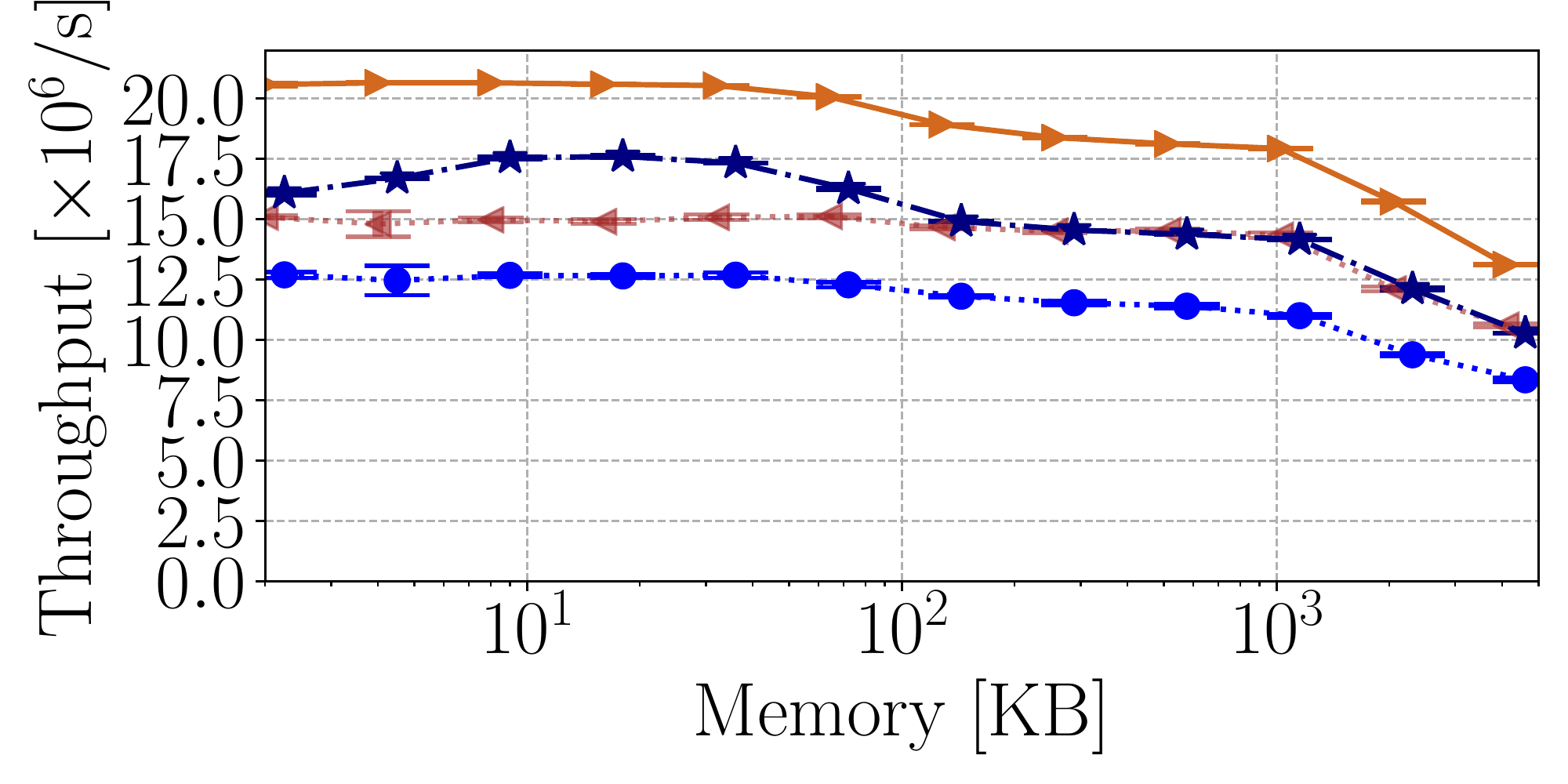}}
    \subfloat[Speed, CH16]
    {\label{3f}\includegraphics[width =0.5\columnwidth]
    {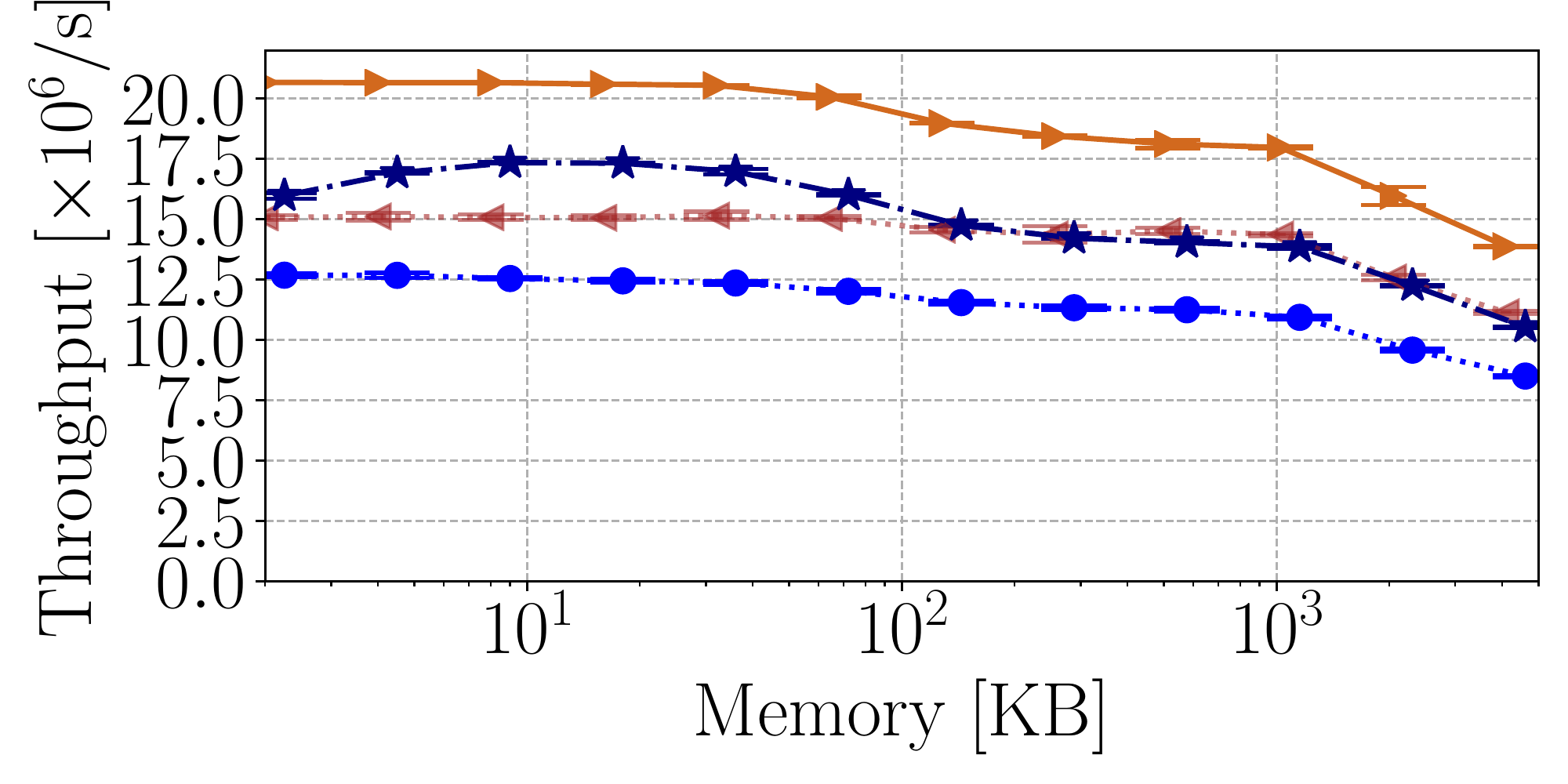}}
    \subfloat[Speed, Univ2]
    {\label{3g}\includegraphics[width =0.5\columnwidth]
    {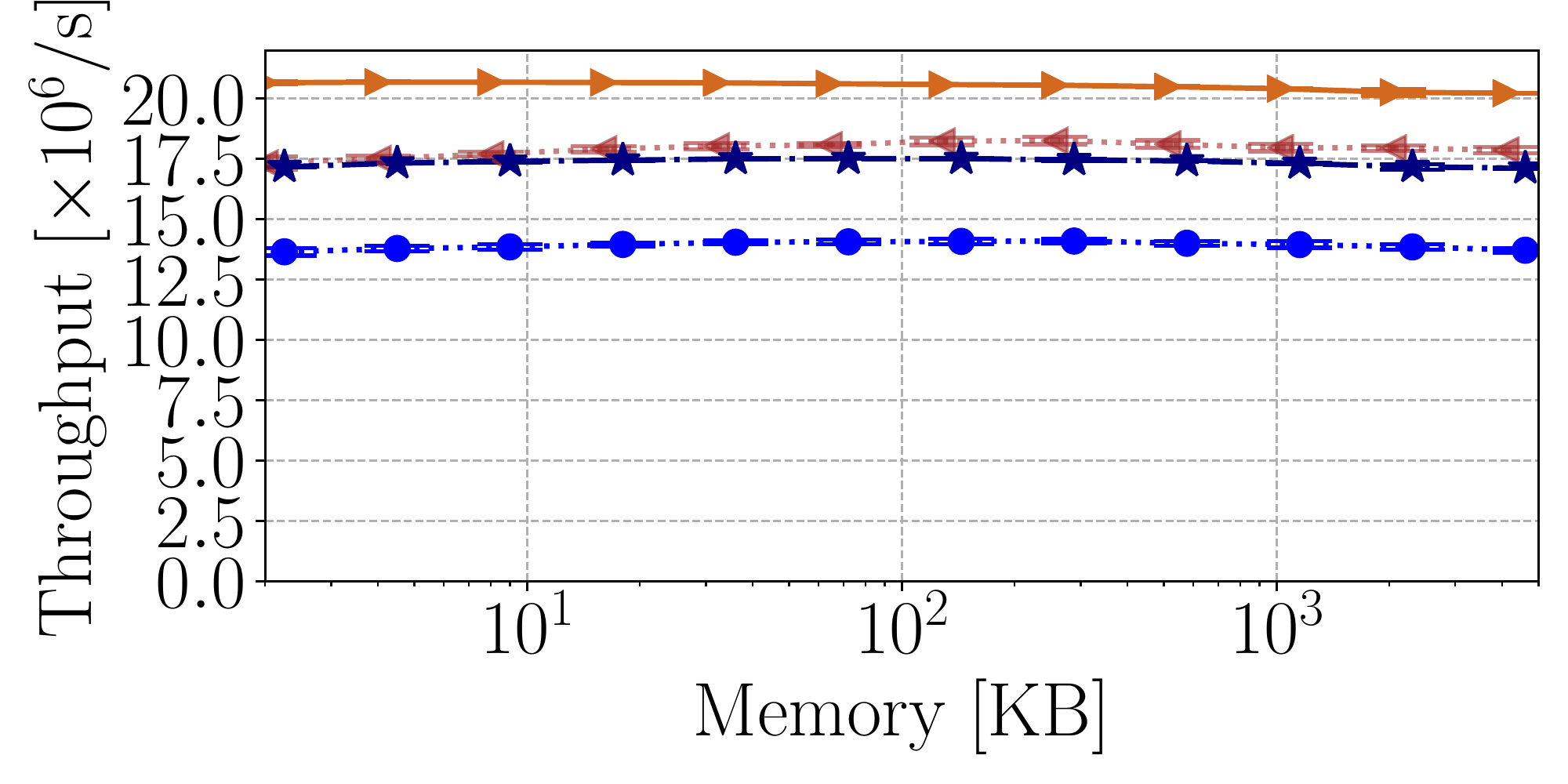}}
    \subfloat[Speed, YouTube]
    {\label{3h}\includegraphics[width =0.5\columnwidth]
    {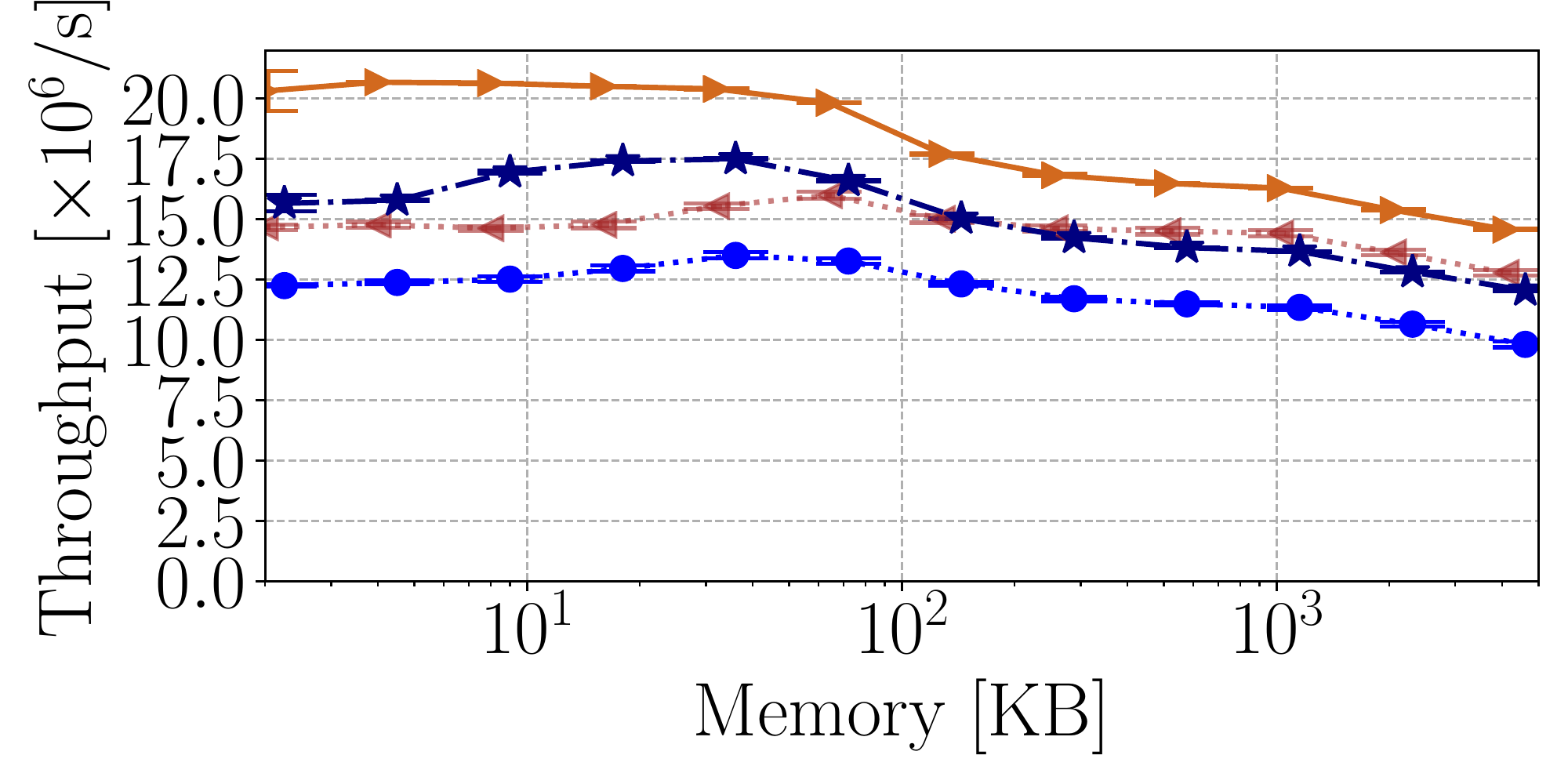}} 
    \ifdefined\sigmodSubmission\vspace*{-1mm}\fi
    \caption{Speed and accuracy of SALSA CMS and SALSA CUS for the real datasets. Notice the log-scale of the error plots. \label{fig:CMSandCUS} }
    \ifdefined\sigmodSubmission\vspace*{-7mm}\fi
\end{figure*}

\textbf{How to Configure SALSA?} 
We perform preliminary experiments to determine the default SALSA configuration.  

\textbf{How Large Should Counters Be?}
We first determine the most effective minimal counter size ($s$) for SALSA. Intuitively, for a fixed row width $w$, smaller $s$ results in lower error but also in larger encoding overheads.  With this tradeoff, it may not be profitable to reduce $s$. 
In this experiment, we fixed the memory of the counters and deliberately ignored the encoding overheads for SALSA. The goal is to quantify the attainable improvement from using smaller counters.
We used synthetic Zipfian trace with skews varying from 0.6 to 1.4.

As shown in Figure~\ref{fig:Zipf}, most of the improvement comes from reducing the counter sizes from $32$ bits to $8$ bits. These results were consistent across different memory footprints.  This is not surprising as almost all counters merge up to at least $8$ bits, but then many do not overflow further. 
We also observe that SALSA offers more gains for low-skew traces and that CS is better suited for lower skews, while CMS offers comparable accuracy with less space for high skew. 
Hereafter, we use $s=8$ bits as the default SALSA configuration. While SALSA with $s=4$ bits is slightly more accurate for low-skew workloads and high memory footprints, its encoding overhead of about 25\% of the sketch size (compared to 12.5\% for $s=8$) is too \mbox{large to justify the benefits.}

\textbf{Which Merging Should We Use?}
As mentioned above, SALSA CS must use sum-merging, and so does SALSA CMS for Strict Turnstile streams. Similarly, in SALSA CUS, we need to use max-merging. This leaves only the choice for SALSA CMS in Cash Register streams, where we can use either sum-merging or max-merging. We quantify the difference in accuracy in Figure~\ref{fig:merging-eval}.
As shown, max-merging is slightly more accurate, especially for low-skew workloads. We conclude that if one only targets Cash Register streams, it is better to use max-merging, but the accuracy of \mbox{sum-merging is not far behind.}

\ifdefined\icdeSubmission
\else
\textbf{Is Fine-grained Merging Worth It?}
To understand the accuracy improvement attainable by fine-grained merging (as opposed to SALSA's approach of doubling the counter size at each overflow), we compare SALSA with Tango. 
As the results in Figure~\ref{fig:tango} indicate, Tango also offers the best accuracy-space tradeoff when starting with $s=8$ bits (Tango$_{16}$ is equivalent to SALSA$_{16}$ and is omitted). However, while it is slightly more accurate, the gains seem marginal considering the computationally expensive operations of determining the counter's size and offset. Further, Tango has an overhead of $1$ bit per counter and does not obviously allow an efficient encoding \mbox{like SALSA does (Section~\ref{sec:smartEncoding}).}
\fi

\begin{figure}[t]
    \centering
    \hspace*{-2mm}
    \subfloat[Error, NY18]
    { \includegraphics[width =0.5\columnwidth]
    {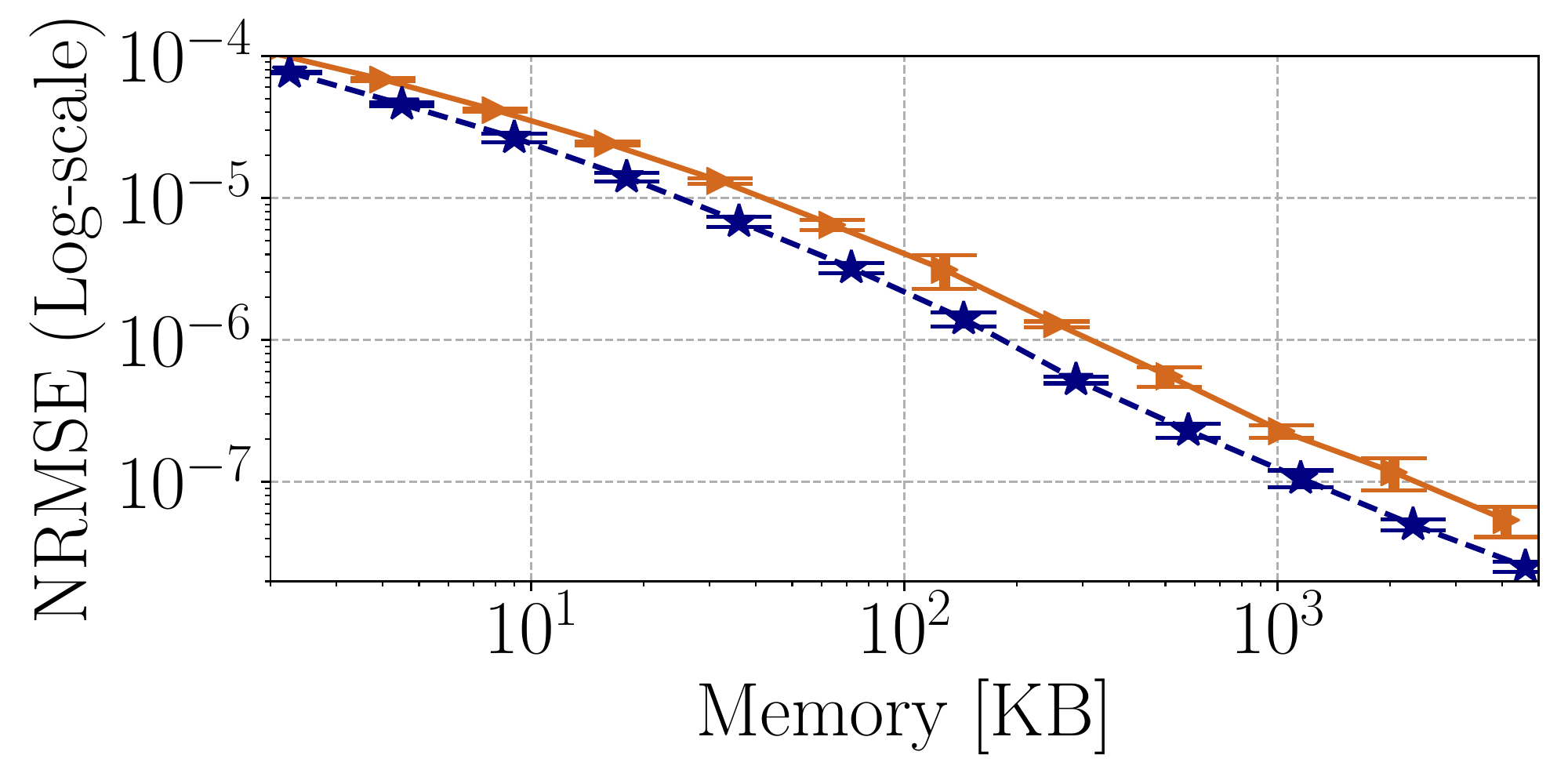}}
    \hspace*{-2mm}
    \subfloat[Error, CH16]
    { \includegraphics[width =0.5\columnwidth]
    {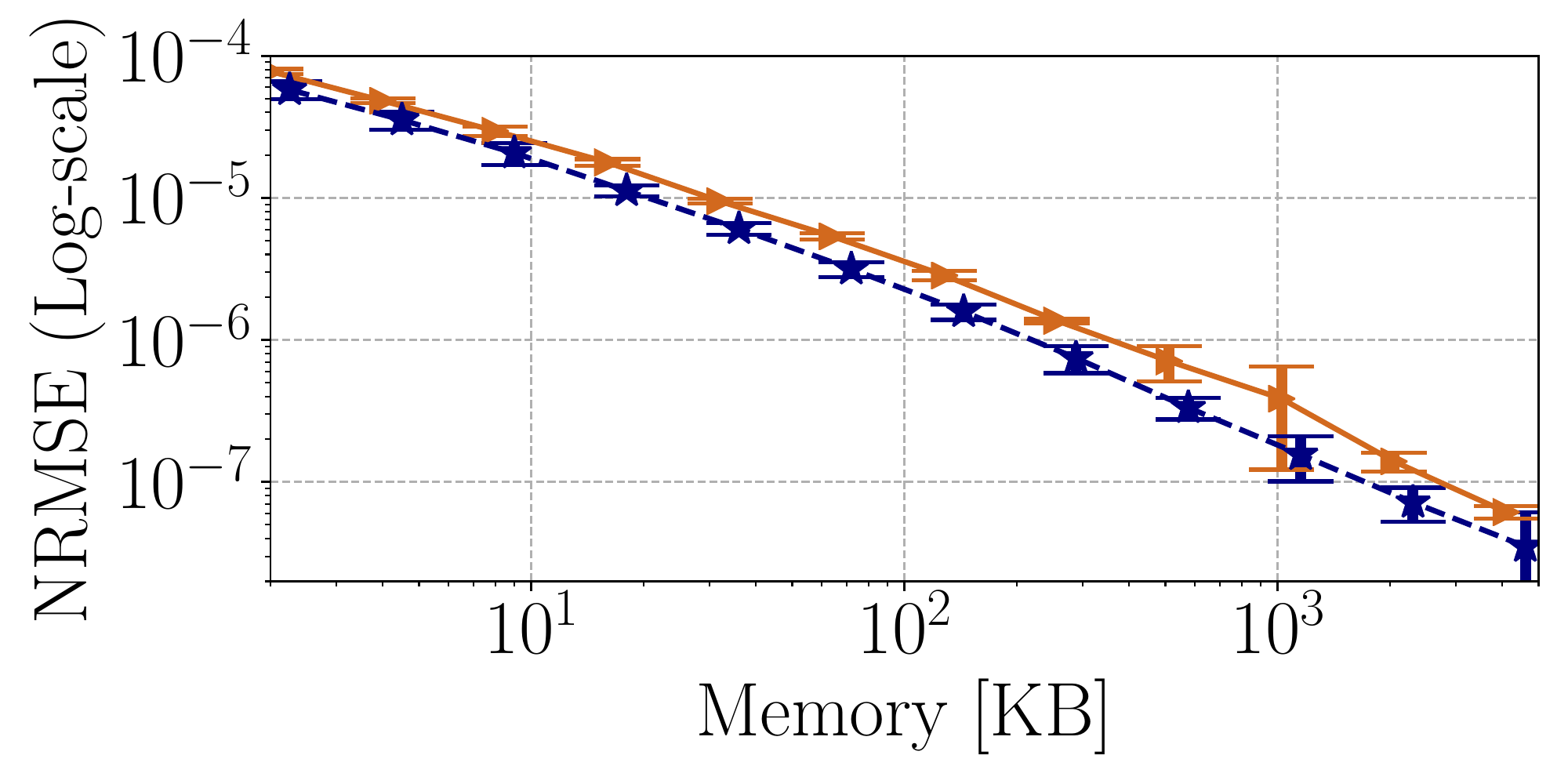}}   
    \hspace*{-1mm}\\
    {\includegraphics[width =0.6\columnwidth]
    {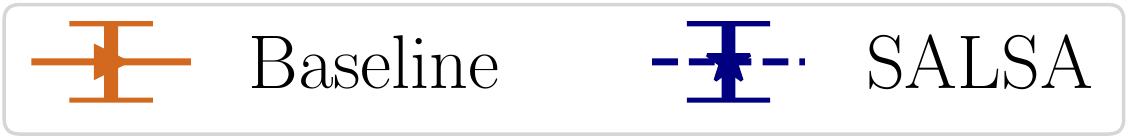}\ifdefined\sigmodSubmission\vspace*{-3mm}\fi}\\
        \subfloat[Error, 
        Univ2]
    {\includegraphics[width =0.5\columnwidth]
    {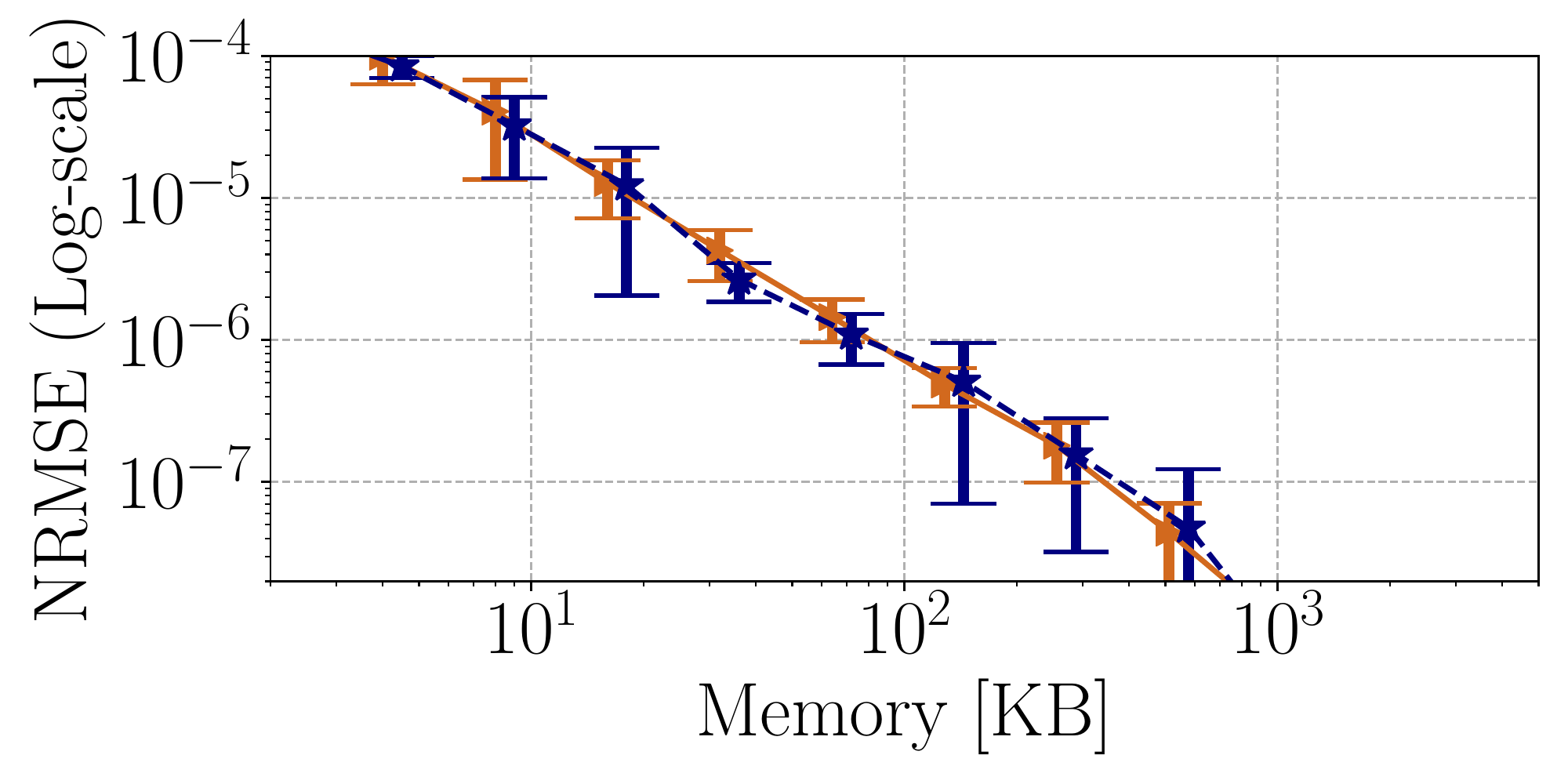}}    
    \subfloat[Error, YouTube]
    {\includegraphics[width =0.5\columnwidth]
    {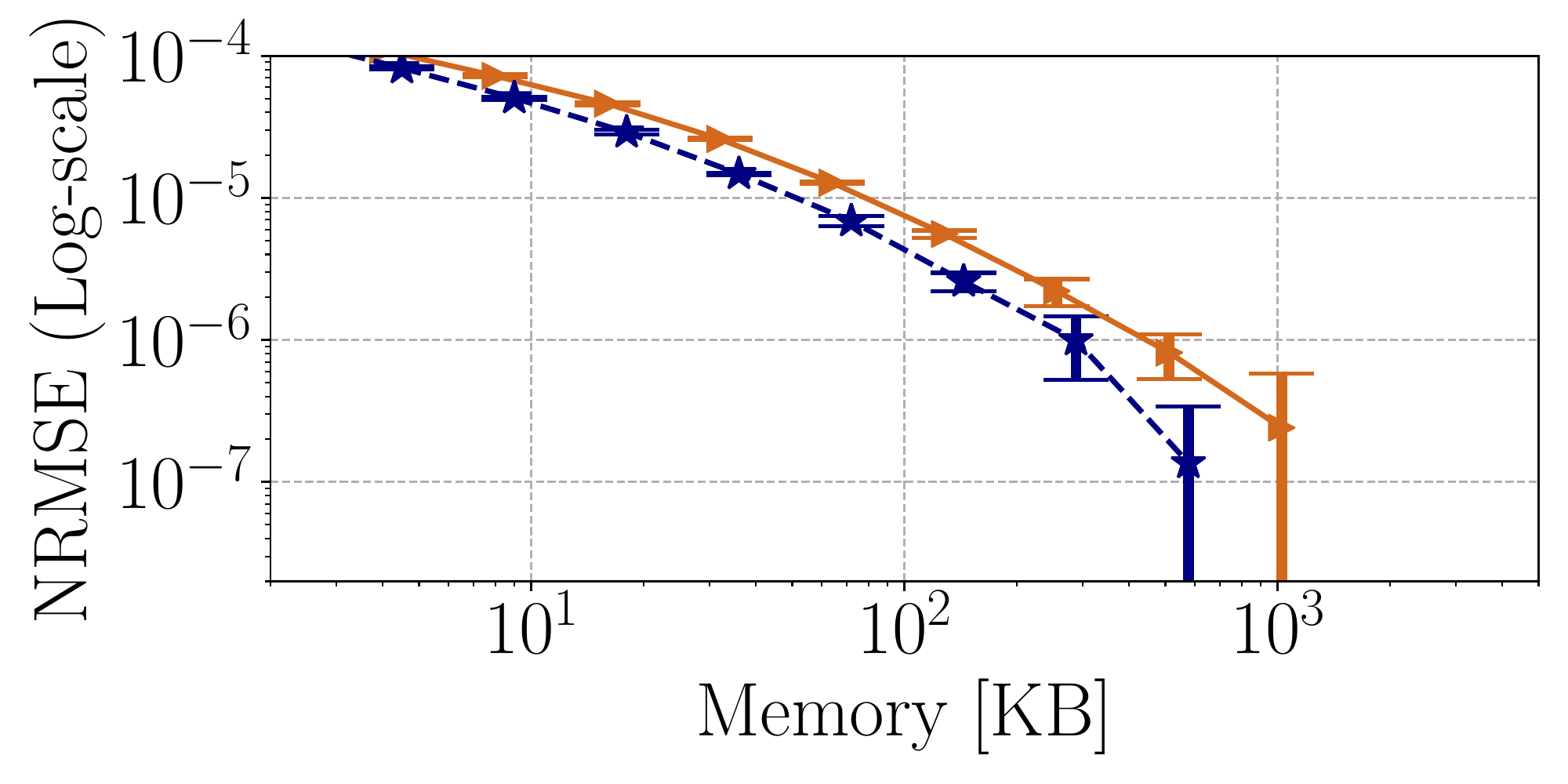}}
    \ifdefined\sigmodSubmission\vspace*{-1mm}\fi
    \caption{\small Accuracy of SALSA CS for the real datasets.
    }\label{fig:cs}
    \ifdefined\sigmodSubmission\vspace*{-4mm}\fi
\end{figure} 

\rev{
\textbf{Can one simply use small counters?}
In our evaluation, SALSA starts from $s=8$ bit counters. We now compare SALSA with a baseline sketch that uses $8$-bit or $16$-bit counters (as proposed, e.g., in~\cite{qi2019cuckoo}). In such a sketch, the counter is only incremented if it does not overflow (i.e., its value is bounded by $2^{b}-1$ for $b$-bit counters). 
We show that such solutions cannot capture the sizes of the heavy hitters -- elements whose frequency is at least a $\phi$ fraction -- which are often considered the most important elements~\cite{SpaceSavingIsTheBest2010}.
First, we show (Figure~\ref{fig:smallCountersVsPhi}) that when estimating heavy hitters, even with a loose definition of $\phi=10^{-4}$, it is best to use $32$-bit counters for CMS. Similarly, as shown in Figure~\ref{fig:smallCountersVsLength}, when the measurement is longer than 10M elements, the $16$-bit variant becomes less accurate. Figure~\ref{fig:smallCounters} is depicted for Zipfian trace with skew=1, and we observed similar behavior for other traces, thresholds, and memory footprints.
}

\textbf{Comparison with Pyramid Sketch and ABC: }
\CRdel{We continue our evaluation by comparing SALSA with two state-of-the-art solutions. 
Pyramid Sketch~\cite{PyramidSketch} extends overflowing counters into a hierarchical sketch-based structure, and ABC~\cite{gong2017abc} copes with overflowing counters by stealing bits from their neighbors. }
We used the authors' original implementations for both Pyramid Sketch and ABC.
We present results for CMS on the NY18 and CH16 datasets; similar results are obtained for  \mbox{additional sketches and workloads.}

As shown in Figure~\ref{fig:pspeed1}, and Figure~\ref{fig:pspeed2}, Pyramid Sketch and SALSA are about 20\% slower than the baseline, while ABC is about 75\% slower. Intuitively, the slowdown is expected, as all these algorithms bring additional complexity to the baseline. ABC is significantly slower due is additional encoding overheads as its bit-borrowing technique does not allow byte-alignment for counters, forcing it to make additional bitwise operations for reading and updating counters.
 
{ In terms of the NRMSE metric (\ref{fig:nrmsep1} and \ref{fig:nrmsep2}), SALSA achieves the best results. Next is the baseline following by Pyramid Sketch, and ABC. 
The on-arrival NRMSE metric gives more weight to the frequent elements, and is more sensitive to larger errors than AAE and ARE. 
 %
Our results indicate that SALSA is also more accurate than Pyramid Sketch and the baseline in terms of AAE and ARE for the entire memory range. Note that Pyramid Sketch is better than the baseline in the memory range 0.5MB-2MB, which is the range it is optimized for according to the paper. 
ABC is slightly more accurate than SALSA for small memory sizes but less accurate than SALSA for large memory sizes, and is comparable in between. Our conclusion is that SALSA is the best in the NRMSE metric and is competitive in the AEE, and ARE, and speed metrics. }

\definecolor{darkgreen}{rgb}{0.1, 0.3, 0.23}
\textbf{Understanding the differences:} 
\CR{Our first observation, is that the AAE and ARE metrics are not suitable when estimating the size of the heavy hitters, which are often considered the most important elements~\cite{SpaceSavingIsTheBest2010}. This is because both metrics give equal weight to all items, making impact of the error of the largest ones vanish due to averaging. 
This is evident from Figure~\ref{fig:smallCounters}, in which the leftmost point ($\phi=10^{-8}$) corresponds to the AAE and ARE metrics (as all items are accounted for). As shown, in such a case, using $8$-bit counters yields lower error rates. Nevertheless, such a solution cannot count beyond the value of $255$, which results in excessive error for the heavy hitters (e.g., $\phi=10^{-3}$).
\ifdefined\sigmodSubmission
In fact, our full version~\cite{} shows that in these metrics, for CMS, and this dataset, it is better to estimate all sizes as $0$ without performing any measurement.}
\else
In fact, as we show in Appendix~\ref{app:stupidMetrics}, for CMS, and this dataset, it is better to estimate all sizes as $0$ without performing any measurement.
\fi
To illustrate the differences that make Pyramid Sketch and ABC competitive in the AAE/ARE metrics, but not in NRMSE, we visualize the errors of estimating individual element frequencies. 
We sampled one random element from each possible frequency to reduce clutter.
The results, showing in Figure~\ref{fig:scatter}, demonstrate the differences between the algorithms.
SALSA has a low error-variance and is consistently more accurate than the Baseline.
In contrast, Pyramid Sketch (as shown in region \textcolor{darkgreen}{\textbf{A}}) has much higher variance, as elements whose counters overflow share the \emph{most significant} bits with other elements. 
ABC, as evident in region \textcolor{magenta}{\textbf{B}}, has a high error on heavy hitters as its counters can at most double in size by combining with their neighbors. We configured ABC to start with $8$-bit counters as suggested by the authors, limiting its estimation to at most $2^{13}-1$ (as three bits are spent on overhead). While one could use larger counters, it would decrease their number and diminish the benefit over the baseline sketch.
To conclude, both ABC and Pyramid Sketch have elements with high estimations errors, making them less attractive for (Mean Square Error)-like metrics.

\textbf{L1 Sketches:}
We proceed by testing the impact that SALSA (with $s=8$ bit counters) has on the accuracy and speed of L1 sketches, such as CMS and CUS. The results, depicted in Figure~\ref{fig:CMSandCUS}, show that SALSA CMS is substantially more accurate (roughly requiring half the space for the same error) than the Baseline for the NY18, CH16, and YouTube datasets. For Univ2, SALSA's improvement is less noticeable, and due to its encoding overheads, the tradeoff is not statistically significant. SALSA CUS is better than the Baseline on all traces, and often requires half the space for a given error.

SALSA's accuracy comes at the cost of additional operations that are required to maintain the counter layout. We measured SALSA to be 17\%-23\% slower than the corresponding Baseline variants, but can nonetheless handle 10-17.5 million elements per second on a single core, which is sufficient to support the high link rate forwarding at modern large-scale clusters, such as Google, which is estimated at 9M packets per second (see ~\cite[Sec. 3.2]{eisenbud2016maglev}) . 
We note that by combining SALSA with estimators (Section~\ref{sec:estimatorsEval}), we can make faster counter sketches.
We conclude that SALSA offers an \mbox{appealing accuracy to space tradeoff.}

\textbf{Count Sketch:}
Next, we evaluate SALSA for Count Sketch, whose L2 guarantee is important for low-skew workloads and more complex algorithms such as UnivMon.
As shown in Figure~\ref{fig:cs}, SALSA offers statistically significant improvement for the NY18, CH16, and YouTube datasets. For Univ2, the accuracy improvement is offset by the encoding overhead, and it is not clear which variant is better.

\textbf{UnivMon: }
We use SALSA CS to extend the Universal Monitoring (UnivMon) sketch that supports estimating a wide variety of functions of the frequency vector. 
Our experiment includes estimating the element size entropy and $F_p$ moments, for $0\le p \le 2$.  
The result in Figure~\ref{fig:univmon} indicates that SALSA improves the accuracy of both tasks. Interestingly, for entropy estimation, we observe that SALSA's accuracy (and variance) improve when using smaller ($s=2$ or $s=4$ bit) counters. When using a large amount of memory, SALSA has roughly the same accuracy as the baseline, as both hit a bottleneck in the size of the sketches' heaps (set to $100$ elements, as in \mbox{the implementation of~\cite{Nitro}).}

For estimating $F_p$ moments, we measure similar accuracy for small values of $p$ while SALSA improves the accuracy for large $p$ values. To explain this, notice that the element size estimates mainly affect the $F_p$ for large $p$, while for $p\approx 0$, the value is determined primarily by the cardinality.

\begin{figure}[t]
    \centering
    \hspace*{-2mm}
    \subfloat[Entropy Estimation]
    { \includegraphics[width =0.5\columnwidth]
    {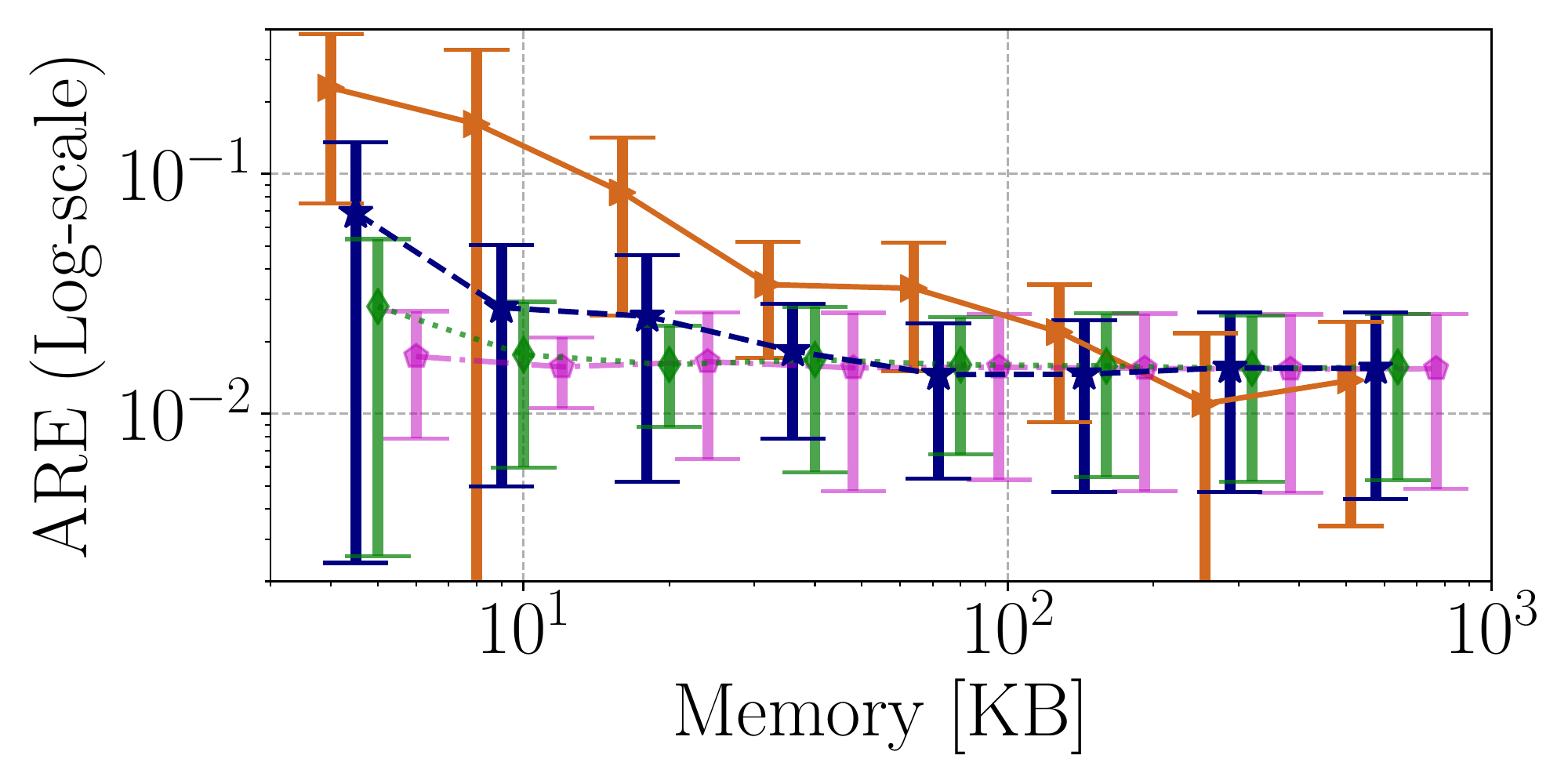}}
    \hspace*{-2mm}
    \subfloat[{\scriptsize Frequency Moment Estimation (400KB)}]
    { \includegraphics[width =0.5\columnwidth]
    {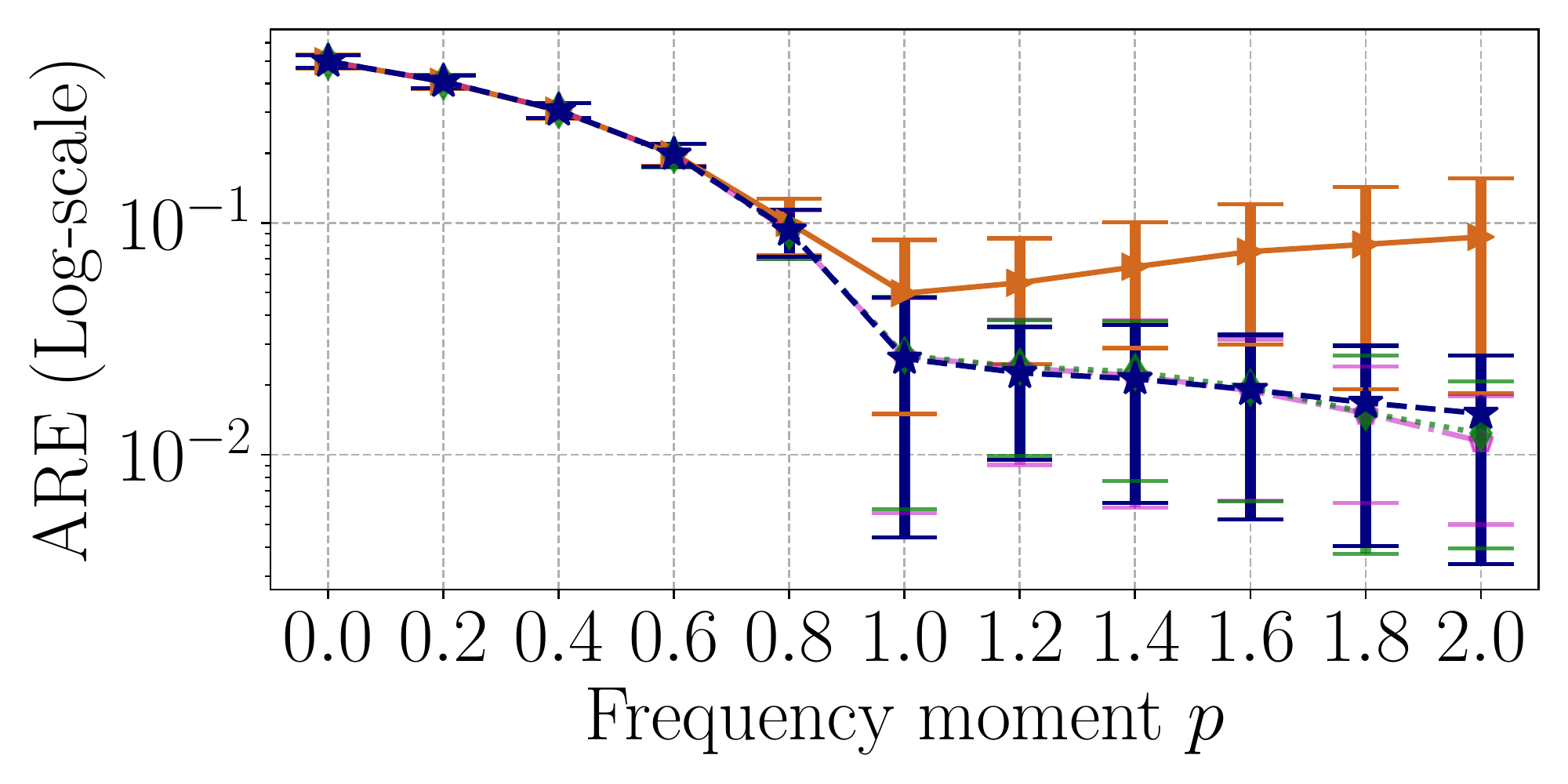}}\\
    {\vspace*{-0mm}\includegraphics[width =1.04\columnwidth]
    {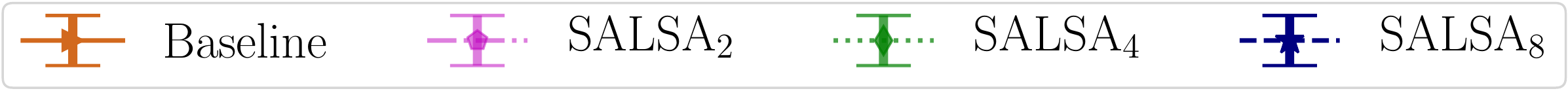}}
    \ifdefined\sigmodSubmission\vspace*{-5mm}\fi
    \caption{\small Accuracy of SALSA UnivMon for the NY18 dataset.
    }\label{fig:univmon}
    \ifdefined\sigmodSubmission\vspace*{-3mm}\fi
\end{figure} 

\textbf{Cold Filter: }
We extend Cold Filter by replacing the second-stage CUS (denoted CM-CU in the original paper) algorithm with our SALSA variant. The results in Figure~\ref{fig:coldfilter} use the AAE and ARE metrics suggested by its authors~\cite{ColdFilter}.
The results show that SALSA saves up to 50\% of the space for a similar error.  However, the improvement is more evident when the allocated memory is small, as in these cases the second stage algorithm plays a significant role. When the memory size is large (compared to the measurement length), the first-stage algorithm handles most of the flows, and improving the second-stage CUS algorithm yields marginal benefits. 
We observed negligible differences in the processing speed, which is expected as many elements only touch the first stage and do not reach the second. 
We also tested Cold Filter versus its SALSA variant using the NRMSE metric; there, SALSA yields even larger accuracy gains. However, Cold Filter's aggregation buffer needs to be drained upon query, which negates its speedup potential \mbox{in the on-arrival model.
}

\begin{figure}[]
    \centering
    \hspace*{-2mm}
    { \includegraphics[width =0.49\columnwidth]
    {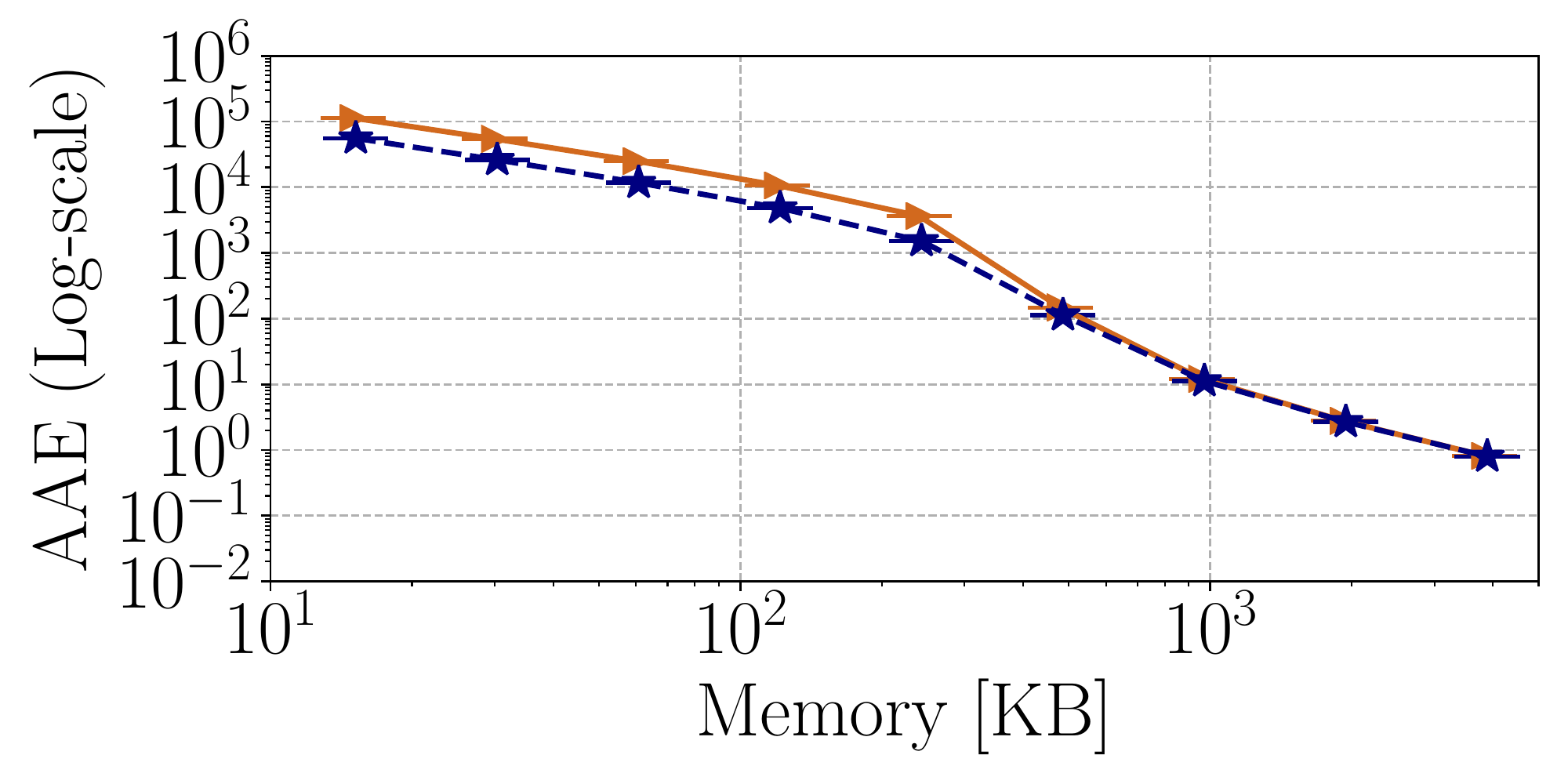}}
    \hspace*{-2mm}
    { \includegraphics[width =0.49\columnwidth]
    {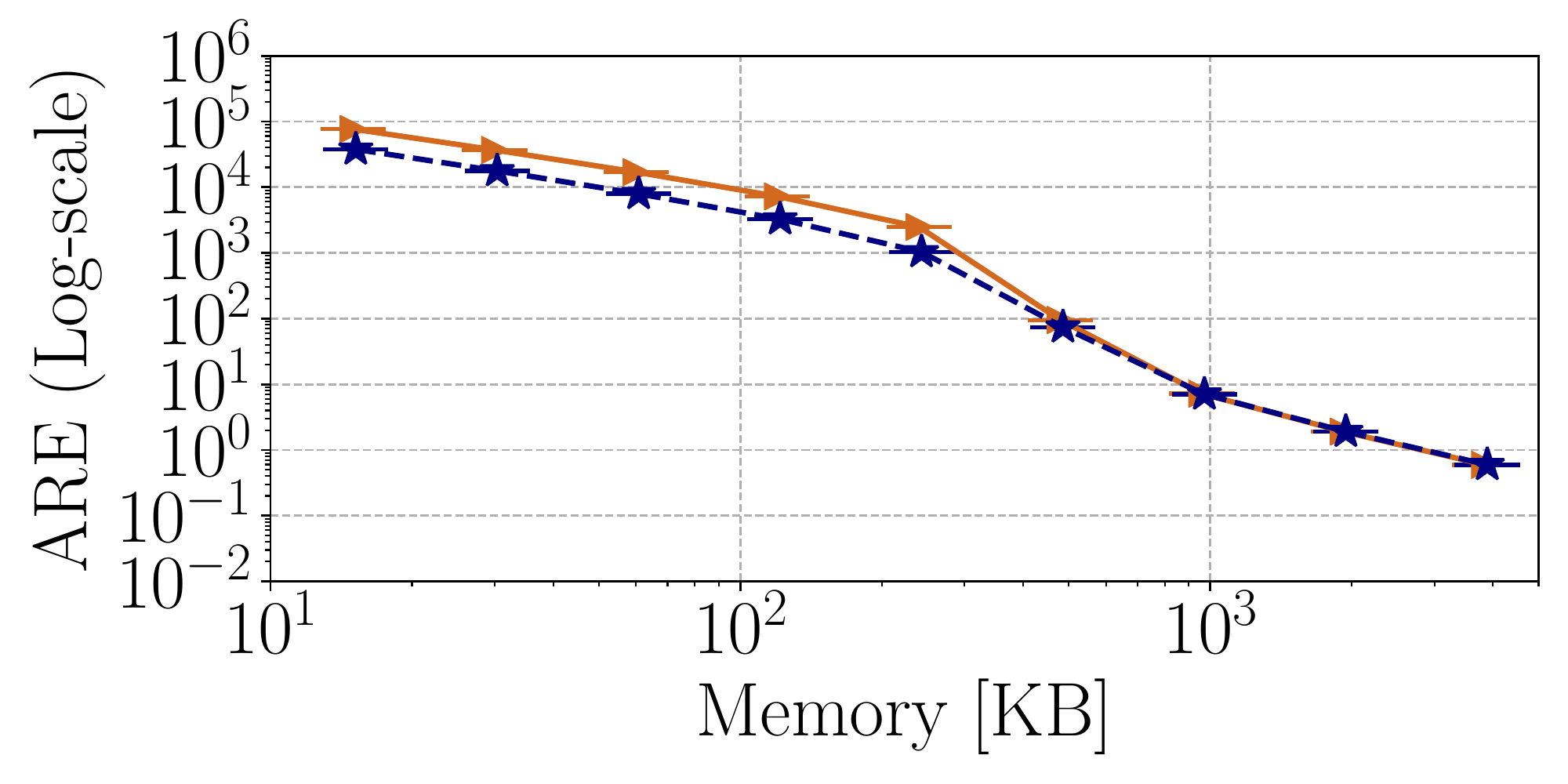}}\\
    {\ifdefined\sigmodSubmission\vspace*{-1mm}\fi\includegraphics[width =.6\columnwidth]
    {figs/Salsa/cs/cs_legend_cropped}}
    \ifdefined\sigmodSubmission\vspace*{-1mm}\fi
    \caption{Accuracy of SALSA Cold Filter for the NY18 dataset.
    }\label{fig:coldfilter}
    \vspace*{-5mm}
\end{figure} 
\countdistinct{\textbf{Count Distinct and Heavy Hitters using Count Min:}
We evaluate the performance of SALSA CMS on additional applications such as counting distinct elements and estimating the size of the heavy hitters.
As shown in the count, distinct results (Figures~\ref{fig:cmsCDandHH}(a)-(c)), neither SALSA CMS nor the Baseline are effective with low memory footprints. This is because no counters remain zero-valued, and the Linear Counting estimator fails. 
Nevertheless, SALSA CMS can work with less memory (4.5MB for NY18 and 1.125MB for CH16) and reduce the estimation error when the Baseline does produce estimates. Intuitively, Linear Counting with $w$ buckets can count up to $w\ln w$ elements, so the number of elements in the datasets (6.5M for NY18 and 2.5M for CH16) \mbox{imposes a lower bound on the amount of space needed.}}
\countdistinctapp{\textbf{Heavy Hitters using Count Min:}}
We evaluate the accuracy for estimating the \emph{heavy hitters} (elements with frequency of at least $\phi\cdot N$) frequencies, while varying $\phi$ between $10^{-4}$ and $10^{-2}$ as in~\cite{SpaceSavingIsTheBest2010}.
SALSA CMS is more accurate, especially for small values of $\phi$. The lower improvement for large $\phi$ values can be explained as $\phi\cdot N\ge 2^{16}$ for $\phi\ge \frac{2^{16}}{98\cdot 10^6}\approx 7\cdot 10^{-4}$, which means that all such heavy hitters cause their counters to merge to $32$ bits (the same as the Baseline). The plot of Figure~\ref{fig:hhNY18} stops around $\phi\approx 3.16\cdot 10^{-4}$ as no element in the NY18 dataset has frequency larger \mbox{than $5.62\cdot 10^{-4}\cdot N\approx 551K$  packets.}
 
\begin{figure}[t]
    \centering
    \hspace*{-2mm}
\ifdefined\icdeSubmission
\else    
    \subfloat[NY18 Dataset]
    { \includegraphics[width =0.16\textwidth]
    {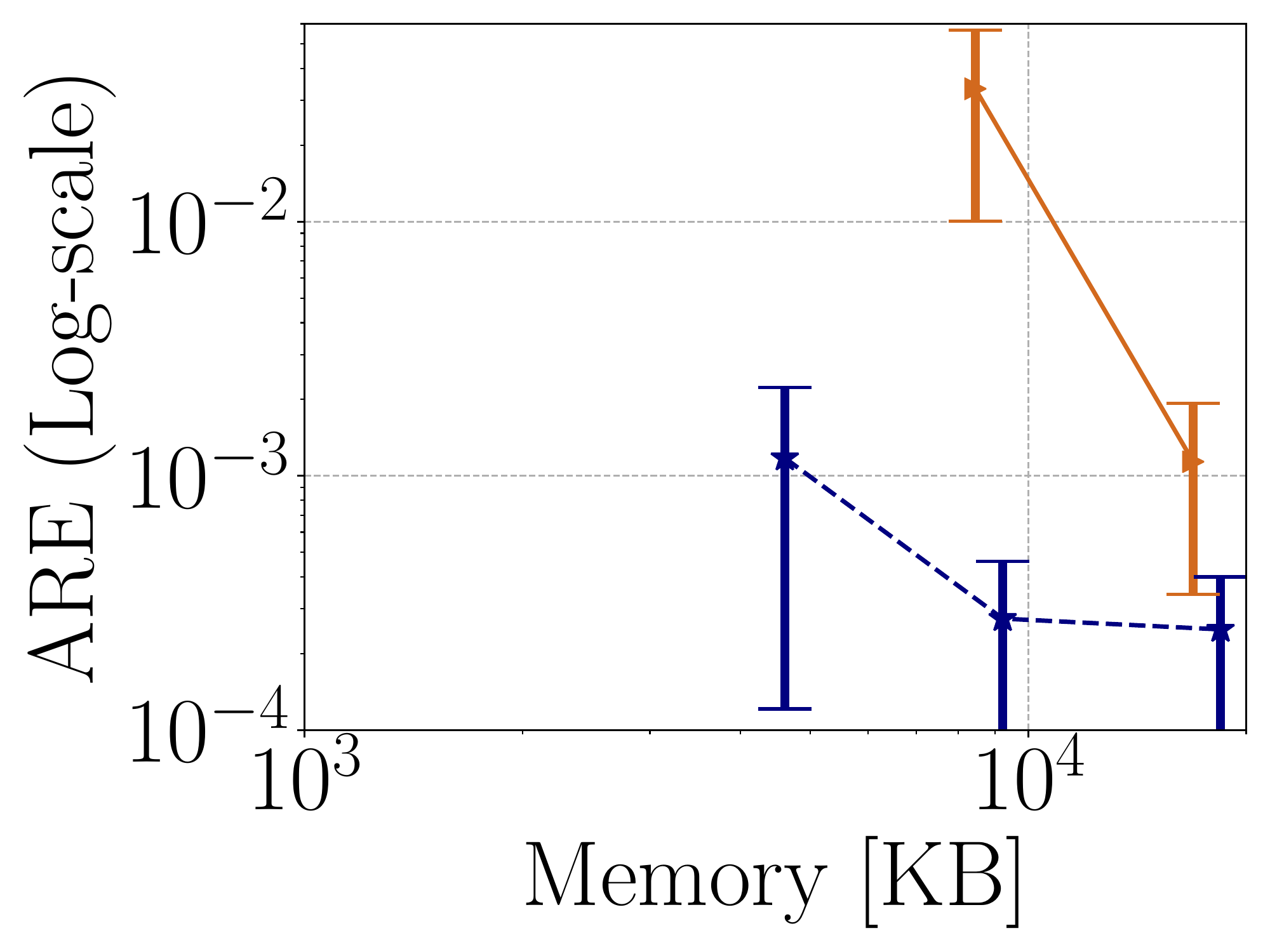}}
    \hspace*{-2mm}
    \subfloat[CH16 Dataset]
    { \includegraphics[width =0.16\textwidth]
    {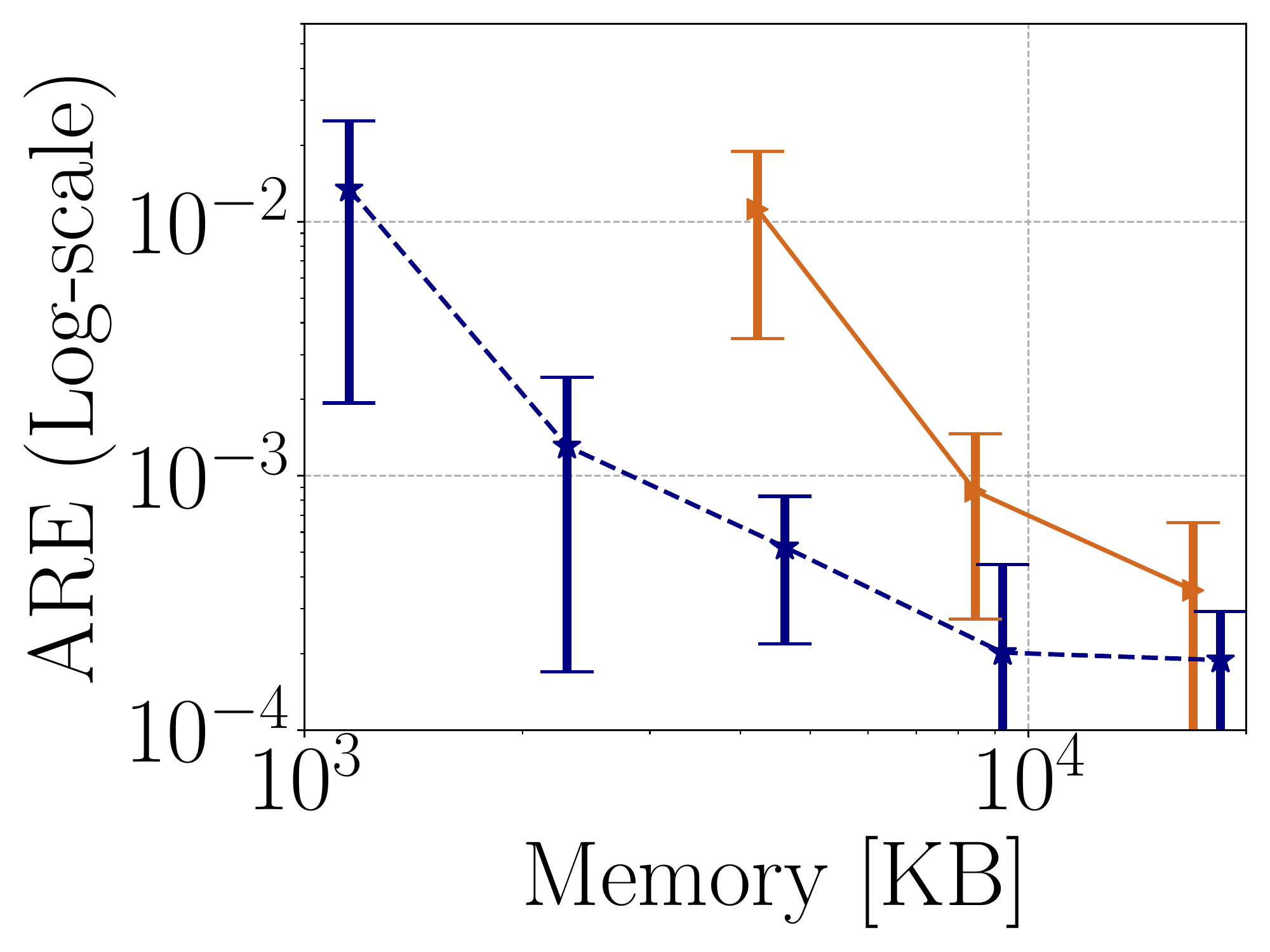}}
    \subfloat[Zipf (8MB)]
    { \includegraphics[width =0.16\textwidth]
    {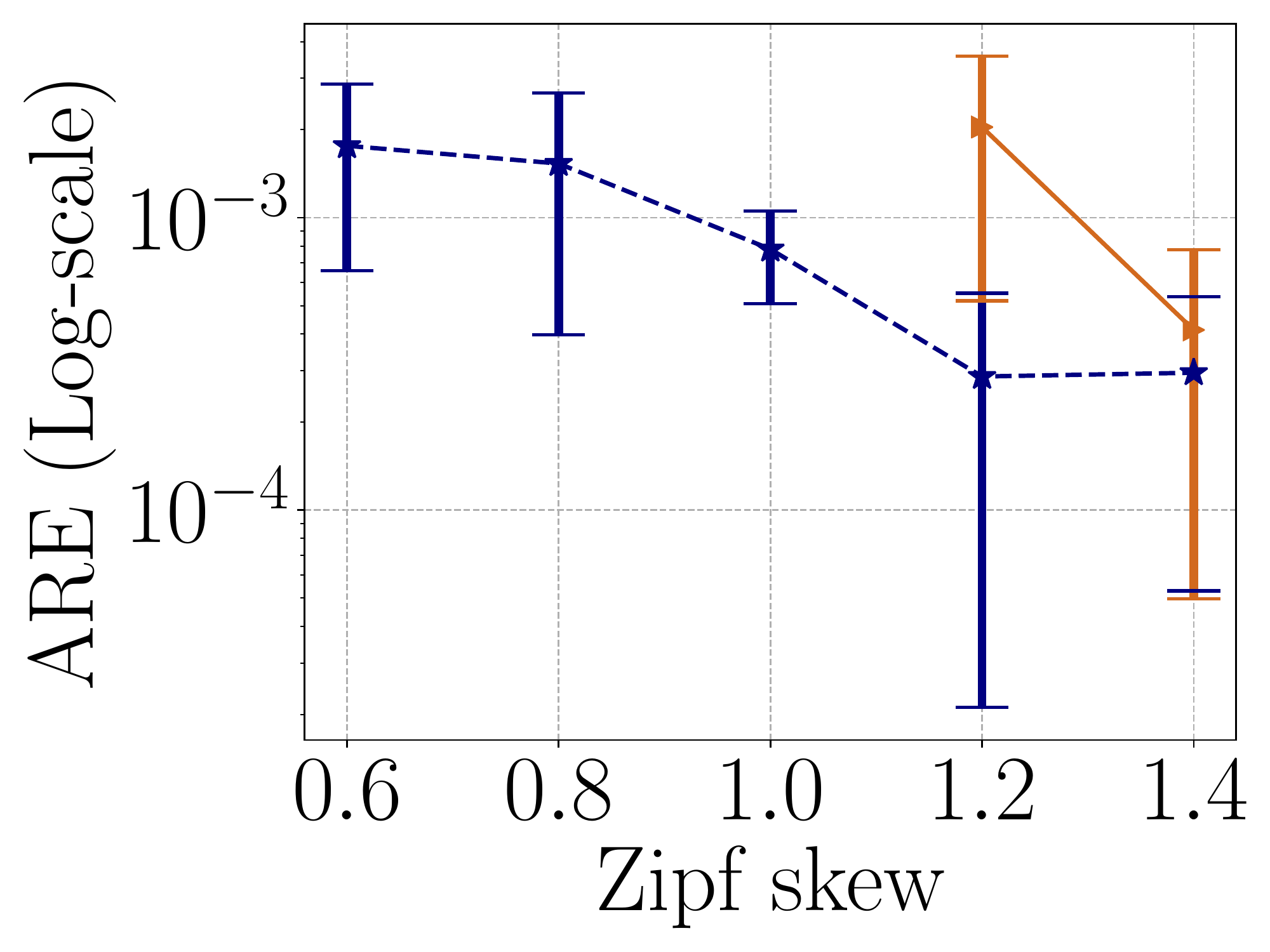}}    \\
    {\vspace*{-0mm}\includegraphics[width =.6\columnwidth]
    {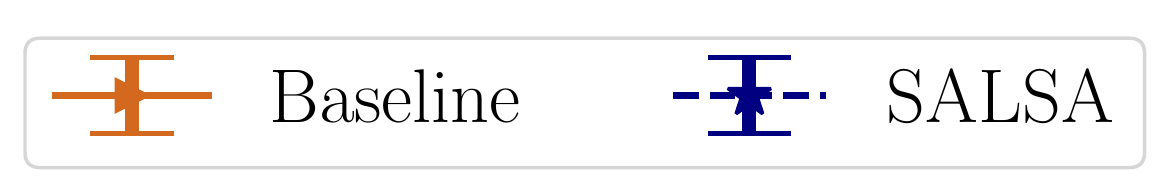}}\\
\fi    

    \ifdefined\sigmodSubmission\vspace*{-3mm}\fi
    \subfloat[NY18 Dataset]
    { \includegraphics[width =0.16\textwidth]
    {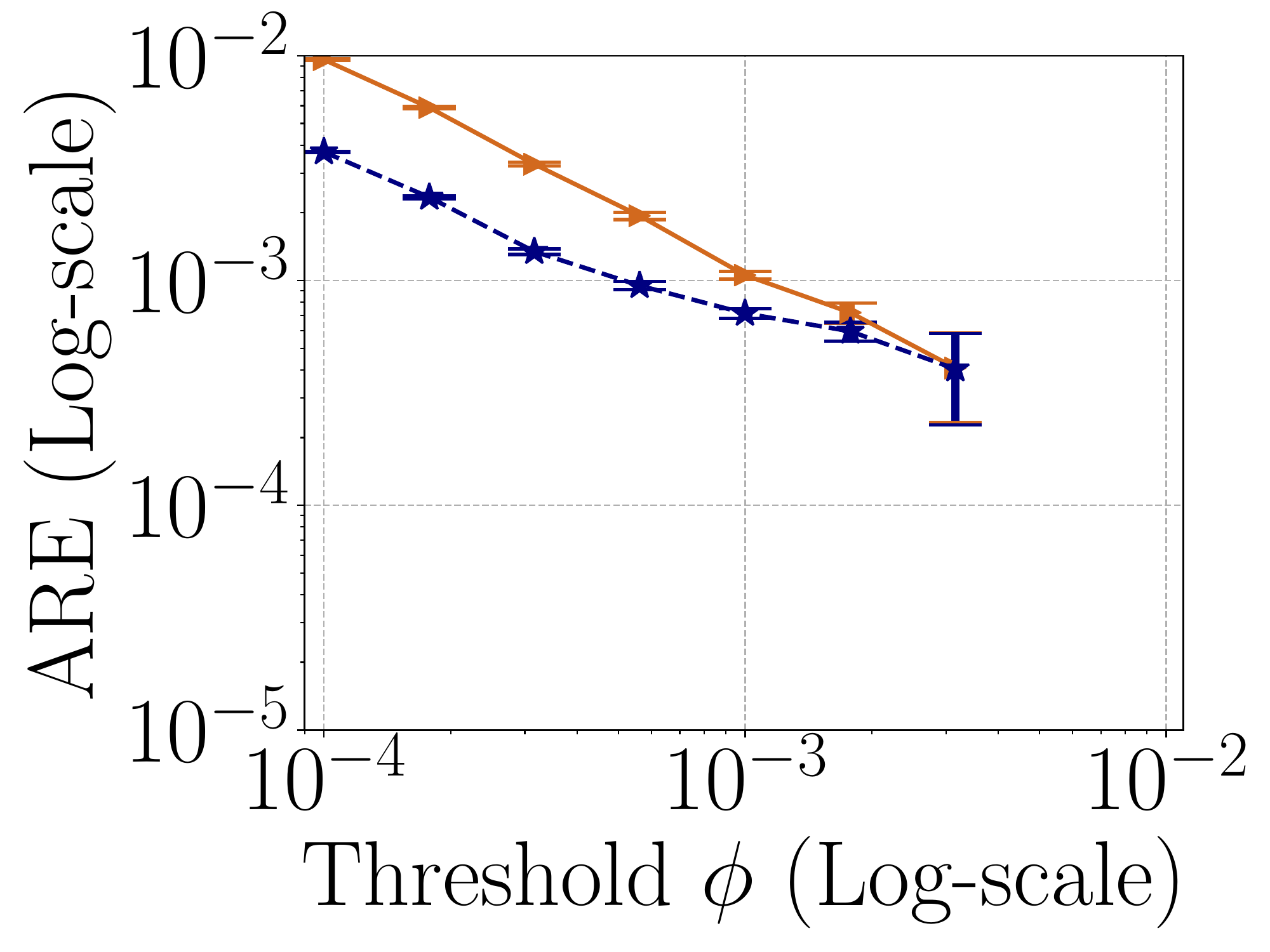}\label{fig:hhNY18}}
    \hspace*{-2mm}
    \subfloat[CH16 Dataset]
    { \includegraphics[width =0.16\textwidth]
    {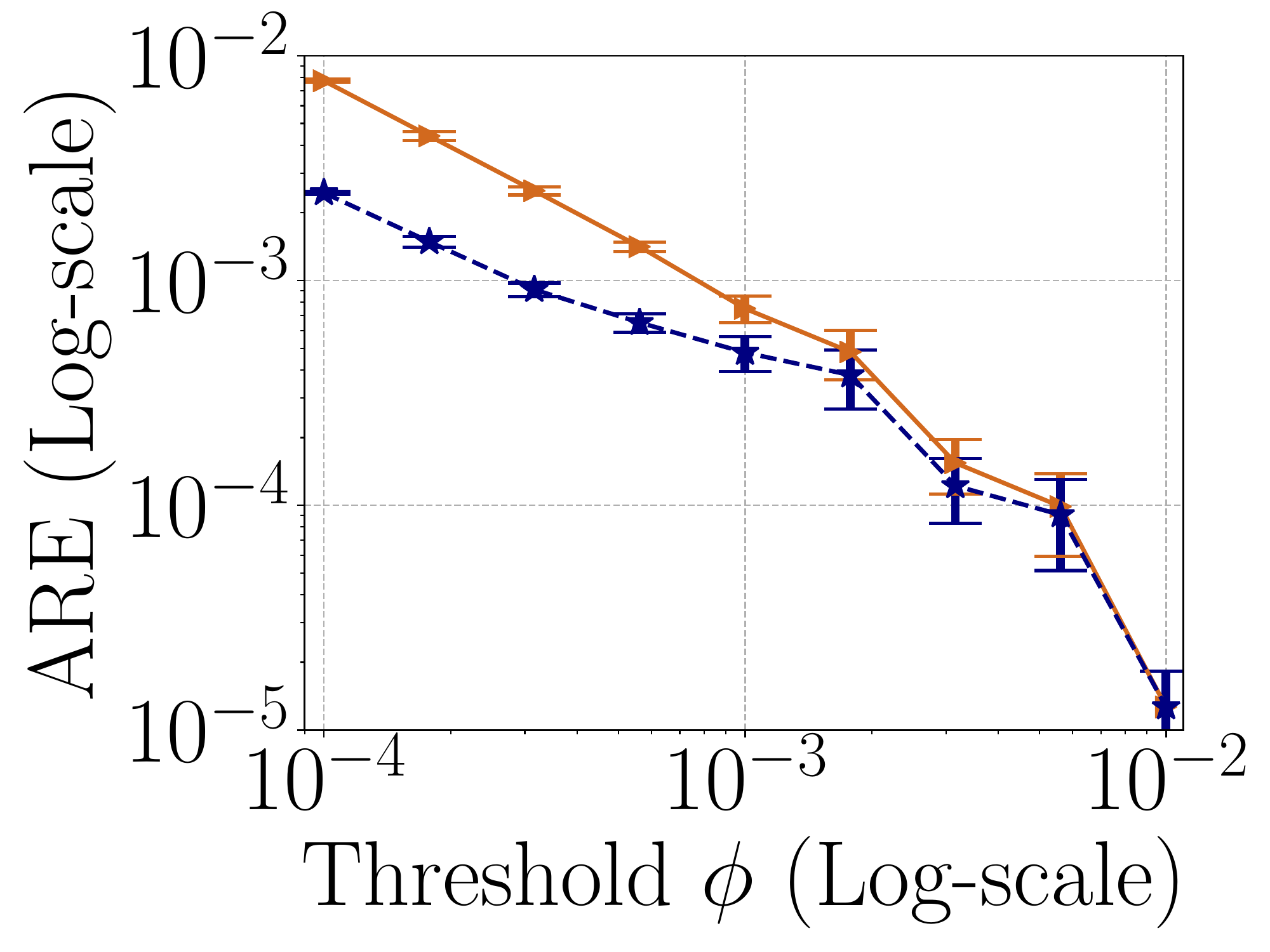}}
    \subfloat[Zipf ($\phi=10^{-4}$)]
    { \includegraphics[width =0.16\textwidth]
    {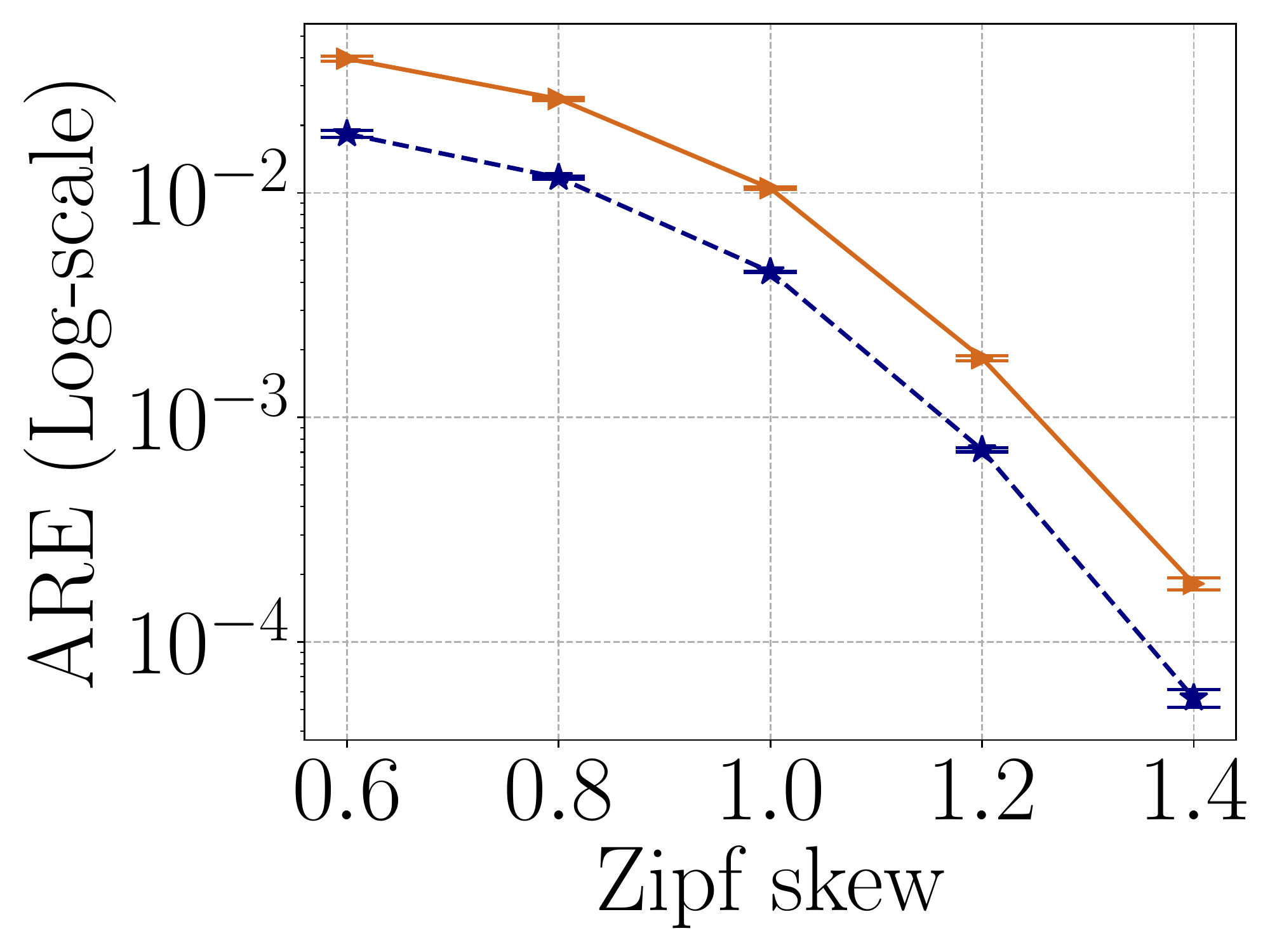}}    
\ifdefined\icdeSubmission
\\
    {\vspace*{-0mm}\includegraphics[width =.6\columnwidth]
    {figs/Salsa/cs/cs_legend_cropped}}
    \ifdefined\sigmodSubmission\vspace*{-2mm}\fi
    \caption{\small Accuracy of SALSA CMS on estimating the size of heavy hitters with 2MB.}
\else     
    \ifdefined\sigmodSubmission\vspace*{-2mm}\fi
    \caption{\small Accuracy of SALSA CMS on ((a)-(c)) counting distinct \mbox{elements and ((d)-(f)) estimating the size of heavy hitters with 2MB.}}
\fi    
    \label{fig:cmsCDandHH}
    \vspace*{-3mm}
\end{figure} 

\begin{figure}[t]
\ifdefined\sigmodSubmission\vspace*{-3mm}\fi
    \centering
    \hspace*{-2mm}
    
    \subfloat[Top-$k$, NY18 (640KB)]
    { \includegraphics[width =0.5\columnwidth]
    {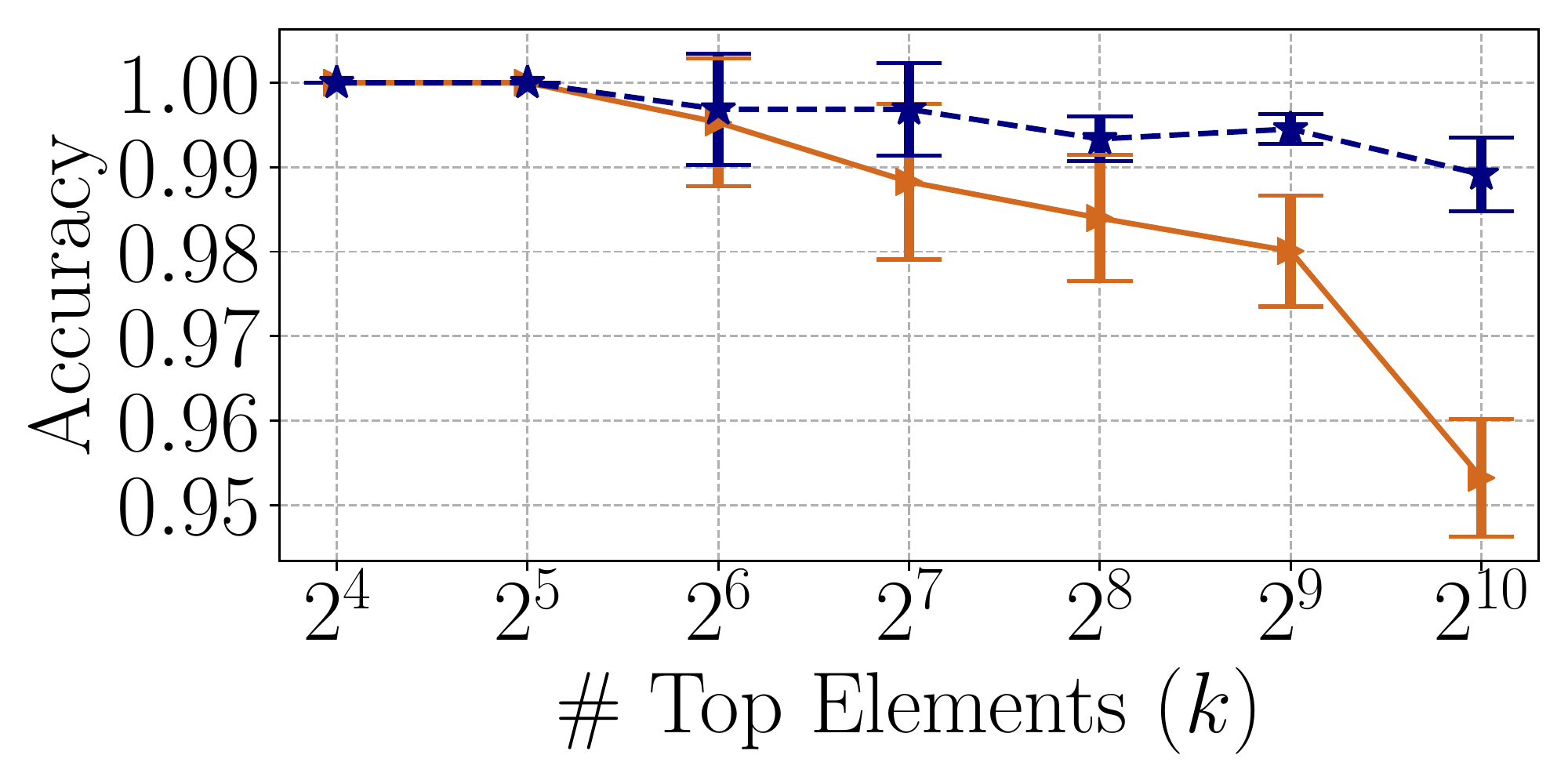}}
    \hspace*{-2mm}
    \subfloat[Top-1024, Zipf (640KB)
    ]
    { \includegraphics[width =0.5\columnwidth]
    {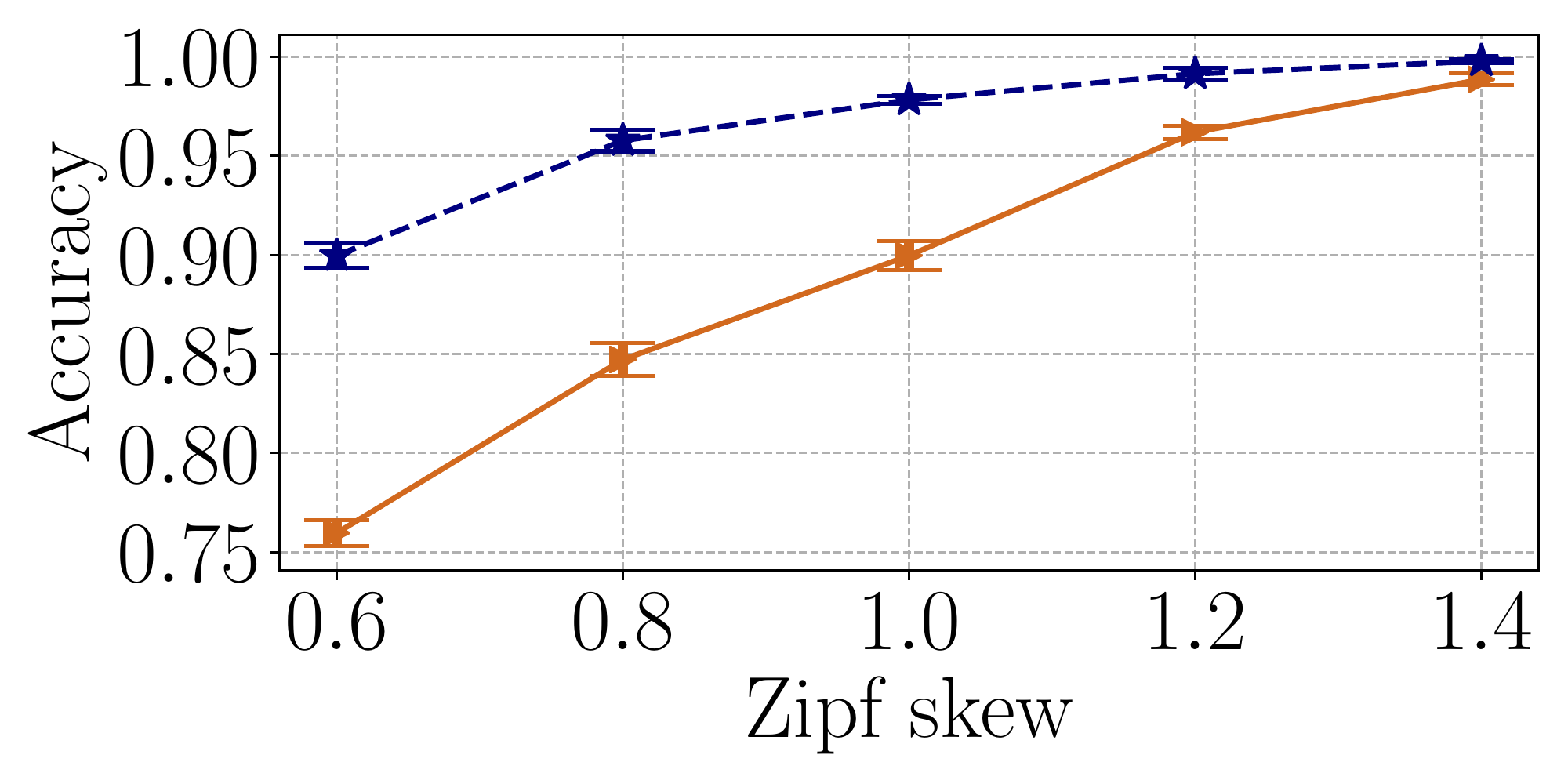}}   
    \hspace*{-1mm}\\
    {\includegraphics[width =0.6\columnwidth]
    {figs/Salsa/cs/cs_legend_cropped}}\\\ifdefined\sigmodSubmission\vspace*{-3mm}\fi
    \subfloat[Change Detection, NY18]
    { \includegraphics[width =0.5\columnwidth]
    {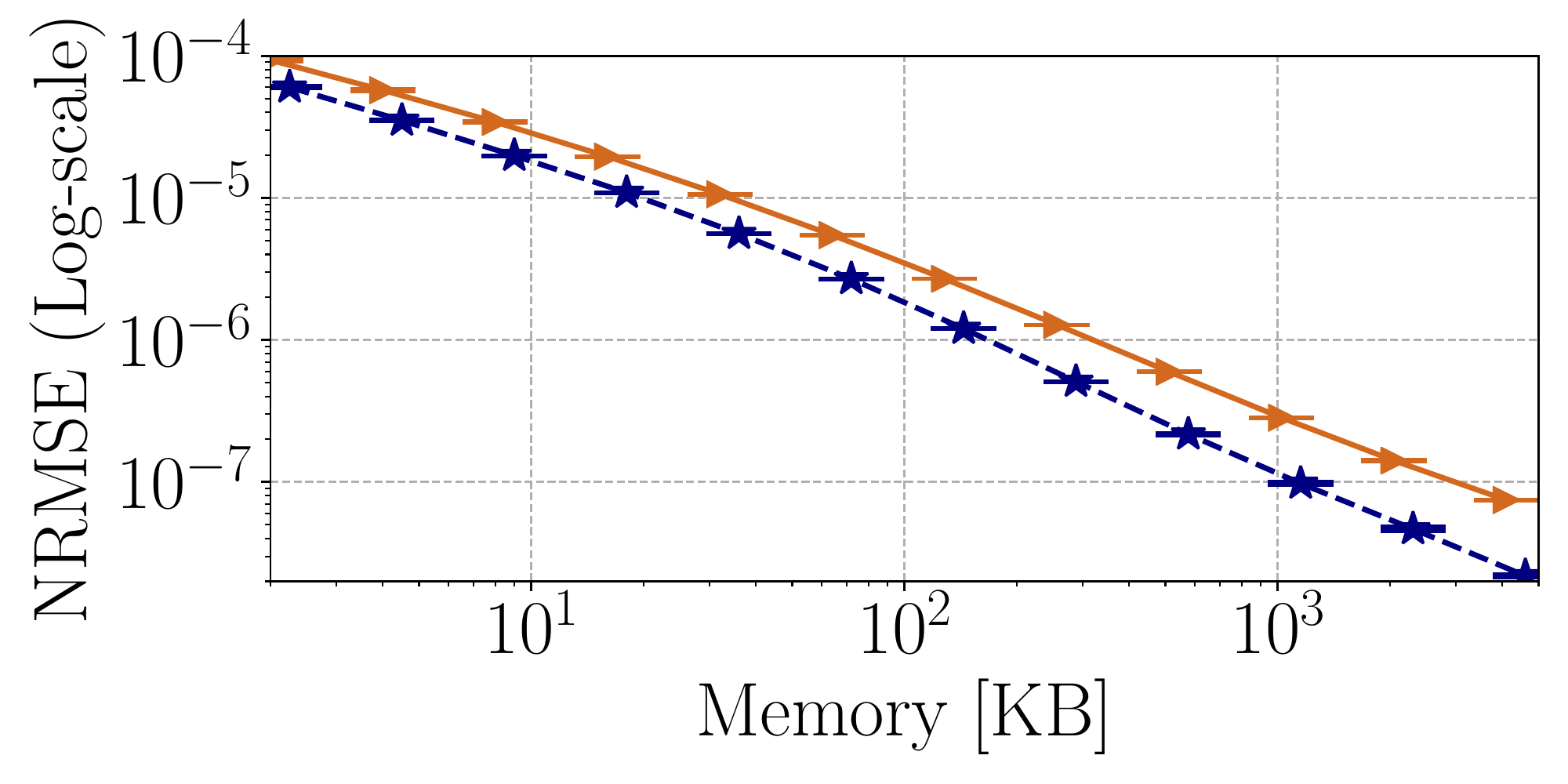}}
    \hspace*{-2mm}
    \subfloat[Change Detection, Zipf (2.5MB)]
    { \includegraphics[width =0.5\columnwidth]
    {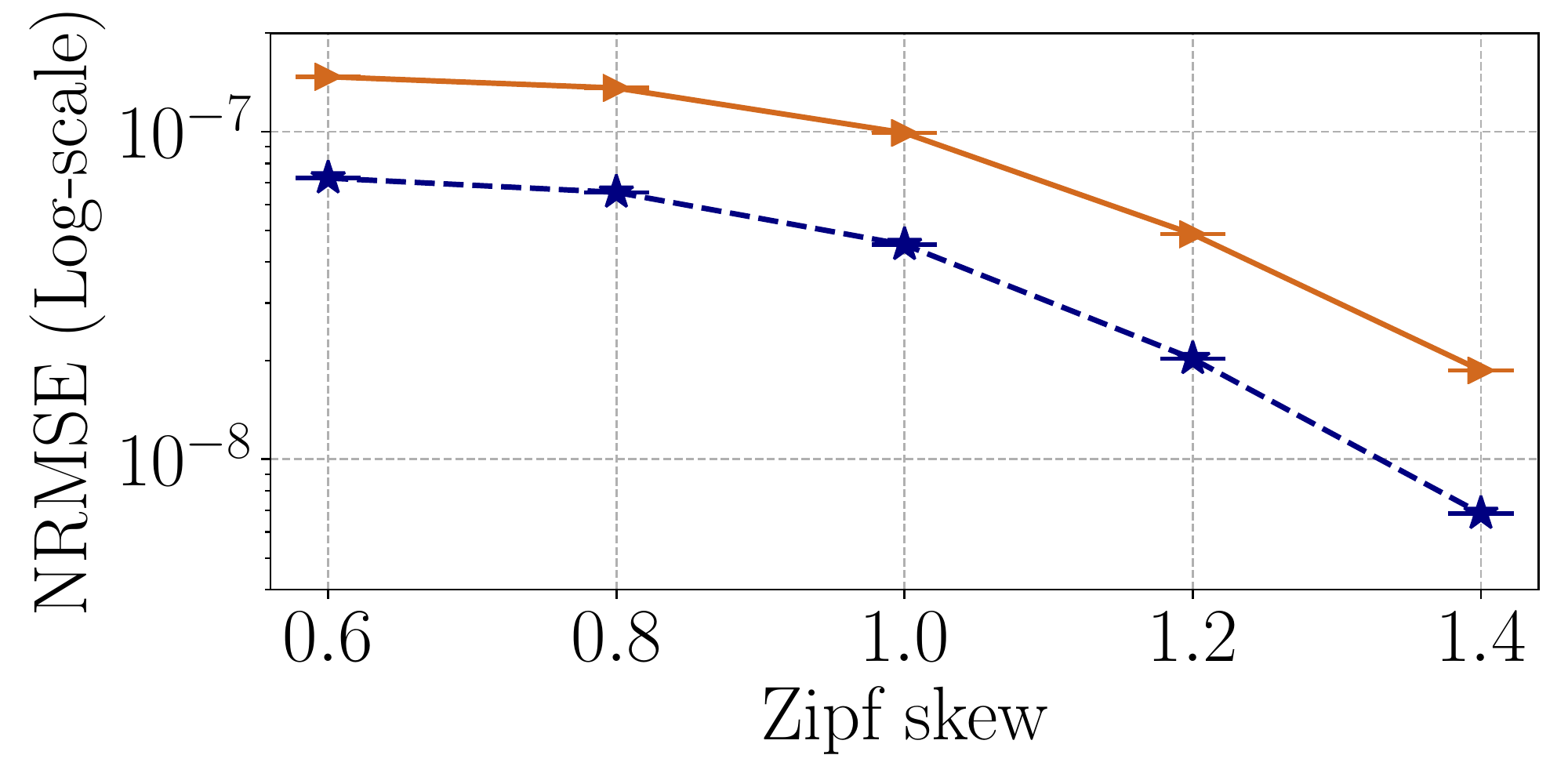}}   
    \hspace*{-1mm}
    \vspace*{-1mm}
    \caption{Accuracy of SALSA CS for the Top-$k$ ((a) and (b)) and Change Detection ((c) and (d)) tasks.
    }\label{fig:csApps}
    \vspace*{-1mm}
\end{figure} 

\textbf{Top-$k$ and Change Detection using Count Sketch:}\label{par:changeDetectionEval}
We also examine SALSA's effect on other \rev{uses of CS} such as Top-$k$ and Change Detection (which requires Turnstile support). For Top-$k$, our experiments indicate that using sufficient memory (e.g., 2MB), the Baseline CS detects the largest elements accurately for reasonable $k$ values. Therefore, we focus on a constrained memory setting (640KB). As shown in Figure~\ref{fig:csApps} (a) and (b), SALSA detects the top-$k$ accurately, especially for large values of $k$ \mbox{and low-skew workloads.}



\rev{We also evaluate SALSA CS on a Change Detection task.} Here, we split the input into two equal-length intervals $A$ and $B$, the algorithm needs to estimate the change in the frequency of an element $x$ between the first and second halves. To that end, we create sketches $s(A)$ and $s(B)$ and the difference sketch $s(A\setminus B)$ as described in Section~\ref{sec:merging}. Intuitively, the frequency difference can be small compared with the frequencies of each interval, and thus directly subtracting the \rev{estimates of $s(A)$ and $s(B)$ could yield a poor result compared to taking the difference sketch} (as the desired error is a fraction of the $L_2$ norm of the frequency difference).
In Figure~\ref{fig:csApps} (c) and (d), we compute the NRMSE error\footnote{Note that this is not on-arrival computation and the results are not comparable with those obtained in figures~\ref{fig:CMSandCUS} and~\ref{fig:cs}.} over the set of elements that appear in either $A$ or $B$. As shown, SALSA provides a statistically significant accuracy improvement \mbox{in all tested memory allocations and dataset skews.}

\begin{figure}[t]
    \centering
    \ifdefined\sigmodSubmission\vspace*{-2mm}\fi
    \hspace*{-2mm}
    \subfloat[Error, NY18]
    { 
    \includegraphics[width =0.5\columnwidth]
    {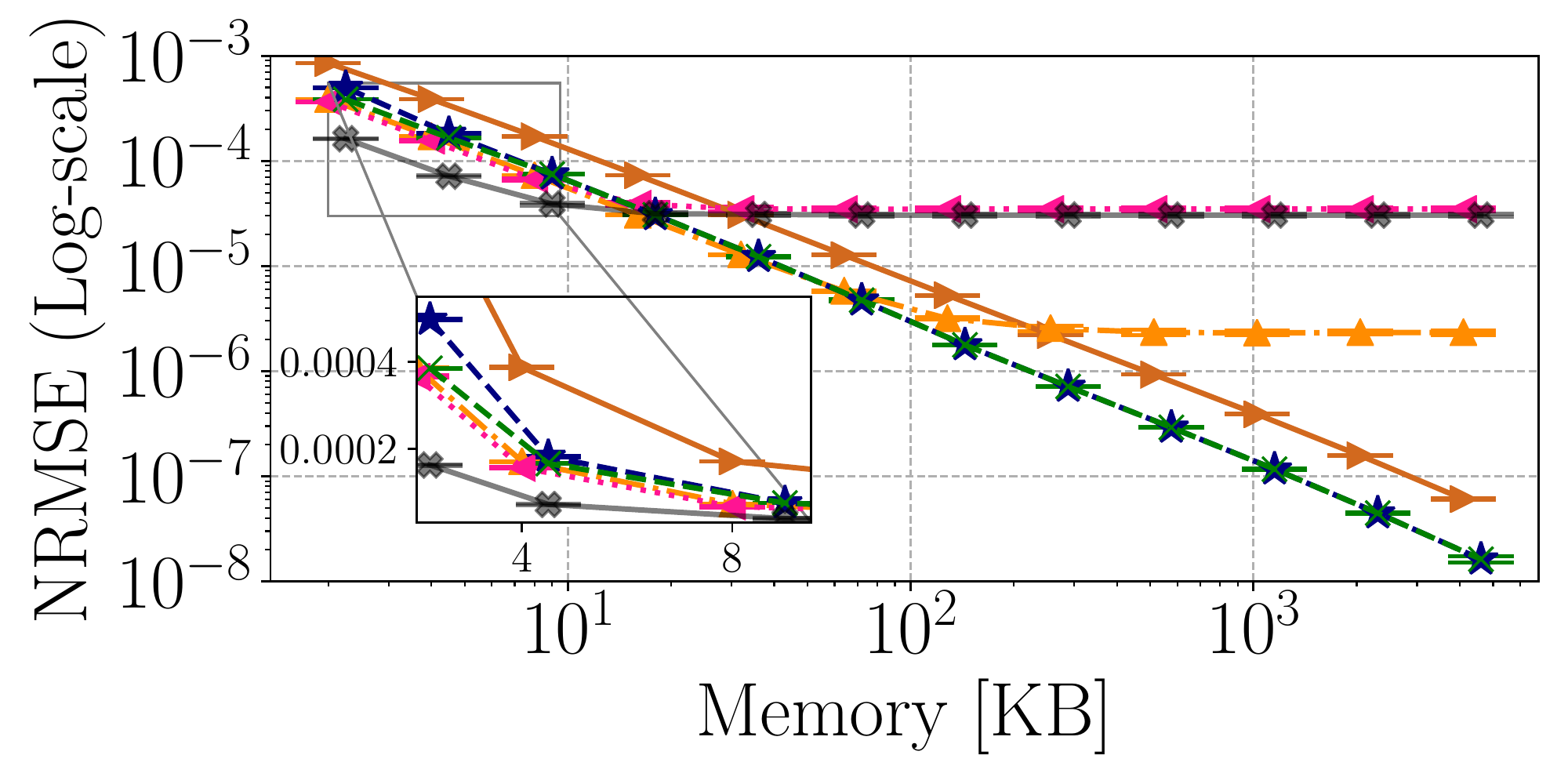}}
    \hspace*{-2mm}
    \subfloat[Error, CH16]
    { \includegraphics[width =0.5\columnwidth]
    {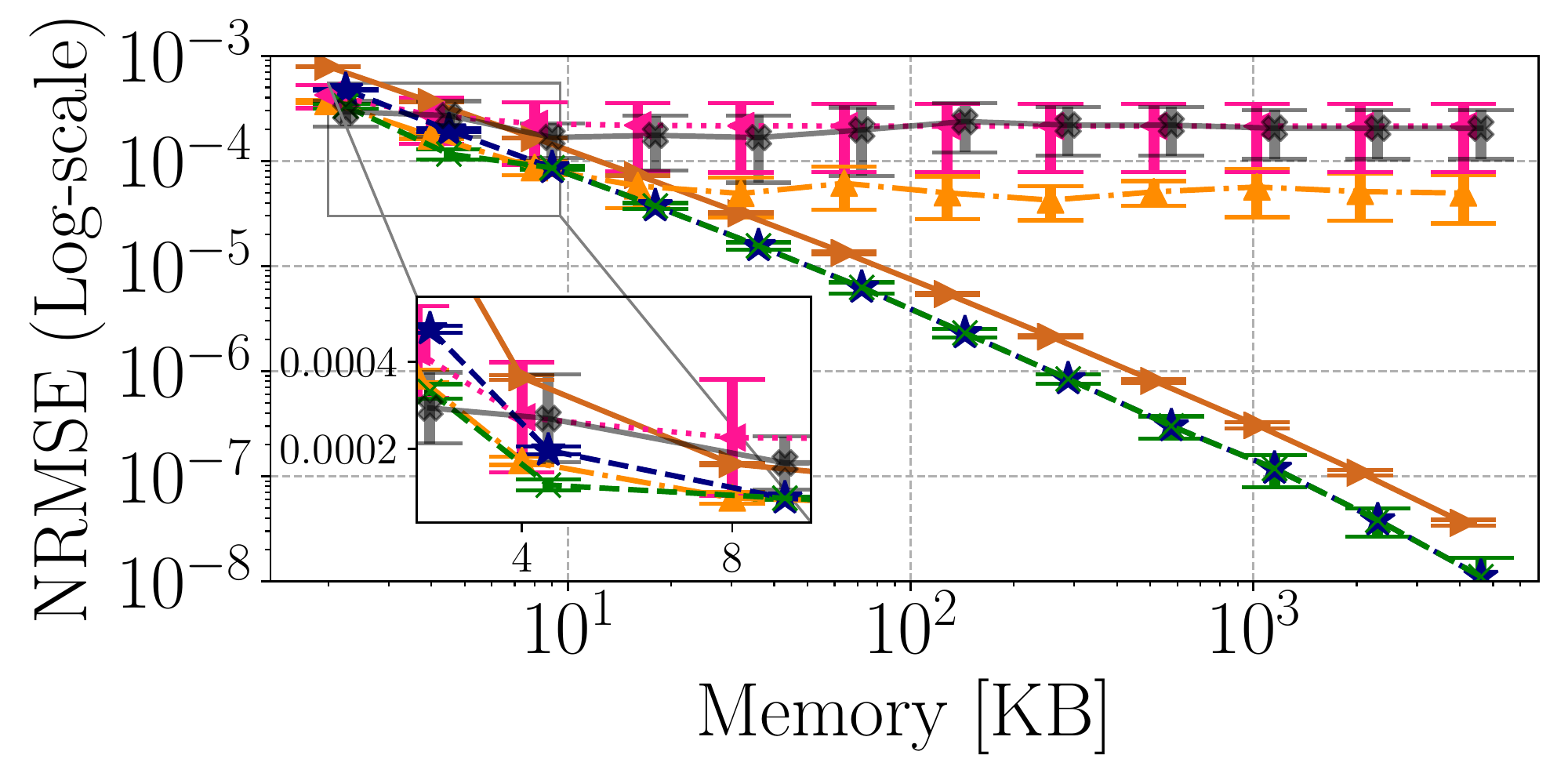}}   
    \hspace*{-1mm}\vspace*{-0mm}\\\ifdefined\fullversion\vspace{3mm}\fi
    {\includegraphics[width =1.04\columnwidth]
    {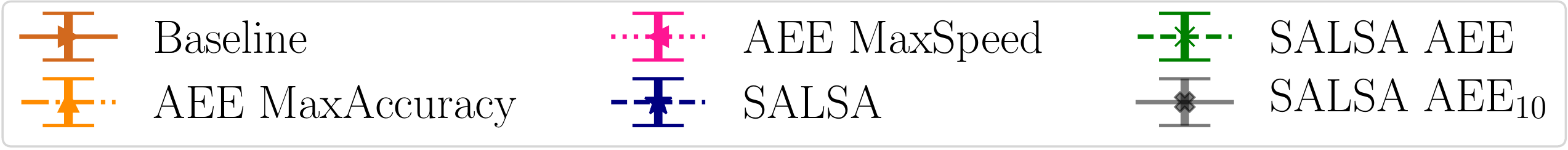}}\\\vspace*{-2.5mm}
    \subfloat[Speed, NY18]
    { \includegraphics[width =0.5\columnwidth]
    {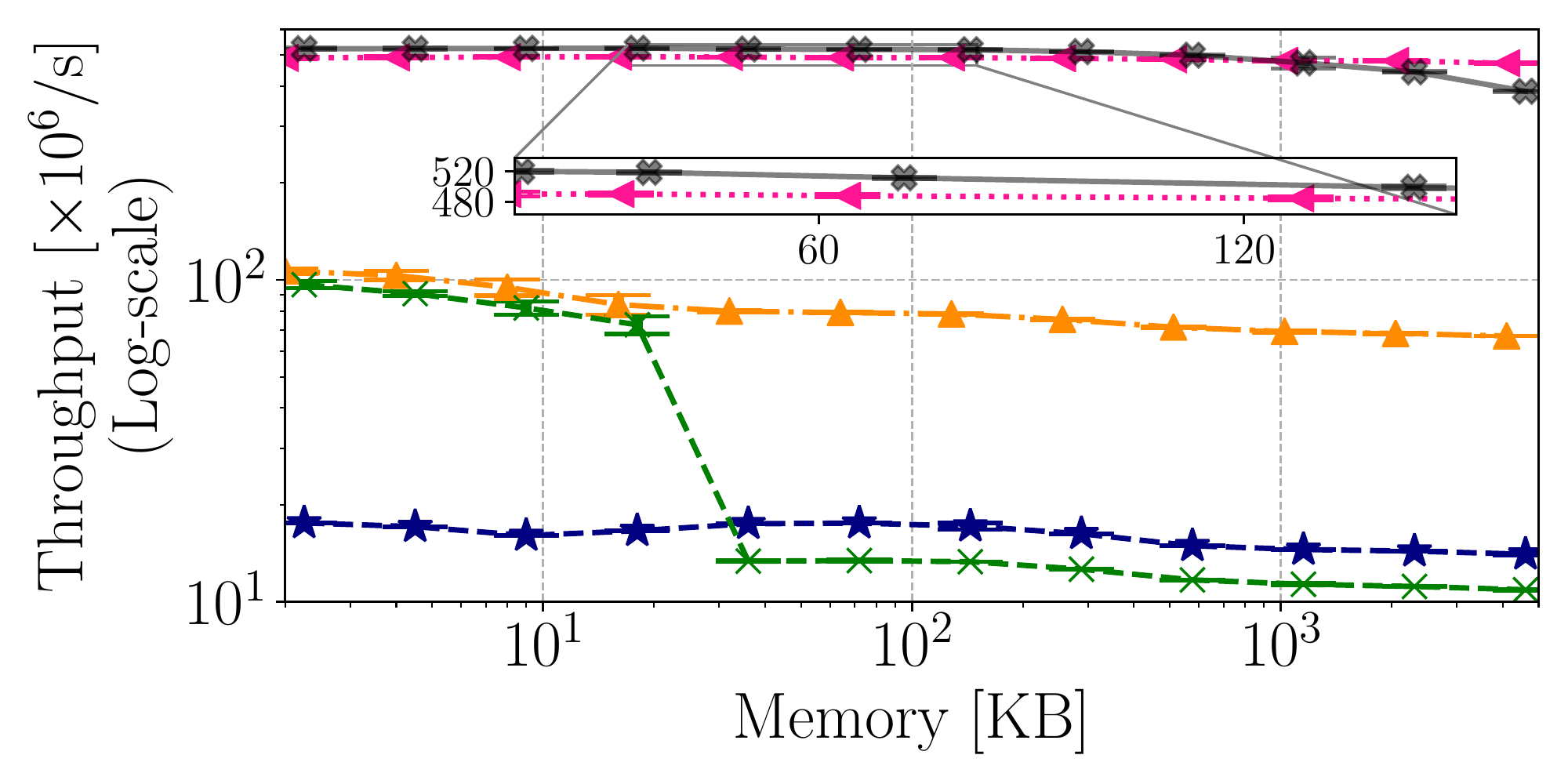}}
    \hspace*{-2mm}
    \subfloat[Speed, CH16]
    { \includegraphics[width =0.5\columnwidth]
    {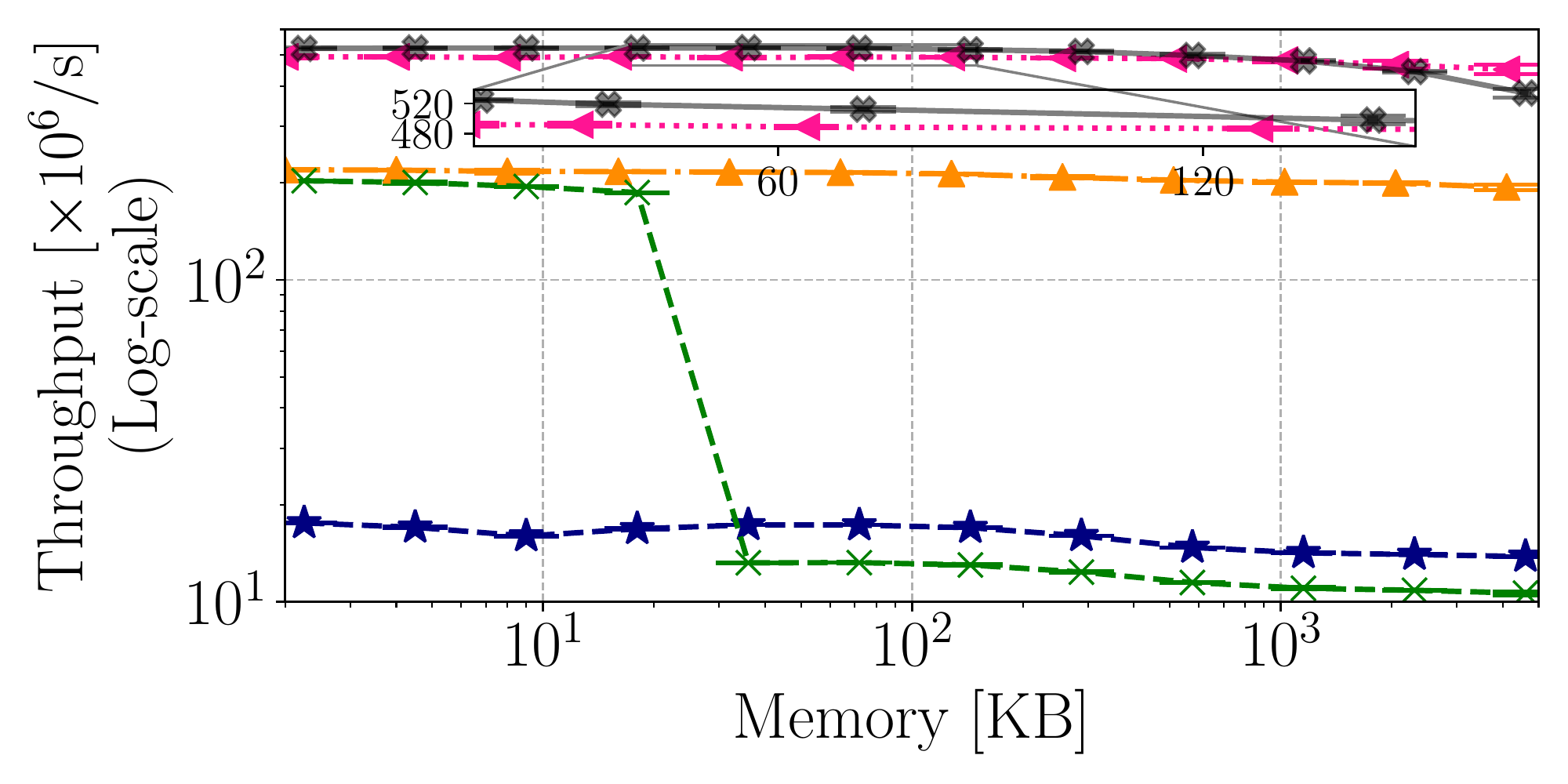}}   
    \hspace*{-1mm}
    \vspace*{-1mm}
    \caption{\small Comparison with Estimator Algorithms (CM sketch).
    }\label{fig:estimators}
    \vspace*{-3mm}
\end{figure}
\textbf{Estimators:}\label{sec:estimatorsEval}
\rev{We now experiment with integrating estimators, specifically AEE ~\cite{Compsketch}, into SALSA CMS.}
Under certain conditions, AEE increases the accuracy and processing speed of the sketch. 
Intuitively, the accuracy \rev{can be increased as the sketch can use more estimators than it would use counters}, and the speed is increased because some packets are ignored without updating the estimators. 
Our estimator-integrated solution SALSA AEE (from Section~\ref{sec:estimators}) optimizes the accuracy by interleaving estimator downsampling and estimator merges. Roughly speaking, SALSA AEE aims to be at least as accurate as the best of SALSA CMS and AEE MaxAccuracy by choosing the best method to cope with each overflow.
Similarly to AEE MaxSpeed, we create a speed-optimized variant called SALSA AEE$_{\mathfrak d}$ that downsamples on the first ${\mathfrak d}$ overflows (and selects whether to merge or downsamples afterward according to the logic presented in Section~\ref{sec:estimators}). This allows the algorithm to reach a sampling rate of $2^{-\mathfrak d}$ and thus obtain \mbox{speedups by reducing hash computations.}

The results, shown in Figure~\ref{fig:estimators}, illustrate that SALSA AEE is always as accurate as SALSA (when SALSA only merges counters), and more accurate for small amounts of memory. For large amounts of memory, SALSA AEE only merges, and therefore its accuracy is identical to SALSA while it is slightly slower due to the added logic.
Compared with AEE MaxAccuracy, SALSA AEE has comparable accuracy for small memory allocations (where it is mostly better to downsample than merge). Further, for large memory allocations (e.g., 100KB or higher), SALSA AEE is more accurate than AEE MaxAccuracy, as in such scenarios it is better to merge than to downsample.
Compared with AEE MaxSpeed,  SALSA AEE$_{10}$ provides improved accuracy (by up to 25\%), especially for small amounts of memory, while also being faster (by up to 7\%), \mbox{except when using large space (2MB+ in this experiment).}

\begin{figure}[t]
    \centering
    \ifdefined\sigmodSubmission\vspace*{-4mm}\fi
    \hspace*{-2mm}
    \subfloat[Error, NY18]
    { \includegraphics[width =0.5\columnwidth]
    {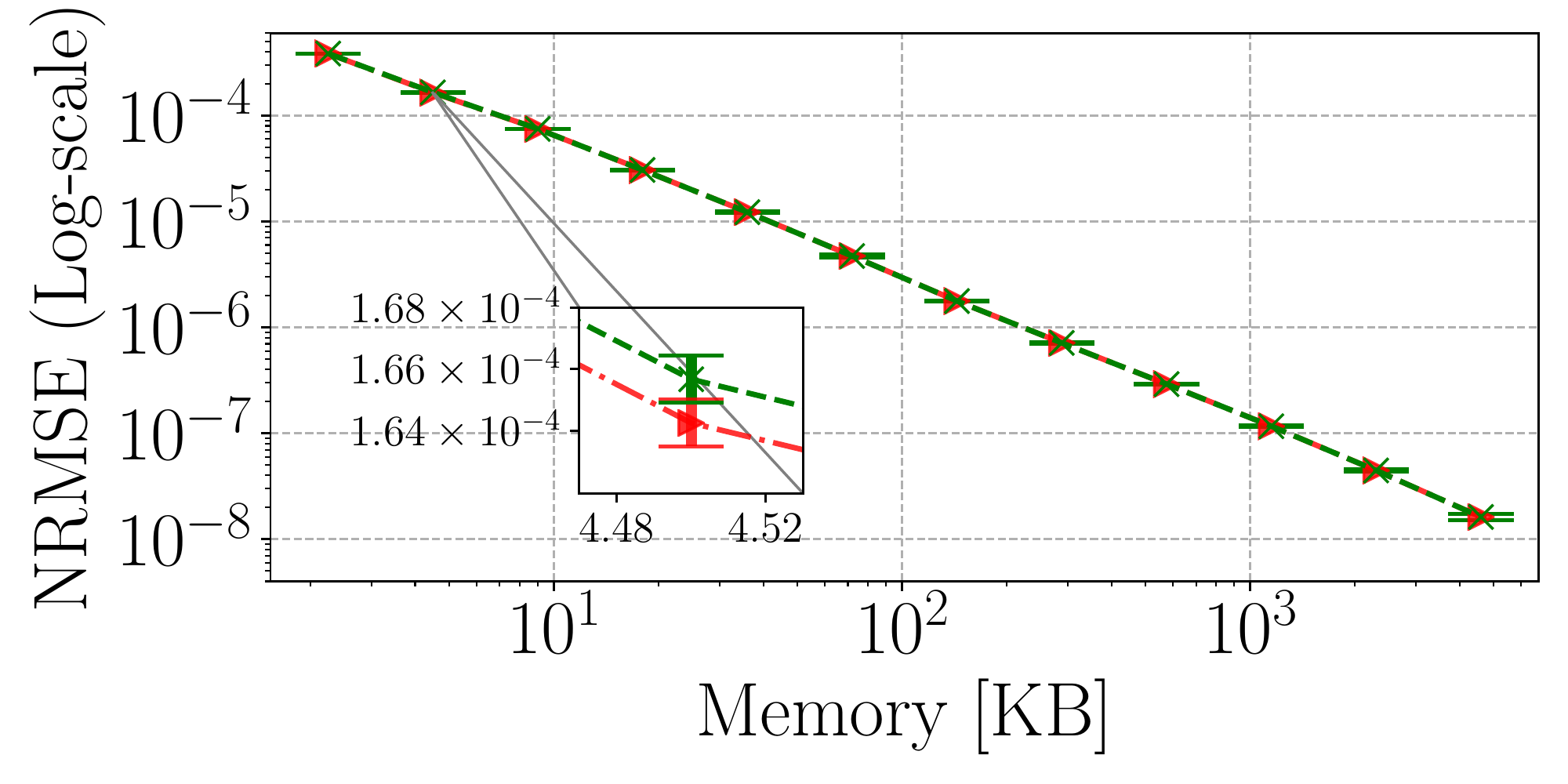}}
    \hspace*{-2mm}
    \subfloat[Error, CH16]
    { \includegraphics[width =0.5\columnwidth]
    {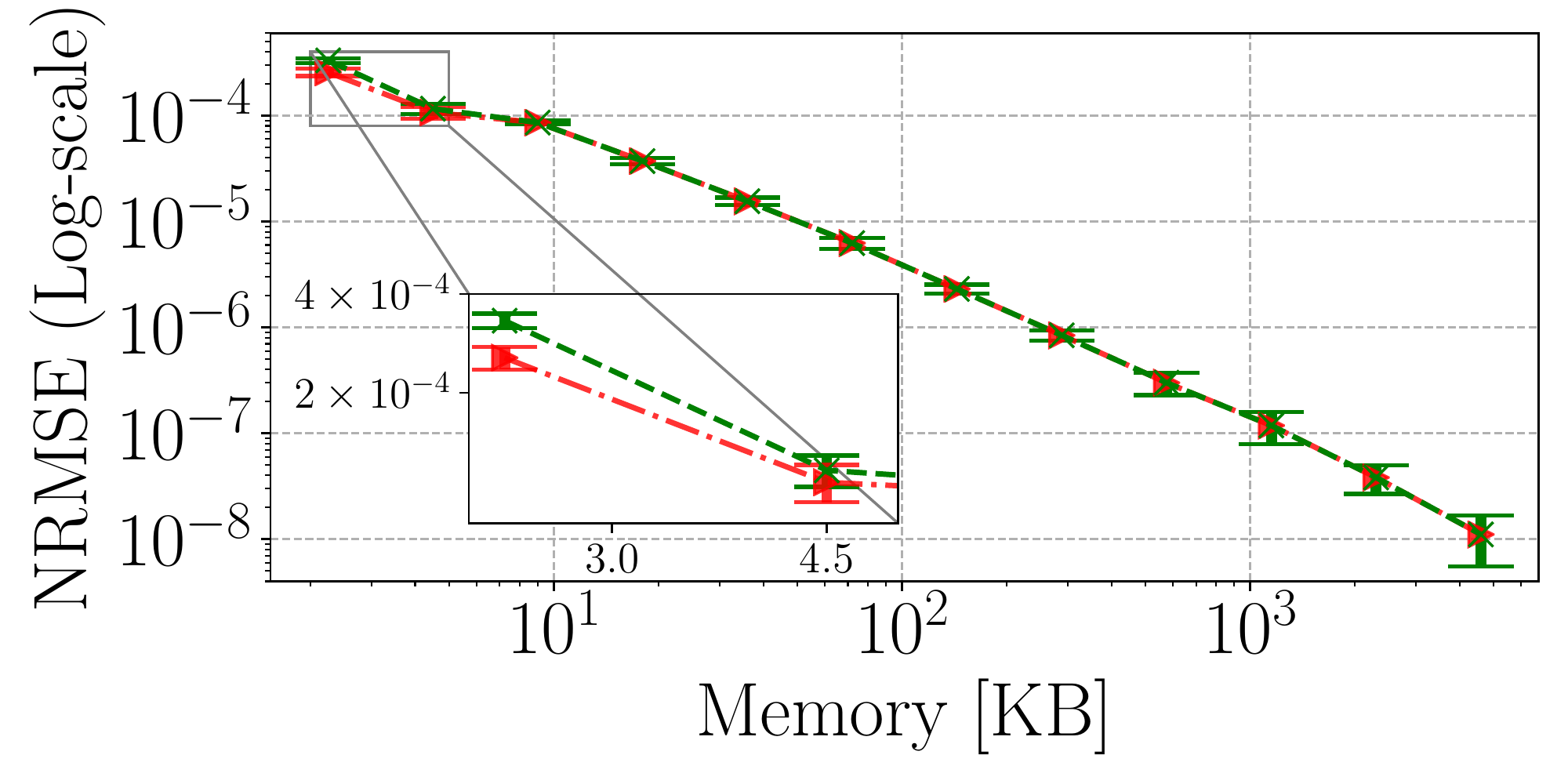}}   
    \hspace*{-1mm}\\
    {\includegraphics[width =0.7\columnwidth]
    {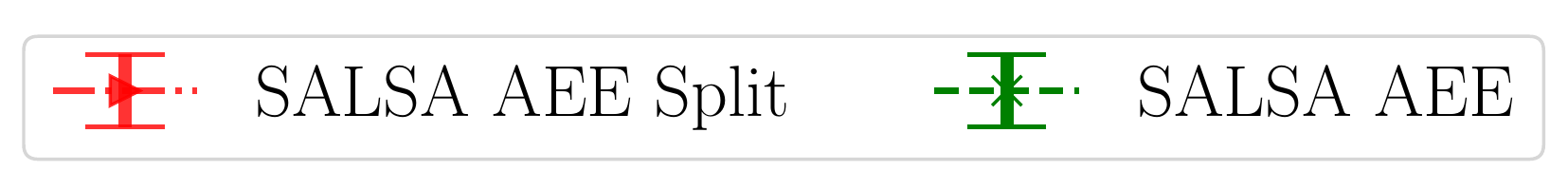}\ifdefined\sigmodSubmission\vspace*{-1mm}\fi}\\
    \caption{Affect of splitting counters in SALSA AEE (CM sketch).}\label{fig:splitting}
    \vspace*{-3mm}
\end{figure}

\textbf{Should We Split Counters?}
Finally, we check the accuracy gains obtainable by splitting counters. 
Intuitively, once a counter is downsampled, it may require fewer bits to represent. Therefore, if previously the counter had $s\cdot 2^{\ell}$ bits and the downsampled value is lower than $2^{s\cdot 2^{\ell-1}}-1$ (and $\ell\ge 1$), we can split the counter into two $s\cdot 2^{\ell-1}$-bit counters. As a result, there are now fewer collisions between elements, and SALSA AEE has better accuracy. However, as the results in Figure~\ref{fig:splitting} suggest, this effect is minor, and in most cases, the accuracy gains are insignificant.



\section{Conclusions}
We have presented SALSA, an efficient framework for dynamically re-sizing counters in sketching algorithms,  extending counters only when needed to represent larger numbers. 
SALSA starts from small counters and gradually adapts its memory layout to optimize the space-accuracy tradeoff. 
%
By evaluating across multiple real-world traces, sketches, and tasks, we have shown that, for a small overhead for its merging logic, SALSA reduces considerably the measurement error. In particular, our evaluation indicates that SALSA improves on state-of-the-art solutions such as Pyramid Sketch~\cite{PyramidSketch}, \mbox{ABC~\cite{gong2017abc}, Cold Filter~\cite{ColdFilter}, and the AEE estimators~\cite{Compsketch}.}

\ifdefined\fullversion
We believe that SALSA can replace and enhance existing sketches in more complex algorithms, such as $L_p$-samplers~\cite{LpSamplersSurvey} and database systems (e.g.,~\cite{Cheetah,SparkCMS}).
All of our code \mbox{is released as open source~\cite{SALSACode}.
}

\else
\CRdel{We believe that SALSA can replace and enhance existing sketches in more complex algorithms, such as $L_p$-samplers~\cite{LpSamplersSurvey} and database systems (e.g.,~\cite{Cheetah,SparkCMS}).
All of our code \mbox{is released as open source~\cite{SALSACode}.}}
Due to lack of space, some of our results are deferred to the full version~\cite{fullVersion}. 
Specifically, we explain and evaluate fine-grained merges, where counters can grow slower than doubling in size. 
As shown, the gain is marginal.
We also discuss  the count distinct functionality of the CM sketch, discuss how SALSA \rev{empirically improves} it, and evaluate the accuracy gain. 
The full version also includes an improved encoding for SALSA that requires under 0.6 bits per counter, with a lower bound that shows that this is near-optimal. Nonetheless, this encoding is more complex and may reduce the measurement throughput. Finally, we \mbox{provide the missing proofs and additional explanations.} All of our code \mbox{is released as open source~\cite{SALSACode}.}
\vspace*{-0mm}

\fi
\newpage
	\bibliographystyle{IEEEtran}
	\bibliography{references}

\begin{thebibliography}{10}
\providecommand{\url}[1]{#1}
\csname url@samestyle\endcsname
\providecommand{\newblock}{\relax}
\providecommand{\bibinfo}[2]{#2}
\providecommand{\BIBentrySTDinterwordspacing}{\spaceskip=0pt\relax}
\providecommand{\BIBentryALTinterwordstretchfactor}{4}
\providecommand{\BIBentryALTinterwordspacing}{\spaceskip=\fontdimen2\font plus
\BIBentryALTinterwordstretchfactor\fontdimen3\font minus
  \fontdimen4\font\relax}
\providecommand{\BIBforeignlanguage}[2]{{%
\expandafter\ifx\csname l@#1\endcsname\relax
\typeout{** WARNING: IEEEtran.bst: No hyphenation pattern has been}%
\typeout{** loaded for the language `#1'. Using the pattern for}%
\typeout{** the default language instead.}%
\else
\language=\csname l@#1\endcsname
\fi
#2}}
\providecommand{\BIBdecl}{\relax}
\BIBdecl

\bibitem{SALSACode}
``Salsa open-source code,'' 2020, \url{https://github.com/SALSA-ICDE2021}.

\bibitem{NLPsketches}
A.~Goyal, H.~Daum{\'e}~III, and G.~Cormode, ``Sketch algorithms for estimating
  point queries in nlp,'' in \emph{EMNLP-CoNLL}, 2012.

\bibitem{LoadBalancing}
G.~Dittmann and A.~Herkersdorf, ``Network processor load balancing for
  high-speed links,'' in \emph{SPECTS}, 2002.

\bibitem{Forensic2}
A.~K. Kaushik, E.~S. Pilli, and R.~C. Joshi, ``{"Network Forensic Analysis by
  Correlation of Attacks with Network Attributes"},'' in \emph{Information and
  Communication Technologies}, 2010.

\bibitem{ConextPaper}
R.~Ben{-}Basat, G.~Einziger, I.~Keslassy, A.~Orda, S.~Vargaftik, and
  E.~Waisbard, ``Memento: making sliding windows efficient for heavy hitters,''
  in \emph{ACM CoNEXT}, 2018.

\bibitem{DBLP:conf/sigmetrics/LallSOXZ06}
A.~Lall, V.~Sekar, M.~Ogihara, J.~J. Xu, and H.~Zhang, ``Data streaming
  algorithms for estimating entropy of network traffic,'' in \emph{ACM
  SIGMETRICS/Performance}, 2006.

\bibitem{RAP}
R.~{Ben-Basat}, X.~{Chen}, G.~{Einziger}, R.~{Friedman}, and Y.~{Kassner},
  ``Randomized admission policy for efficient top-k, frequency, and volume
  estimation,'' \emph{IEEE/ACM Transactions on Networking}, 2019.

\bibitem{AugmentedSketch}
P.~Roy, A.~Khan, and G.~Alonso, ``Augmented sketch: Faster and more accurate
  stream processing,'' in \emph{ACM SIGMOD}, 2016.

\bibitem{ColdFilter}
T.~Yang, J.~Jiang, Y.~Zhou, L.~He, J.~Li, B.~Cui, S.~Uhlig, and X.~Li, ``Fast
  and accurate stream processing by filtering the cold,'' \emph{{VLDB} J.},
  2019.

\bibitem{CountMinSketch}
G.~Cormode and S.~Muthukrishnan, ``An improved data stream summary: The
  count-min sketch and its applications,'' \emph{J. Algorithms}, 2004.

\bibitem{Brick}
N.~Hua, B.~Lin, J.~J. Xu, and H.~C. Zhao, ``Brick: A novel exact active
  statistics counter architecture,'' in \emph{ACM/IEEE ANCS}, 2008.

\bibitem{PyramidSketch}
T.~Yang, Y.~Zhou, H.~Jin, S.~Chen, and X.~Li, ``Pyramid sketch: A sketch
  framework for frequency estimation of data streams,'' 2017, code Available:
  \url{https://github.com/zhouyangpkuer/Pyramid_Sketch_Framework}.

\bibitem{CUSketch}
C.~Estan and G.~Varghese, ``New directions in traffic measurement and
  accounting,'' \emph{ACM SIGCOMM}, 2002.

\bibitem{CountSketch}
M.~Charikar, K.~Chen, and M.~Farach-Colton, ``Finding frequent items in data
  streams,'' in \emph{EATCS ICALP}, 2002.

\bibitem{univmon}
Z.~Liu, A.~Manousis, G.~Vorsanger, V.~Sekar, and V.~Braverman, ``One sketch to
  rule them all: Rethinking network flow monitoring with univmon,'' in
  \emph{ACM SIGCOMM}, 2016.

\bibitem{Compsketch}
R.~B. Basat, G.~Einziger, M.~Mitzenmacher, and S.~Vargaftik, ``Faster and more
  accurate measurement through additive-error counters,'' in \emph{IEEE
  INFOCOM}, 2020, code: \url{https://github.com/additivecounters/AEE}.

\bibitem{gong2017abc}
J.~Gong, T.~Yang, Y.~Zhou, D.~Yang, S.~Chen, B.~Cui, and X.~Li, ``Abc: a
  practicable sketch framework for non-uniform multisets,'' in \emph{2017 IEEE
  Big Data}, 2017.

\bibitem{Nitro}
Z.~Liu, R.~Ben-Basat, G.~Einziger, Y.~Kassner, V.~Braverman, R.~Friedman, and
  V.~Sekar, ``Nitrosketch: Robust and general sketch-based monitoring in
  software switches,'' in \emph{ACM SIGCOMM}, 2019.

\bibitem{PRECISION}
{Ran Ben Basat, Xiaoqi Chen, Gil Einzinger, Ori Rottenstreich}, ``{Efficient
  Measurement on Programmable Switches Using Probabilistic Recirculation},'' in
  \emph{IEEE ICNP}, 2018.

\bibitem{CASE}
L.~Yang, W.~Hao, P.~Tian, D.~Huichen, L.~Jianyuan, and L.~Bin, ``Case:
  Cache-assisted stretchable estimator for high speed per-flow measurement,''
  in \emph{IEEE INFOCOM}, 2016.

\bibitem{RandomizedCounterSharing}
T.~Li, S.~Chen, and Y.~Ling, ``Per-flow traffic measurement through randomized
  counter sharing,'' \emph{IEEE/ACM Trans. on Networking}, 2012.

\bibitem{CounterBraids}
Y.~Lu, A.~Montanari, B.~Prabhakar, S.~Dharmapurikar, and A.~Kabbani, ``Counter
  braids: a novel counter architecture for per-flow measurement,'' in \emph{ACM
  SIGMETRICS}, 2008.

\bibitem{CounterTree}
M.~Chen and S.~Chen, ``Counter tree: {A} scalable counter architecture for
  per-flow traffic measurement,'' in \emph{IEEE ICNP}, 2015.

\bibitem{CEDAR}
E.~Tsidon, I.~Hanniel, and I.~Keslassy, ``Estimators also need shared values to
  grow together,'' in \emph{IEEE INFOCOM}, 2012.

\bibitem{ANLSUpscaling}
C.~Hu and B.~Liu, ``Self-tuning the parameter of adaptive non-linear sampling
  method for flow statistics,'' in \emph{CSE}, 2009.

\bibitem{ApproximateCounting}
R.~Morris, ``Counting large numbers of events in small registers,''
  \emph{Commun. ACM}, 1978.

\bibitem{CUAnal}
G.~{Einziger} and R.~{Friedman}, ``A formal analysis of conservative update
  based approximate counting,'' in \emph{ICNC}, 2015.

\bibitem{UnivMonTheory}
V.~Braverman and R.~Ostrovsky, ``Zero-one frequency laws,'' in \emph{ACM STOC},
  2010.

\bibitem{IntrusionDetection2}
P.~Garcia-Teodoro, J.~E. Díaz-Verdejo, G.~Maciá-Fernández, and E.~Vázquez,
  ``Anomaly-based network intrusion detection: Techniques, systems and
  challenges,'' \emph{Computers and Security}, 2009.

\bibitem{elasticsketch}
T.~Yang, J.~Jiang, P.~Liu, Q.~Huang, J.~Gong, Y.~Zhou, R.~Miao, X.~Li, and
  S.~Uhlig, ``Elastic sketch: Adaptive and fast network-wide measurements,'' in
  \emph{Proc. of ACM SIGCOMM}, 2018.

\bibitem{LinearCounting}
K.-Y. Whang, B.~T. Vander-Zanden, and H.~M. Taylor, ``A linear-time
  probabilistic counting algorithm for database applications,'' \emph{ACM
  Transactions on Database Systems (TODS)}, 1990.

\bibitem{FAST}
R.~B. Basat, G.~Einziger, and R.~Friedman, ``Fast flow volume estimation,''
  \emph{Pervasive and Mobile Computing}, 2018, code available at:
  \url{https://github.com/ranbenbasat/FAST}.

\bibitem{CormodeCode}
\BIBentryALTinterwordspacing
G.~Cormode, ``Implementation of heavy hitter algorithms.'' [Online]. Available:
  \url{http://hadjieleftheriou.com/frequent-items/}
\BIBentrySTDinterwordspacing

\bibitem{teuhola2011interpolative}
J.~Teuhola, ``Interpolative coding of integer sequences supporting log-time
  random access,'' \emph{Information processing \& management}, 2011.

\bibitem{elmasry2012improved}
A.~Elmasry, J.~Katajainen, and J.~Teuhola, ``Improved address-calculation
  coding of integer arrays,'' in \emph{SPIRE}, 2012.

\bibitem{SpectralBloom}
S.~Cohen and Y.~Matias, ``Spectral bloom filters,'' in \emph{ACM SIGMOD}, 2003.

\bibitem{LC}
G.~S. Manku and R.~Motwani, ``Approximate frequency counts over data streams,''
  in \emph{VLDB}, 2002.

\bibitem{caffein}
B.~Manes, ``Caffeine: A high performance caching library for java 8,''
  \url{https://github.com/ben-manes/caffeine}.

\bibitem{L2IQ}
N.~Ivkin, R.~B. Basat, Z.~Liu, G.~Einziger, R.~Friedman, and V.~Braverman, ``I
  know what you did last summer: Network monitoring using interval queries,''
  in \emph{ACM SIGMETRICS}, 2020.

\bibitem{CAIDA2018}
``The caida equinix-newyork packet trace, 20181220-130000.'' 2018.

\bibitem{CAIDA2016}
``The caida equinix-chicago packet trace, 20160406-130000.'' 2016.

\bibitem{UWISC}
T.~Benson, A.~Akella, and D.~A. Maltz, ``Network traffic characteristics of
  data centers in the wild,'' in \emph{ACM IMC}, 2010.

\bibitem{kaggleYouTubeDataset}
``Trending youtube video statistics, kaggle,'' 2018,
  \url{https://www.kaggle.com/datasnaek/youtube-new}.

\bibitem{HeavyHitters}
R.~Ben-Basat, G.~Einziger, R.~Friedman, and Y.~Kassner, ``Heavy hitters in
  streams and sliding windows,'' in \emph{IEEE INFOCOM}, 2016.

\bibitem{krishnamurthy2003sketch}
B.~Krishnamurthy, S.~Sen, Y.~Zhang, and Y.~Chen, ``Sketch-based change
  detection: methods, evaluation, and applications,'' in \emph{ACM IMC}, 2003.

\bibitem{student1908probable}
Student, ``The probable error of a mean,'' \emph{Biometrika}, 1908.

\bibitem{qi2019cuckoo}
J.~Qi, W.~Li, T.~Yang, D.~Li, and H.~Li, ``Cuckoo counter: A novel framework
  for accurate per-flow frequency estimation in network measurement,'' in
  \emph{ACM/IEEE ANCS}, 2019.

\bibitem{SpaceSavingIsTheBest2010}
G.~C{or}mode and M.~Hadjieleftheriou, ``Methods for finding frequent items in
  data streams,'' \emph{J. VLDB}, 2010.

\bibitem{eisenbud2016maglev}
D.~E. Eisenbud, C.~Yi, C.~Contavalli, C.~Smith, R.~Kononov, E.~Mann-Hielscher,
  A.~Cilingiroglu, B.~Cheyney, W.~Shang, and J.~D. Hosein, ``Maglev: A fast and
  reliable software network load balancer,'' in \emph{USENIX NSDI}, 2016.

\bibitem{LpSamplersSurvey}
G.~Cormode and H.~Jowhari, ``\emph{L}\({}_{\mbox{\emph{p}}}\) samplers and
  their applications: {A} survey,'' \emph{{ACM} Comput. Surv.}, 2019.

\bibitem{Cheetah}
M.~Tirmazi, R.~Ben~Basat, J.~Gao, and M.~Yu, ``Cheetah: Accelerating database
  queries with switch pruning,'' in \emph{ACM SIGMOD}, 2020.

\bibitem{SparkCMS}
``{Apache Spark 3 support for Count Min Sketch.}'' 2020,
  \url{https://spark.apache.org/docs/3.0.0-preview/api/scala/org/apache/spark/util/sketch/CountMinSketch.html}.

\end{thebibliography}
%

\ifdefined\sigmodSubmission
\end{document}\endinput
\else

\ifdefined\sigmodSubmission
\cleardoublepage
\begin{appendices}
\section{Cover Letter}
\rev{
We thank the reviewers for their helpful comments and suggestions. 
The main changes we did in this revision are:
\begin{itemize}
    \item We added a new experiment that justifies our choice of running the baseline sketches using 32-bits. The results appear in Figure~\ref{fig:smallCounters} show that while SALSA works well when starting with $s=8$ bit counters, using the baseline with $8$ or $16$ bit counters result in large errors, especially for the heavy hitter elements.
    \item We removed all mentioning of our fine-grained merging algorithm Tango. Instead, we reference the full version and summarize the findings.
    \item We removed the description of how SALSA improves the count distinct functionality from the main text and reference the full version instead. Our view is that this is not a major contribution and can be deferred to make space to address the comments.
\end{itemize}
Below we respond to the individual comments.}

\newcommand{\response}[1]{\par\noindent\textcolor{blue}{\textbf{Response:} #1}\\}

\subsection{META-REVIEWER}
1. A short summary on the recommendation, and suggestions for revision (for decision on revision):

Dear authors,
Thank you for submitting your work to ICDE.

The reviewers identified several positive aspects at the paper, but also expressed concerns over technical details that are not adequately explained (see R3.W1, along with the revision points of Reviewer 3). This is a crucial point to be addressed.
There are also several other revision points, clearly marked by the reviewers, that involve presentation issues and the clarity of the text.
Finally, Revision points R3.W4 involves the discussion of theoretical analysis on how different merging options affect the accuracy guarantees, while R3.W3 involves the benefits of SALSA compared to prior resizing techniques.

\response{Thanks for the summary. We have considered all points raised by the reviewers and respond to the individual comments below. }

\subsection{Reviewer \#1}
W1: There are some presentation issues.
\response{Based on the reviewer comments, and with some additional editing, we improved the presentation.  Many are in blue; some simple comma changes or typo fixes are not marked.}

W2: The technique is quite complex and there is some penalty in throughputs.
\response{Indeed, adding the merging logic adds complexity, which is expected. When designing SALSA we aimed for lightweight encoding that still allows high-speed processing. }

\textbf{8. Detailed evaluation, labeled D1, D2, D3 etc.}

D1: I find the ideas to be quite clever. The counter merge is dynamic, and it can be retrieved at the corresponding merge level.

D2: The paper does a good survey of existing sketch techniques, and shows that the proposed technique can be added on them.

D3: The theoretical analysis is solid and shows that it improves upon the original sketches, although the exact improvement is not quantified.
\response{In an adversarial input (e.g., when all frequencies are at least $2^{16}$ and all occurrences of an element are contiguous), SALSA would produce the exact estimates as the underlying sketch. If one assumes a specific input distribution, our analysis shows the benefit. For example, in the line before~\eqref{eq:CSeq2}, we show that SALSA CS reduces the variance by {\small $\parentheses{\mathbb E\brackets{X_B^2|\neg O_{AB}}\Pr[\neg O_{AB}]+\mathbb E\brackets{X_{CD}^2|\neg O_{ABCD}}\Pr[\neg O_{ABCD}]}$.} Plugging in the distribution parameters, one can compute the exact benefit.}

D4: There are some presentation issues. For example, the "Tango" version is never introduced in the paper, but it is used many times, including the theorems. It is in the long/full version. Also, many places of the paper are hard to read -- they would benefit from some plain explanations and intuitions. For example, one place \mbox{is the indices of the merges and encoding bits.}
\response{Thank you for noticing that. We removed all mentioning of Tango from the body of the paper and summarized the findings with a reference to the full version. We also add the merging bits to Figure~\ref{fig:cache-based} to illustrate how the encoding changes during insertions.}

D5: There is some overhead associated with the added logic. As shown in the experiments, the throughputs decrease compared to the original versions.
\response{Please see response to W2.}

\textbf{Required changes for a revision, if applicable. Labeled R1, R2, R3, etc.}
Please address W1 and D4.

\subsection{Reviewer \#2}
7. List three or more weak points, labeled W1, W2, W3, etc.
W1: The key ideas are similar to those used in prior techniques like Pyramid Sketch and ABC.
\response{While there are some similarities, SALSA has unique advantages such as unbounded counter sizes (they can merge into a single counter if needed) which is essential for capturing the sizes of heavy hitters accurately (see Figure~\ref{fig:scatter}).}

W2: The paper is written sloppily, and is difficult to read and verify because of missing pieces and terse proofs with significant mathematical notation.
\response{Based on the reviewer comments, and with some additional editing, we improved the presentation and writing.}

W3: Following up on W3, in several cases, the guarantees are not clearly stated. More below.
\response{Please see detailed answers below.}

\textbf{8. Detailed evaluation, labeled D1, D2, D3 etc.}

Overall, the dynamic re-sizing technique proposed in the paper, although somewhat incremental, is a clear and significant new contribution. Generally speaking, the claims My main concerns with the paper have to do with the writing which is quite sloppy in places, and some organizational issues.

D1: Section III: In the beginning, Strict Turnstile is defined to be non-negative, but one paragraph below: the authors state that ".. CMS operates in Strict Turnstile model where all frequencies are positive..". Which is it? Also, it is clear why negative frequencies matter. But why do 0s matter? For researchers well-versed in this topic, this may be obvious.
\response{It should be non-negative, we corrected the text.}

D2: Section V: Linear Counting is explained very briefly, and without prior familiarity with that work, it may be difficult to follow this part.
\response{To meet the page limit and make room to address the comments, we decided to defer linear counting altogether to the full version and summarize the results here. In the full version, Linear Counting is explained with more detail.}

D3: It would be really useful to have a table that clearly summarizes how the guarantees for different sketches are affected by using SALSA. For instance, for Linear Counting, it appears that the authors were not able to prove the accuracy guarantees. However, the paragraph beginning suggests that SALSA improves the performance.
\response{Due to lack of space, we could not add such a table and have moved the count distinct extension to the full~version. We emphasize that for Linear Counting the gains are empirical gains, and not gains in the asymptotic guarantees.}

D4: Similarly for UnivMon, it is unclear if using SALSA is strictly better than not using SALSA. Also, is "L2-sketches" synonymous with CS?
\response{In general, Univmon requires that the underlying L2 sketches provide an $(\epsilon,\delta)$ guarantee. Our SALSA CS analysis shows that the guarantee for SALSA is at least as good as standard CS (and may be better, depending on the workload) and thus we have the same accuracy guarantee for Univmon when using SALSA. We clarified this in the paper.
CS is the most popular L2 sketch, although for cash register streams there are more accurate sketches such as BPTree (Braverman et al., PODS 2017). Nonetheless, to the best of our knowledge, all existing Univmon variants \mbox{use CS as the building block.}}

D5: There are also quite a few spelling mistakes and other types of typos in the paper. I recommend that the authors do a careful reading to fix those.
\response{We have proofread the paper to fix such issues.}

\textbf{Required changes for a revision, if applicable. Labeled R1, R2, R3, etc.}
D1-D5 above.

\subsection{Reviewer \#3}

7. List three or more weak points, labeled W1, W2, W3, etc.
W1. The theoretical analysis in Section 5 fails to show the advantage of SALSA -- the assumption used in the theorems, that is, the baseline sketching uses the same number of counters as the SALSA sketching, is not fair.
\response{Please see our response for the required changes.}

I felt that the merge process will affect things like the randomness of the hashing functions. Thus, I did not see how to formally evaluate the theoretical performance of SALSA.
\response{Notice that the merge process does not change the hash function. Instead, if an element is hashed into location $i$, the counter may be some contiguous block that includes $i$ (e.g., $\langle i, i+1\rangle$). Let us consider an example. Say that $h(x)=13$ for some element $x$ in $s=8$-bits SALSA. In the underlying sketch, $\langle 12,13,14,15\rangle$ serve as a single $32$-bits counter. In SALSA, counter $13$ starts as an independent counter, and may merge with $12$ (and then further with $\langle 14,15\rangle$) if needed to represent a larger counting range. The essence of our proofs is that for CMS and CUS, the estimate of SALSA is at least as accurate as the underlying counter (i.e., its value is always at most the value of the underlying counter). Similarly, we proved that the variance of SALSA CS is always lower than that of the underlying CS counter. Our analysis factors in the possibility of merging and how it would affect the estimates (e.g., see Lemma~\ref{lem:CS_var}).
}

W2. The bit tricks used in SALSA will increase the time complexity of the sketching algorithms, which will diminish the minor improvement in accuracy in practice.
\response{Sketching algorithms are mainly measured in three axes -- speed, accuracy, and memory. In some applications, speed is critical and one may not opt for a slower and more accurate solution. However, in other cases, the accuracy to memory tradeoff is the most important parameter, given that the speed is sufficient. For example, when running on a network switch one needs to meet the line rate, and pick the most accurate solution that satisfies that constraint.}

W3. SALSA does not show much advantage over the other two resizing techniques (Pyramid \mbox{Sketch and ABC) in Section 7.}
\response{The main advantage of SALSA is the unbounded counting range, which allows it to accurately capture the size of large flows. This is evident from Figure~\ref{fig:scatter}. Notice that it shows in log scale, which may make the differences look small, despite SALSA being several orders of magnitude more accurate on large \mbox{(e.g., frequency of at least $10^4$) elements.} We now note the log scale in the axes' legends.}

W4. The authors provided two merging options: sum and max. But there is no theoretical analysis on how different merging options will affect the accuracy.
\response{It is important to notice that the merging strategies do not apply to all sketches. In turnstile streams, one must use sum-merging, and also when using SALSA CS (even on cash register streams). Similarly, SALSA CUS must use max-merging. The only setting in which both options can provide accuracy guarantees is when using SALSA CMS on cash register streams. In this case, the estimates of max-merging are always at least as accurate as sum-merging and this can easily be \mbox{proven by induction on the stream length.}}

\textbf{8. Detailed evaluation, labeled D1, D2, D3 etc.}

D1. $M_4$ in Figure 1 should be 0 according to the description in Section 4 “The SALSA encoding”
\response{Thanks for noticing that! indeed, there was a typo in the formula. We have edited the text and added examples. }

D2. Title of figure 6a should be “Speed, NY18”
\response{Thanks. We corrected it.}

D3. ICDE paper should be single blinded, not double blinded
\response{Thanks. The PC chair contacted us and we revised the manuscript of the original submission a few days after the deadline. It is fixed now.}

\textbf{Required changes for a revision, if applicable. Labeled R1, R2, R3, etc.}
Please address W1, W3 and W4.

For W1, The standard CMS uses hash function h(x) which maps x to [1, W]. That is, the standard CMS has W buckets each of size s bits. While the hash function $\hat{h}(x) (= h(x)/2^\ell)$ introduced by the authors maps x to $[1, W/2^\ell]$. Then the authors used $\hat{h}(x)$ instead of h(x) in the estimation phase even for CMS. This seems to force the standard CMS to use only $W/2^\ell$ buckets each of size $2^{\ell}\cdot s$ bits. 
\response{This is correct. The idea is that if practitioners use, e.g., 32-bit counters for CMS, and we use $s=8$-bits SALSA CMS, then SALSA have four times as many (small) counters which gives both variants the about same space.}
I think this is not quite fair since it is likely that the performance of CMS on W buckets each of size s is better than that of CMS on $W/2^\ell$ buckets each of size $2^{\ell\cdot s}$ bits. Please clarify this.
\response{The above setting means that the space used by SALSA is just its encoding overhead plus the same space as the underlying sketch. For example, given a 1MB space allocation, $d=4$-rows CMS would have $2^{16}$ counters in each row, each of size $32$ bits and $s=8$ SALSA would have $2^{18}$ counters. Notice that our evaluation (e.g., Figure~\ref{fig:Pyramid}) does that the encoding overhead into account.
For completeness, we added Figure~\ref{fig:smallCounters} that shows that running CMS with a larger number of smaller (e.g., $8$-bit) counters may be less accurate than a standard (i.e., with $32$-bit counters) sketch. We used the default configuration (98M packets and $s=8$-bit counters for SALSA).
Additionally, we observed that the ARE and AAE metrics, in which ABC and Pyramid are competitive, may not be useful to evaluate algorithms on large workloads or for capturing heavy hitters, which are often considered the most important elements. The reason is that in heavy tailed workloads, such as CAIDA or skew$\le 1$ Zipf, there are many small elements that carry most of the weight in the averaging. This is especially evident by the experiment whose results are shown in Figure~\ref{fig:zeroAlg} above: we measured the error on all heavy hitters -- elements larger than a $\phi$ fraction of the input. The leftmost point ($\phi=10^{-8}$) corresponds to the ARE metric (i.e., all flows will be considered). As shown, in this case the best algorithm is $\overline 0$ which corresponds to returning $0$ estimates for all element sizes. That is, according to this metric, one can reduce the error by not running measurements at all. A similar result was observed for AAE (see Figure~\ref{fig:zeroAlgAAE}, where the $\overline 0$ algorithm beats the baseline when considering all flows (leftmost point).}
\begin{figure}[t]
    \centering
    \ifdefined\sigmodSubmission\vspace*{-1mm}\fi
    \hspace*{-2mm}
    { \includegraphics[width =1.0\columnwidth]
    {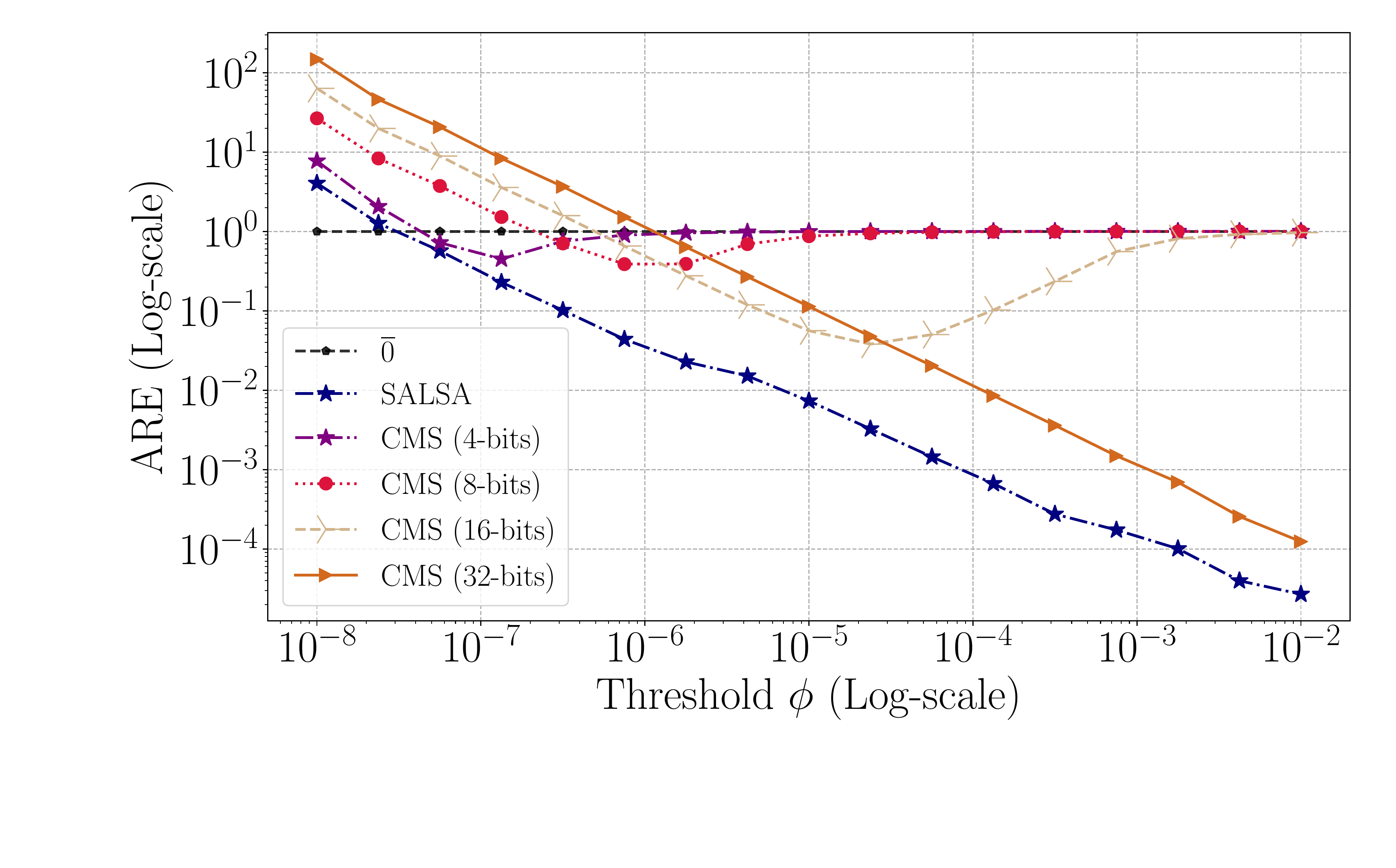}}
    \vspace*{-15mm}
    \caption{\rev{Running CMS with small number of bits and the ``$\overline 0$'' algorithm for estimating heavy hitter sizes (2MB) using average relative error metric. The leftmost point corresponds to the standard ARE metric (used in figures~\ref{fig:areNY18} and \ref{fig:areCH16}), which considers all flows.}}\label{fig:zeroAlg}
    \ifdefined\sigmodSubmission\vspace*{-2mm}\fi
\end{figure}
\begin{figure}[t]
    \centering
    \ifdefined\sigmodSubmission\vspace*{-3mm}\fi
    \hspace*{-2mm}
    { \includegraphics[width =1.0\columnwidth]
    {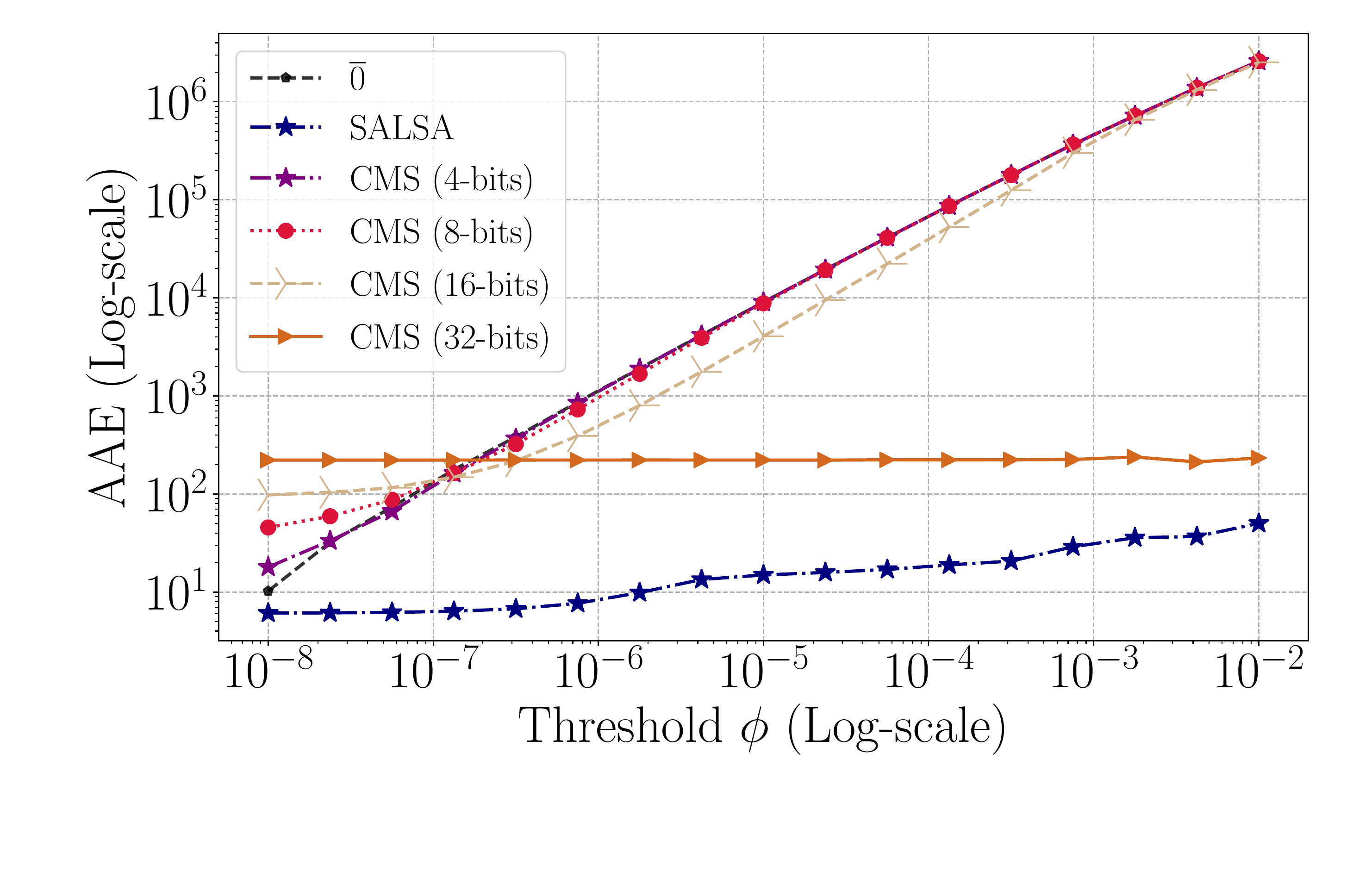}}
    \vspace*{-15mm}
    \caption{\rev{Running CMS with small number of bits and the ``$\overline 0$'' algorithm for estimating heavy hitter sizes (2MB) using average absolute error metric. The leftmost point corresponds to the standard AAE metric (used in figures~\ref{fig:aaeNY18} and \ref{fig:aaeCH16}), which considers all flows.}}\label{fig:zeroAlgAAE}
    \ifdefined\sigmodSubmission\vspace*{-2mm}\fi
\end{figure}
Regarding the randomness, there may be some dependency when switching the hash function from "$h(x) \mod 2^{\ell-1}$" to "$h(x) \mod 2^l$" in the merge process. This should be discussed in the analysis.
\response{This is correct; when merging a $s\cdot 2^{\ell-1}$ bits counter noise from additional elements affect future queries. Notice that the hash functions are determined beforehand and do not change due to merges.
This is taken into account in our analysis, e.g., in Lemma~\ref{lem:CS_var} where we consider the random events that correspond to potential merges. The argument is simpler for sketches such as CMS and CUS, where the error of SALSA is deterministically bounded by that of the underlying sketches.}
\end{appendices}
\else
\begin{appendices}

\ifdefined\icdeSubmission

\section{Fine-grained Counter Merges:}\label{sec:tango}
The SALSA encoding we presented in Section~\ref{sec:SALSA_encoding} doubles the counter size upon an overflow, which may be wasteful when the overflowing counter could benefit from a smaller increase in size. 
Thus, we suggest the more refined Tango algorithms to explore the benefits of a more fine-grained merging strategy. 
In Tango, counters can be merged into sizes that are arbitrary multiples of $s$. For example, if we start from $s=8$ bit counters, Tango can merge a $16$ bit counter into a $24$ bit counter while SALSA would merge from $16$ bits to $32$. The encoding of Tango is simple: each counter $j$ is associated with a merge bit $m_j$ that denotes whether the counter is merged with its right-neighbor.
To compute the counter size and offset in Tango of $j=h(x)$, we scan the number of set bits to the left and right of $m_j$ until we hit a zero at both sides. 
For example, if $j = 5$ and $m_4=m_5=m_6=m_7=1$ while $m_3=m_8=0$ then the counter consists of $s\cdot 5$ bits, spanning $\angles{4,5,6,7,8}$.
In general, one can use complex logic to decide whether to merge with the left or right neighbor once a counter overflows. However, we design Tango to evaluate the potential benefits of fine-grained merging and therefore enforce a merging logic that mimics SALSA. 
Specifically, Tango always tries to be aligned to the smallest possible power of two. For example, if counter $9$ overflows, it merges with $8$ to be aligned with the $2$-block $\angles{8,9}$. If it overflows again, it merges with $10$ (creating a $s\cdot 3$ bits sized counter) and then with $11$. If more bits are needed it will merge with $12$ then with $13,14$ and $15$ (being aligned to the $8$-block $\angles{8,\ldots,15}$). Then it merges with $7, 6,\ldots$, etc. Notice that at every point in time, the Tango counters are contained in the corresponding SALSA counters. In particular, this allows us to produce an estimate that is at least as accurate as SALSA.
We note that Tango poses a tradeoff -- while it allows more accurate sketches (e.g., as a counter may not exceed $2^{24}-1$ and thus it could be wasteful to merge it into $32$ bits), it also has slower decoding time and cannot use the efficient encoding of the previous section.

\textbf{Is Fine-grained Merging Worth It?}
To understand the accuracy improvement attainable by fine-grained merging (as opposed to SALSA's approach of doubling the counter size at each overflow), we compare SALSA with Tango. 
As the results in Figure~\ref{fig:tango} indicate, Tango also offers the best accuracy-space tradeoff when starting with $s=8$ bits (Tango$_{16}$ is equivalent to SALSA$_{16}$ and is omitted). However, while it is slightly more accurate, the gains seem marginal considering the computationally expensive operations of determining the counter's size and offset. Further, Tango has an overhead of $1$ bit per counter and does not obviously allow an efficient encoding \mbox{like SALSA does (Section~\ref{sec:smartEncoding}).}

\begin{figure}[]
    \centering
    \subfloat[Error, NY18]
    { \includegraphics[width =0.49\columnwidth]
    {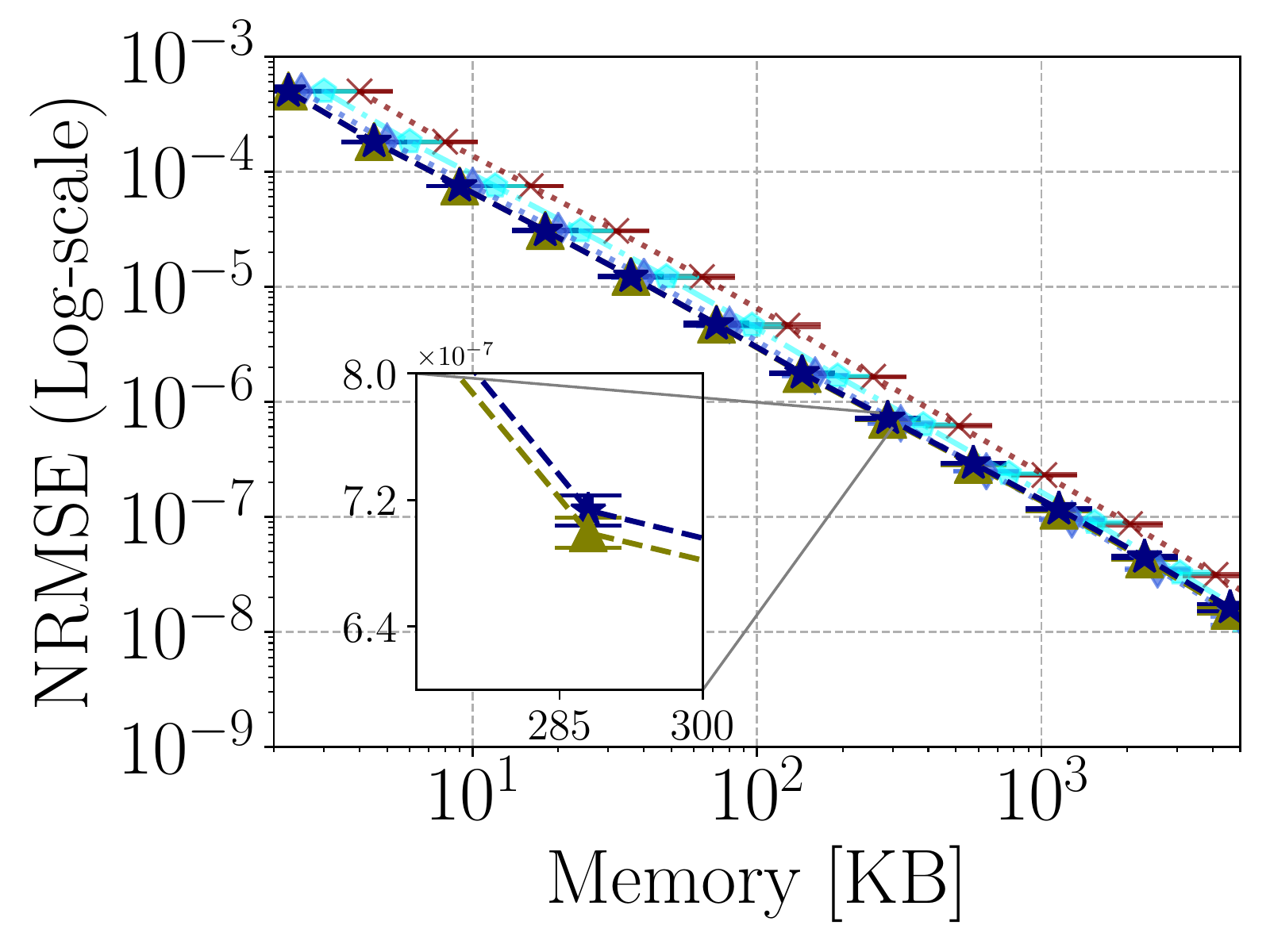}}
    \subfloat[Error, Zipf]
    { \includegraphics[width =0.49\columnwidth]
    {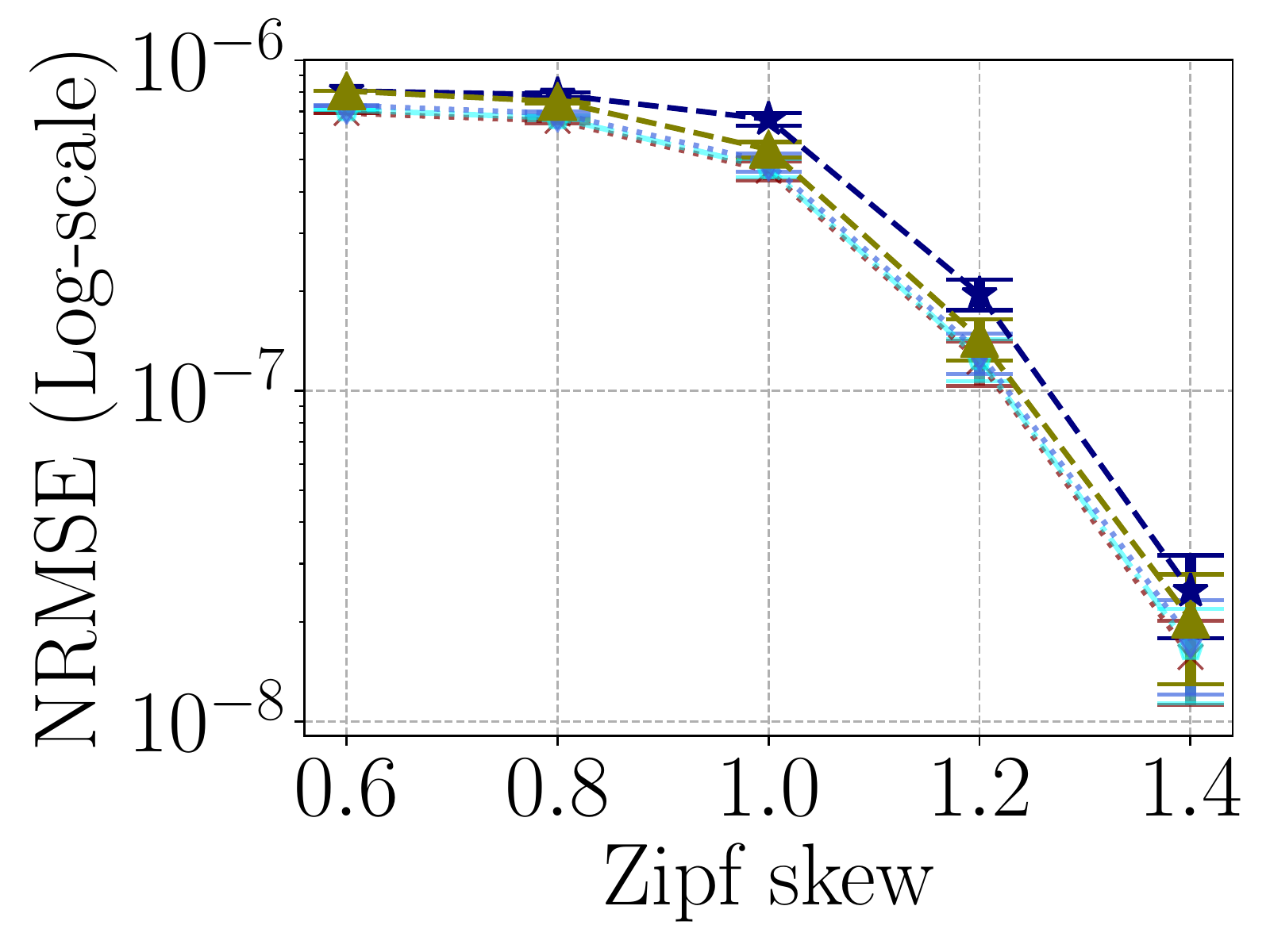}}   
    \\
    {\includegraphics[width =1.00\columnwidth]
    {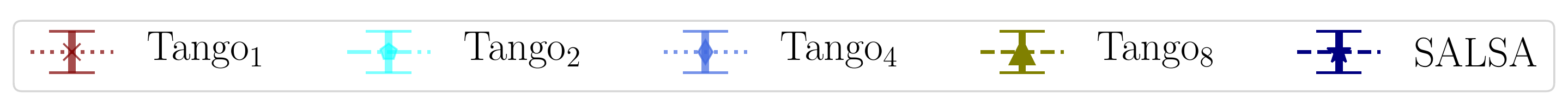}\ifdefined\sigmodSubmission\vspace*{-1mm}\fi}
    \caption{\small Accuracy of SALSA CMS (with $s=8$ bits) vs. Tango CMS. In (b), Tango{\Large$_s$} is allocated with $2(1+1/s)$MB of space while SALSA uses $2(1+1/8)=2.25$MB.
    }\label{fig:tango}
    \ifdefined\sigmodSubmission\vspace*{-3mm}\fi
\end{figure}

\fi

\countdistinctapp{\section{Counting Distinct}\label{app:Counting Distinct}
\textbf{{Counting Distinct Items:}}
Estimating the number of distinct items in a data stream (defined as $F_0\equiv\norm f_0$) is a fundamental primitive for applications such as discovering denial of service attacks~\cite{IntrusionDetection2}.
While UnivMon can natively support such a function, we can also estimate it from CMS and CUS.
By observing the fraction of zero-valued counters in a sketch's row $p$, we can estimate the number of distinct elements (as additional occurrences of the same element would not change this quantity). 
Specifically, a common approach (e.g.,~\cite{elasticsketch}) is use the Linear Counting algorithm~\cite{LinearCounting} that estimates the distinct count as $\frac{\log p}{\log(1-1/w)}\approx -w\log p$. Such an estimate has a standard error of $\frac{\sqrt {w\cdot (e^{\frac{F_0}{w}}-\frac{F_0}{w}-1)}}{F_0}$~\cite{LinearCounting} that \mbox{improves when $w$ grows.}}

\countdistinctapp{\textbf{Count Distinct using Count Min:}
We evaluate the performance of SALSA CMS on additional applications such as counting distinct elements and estimating the size of the heavy hitters.
As shown in the count, distinct results (Figures~\ref{fig:cmsCDandHHapp}), neither SALSA CMS nor the Baseline are effective with low memory footprints. This is because no counters remain zero-valued, and the Linear Counting estimator fails. 
Nevertheless, SALSA CMS can work with less memory (4.5MB for NY18 and 1.125MB for CH16) and reduce the estimation error when the Baseline does produce estimates. Intuitively, Linear Counting with $w$ buckets can count up to $w\ln w$ elements, so the number of elements in the datasets (6.5M for NY18 and 2.5M for CH16) \mbox{imposes a lower bound on the amount of space needed.}}

\ifdefined\icdeSubmission
\begin{figure}[t]
    \centering
    \hspace*{-2mm}
    \subfloat[NY18 Dataset]
    { \includegraphics[width =0.16\textwidth]
    {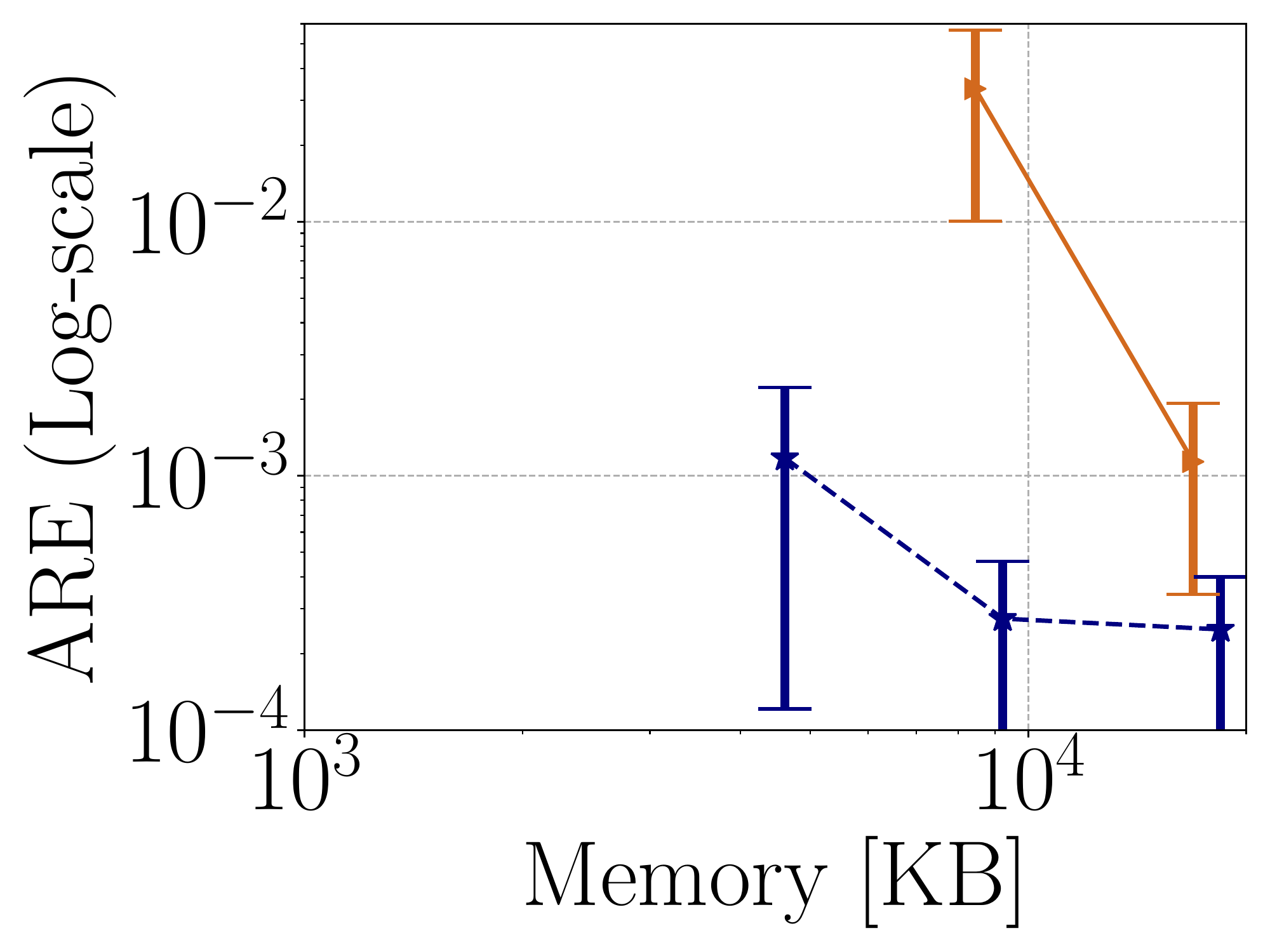}}
    \hspace*{-2mm}
    \subfloat[CH16 Dataset]
    { \includegraphics[width =0.16\textwidth]
    {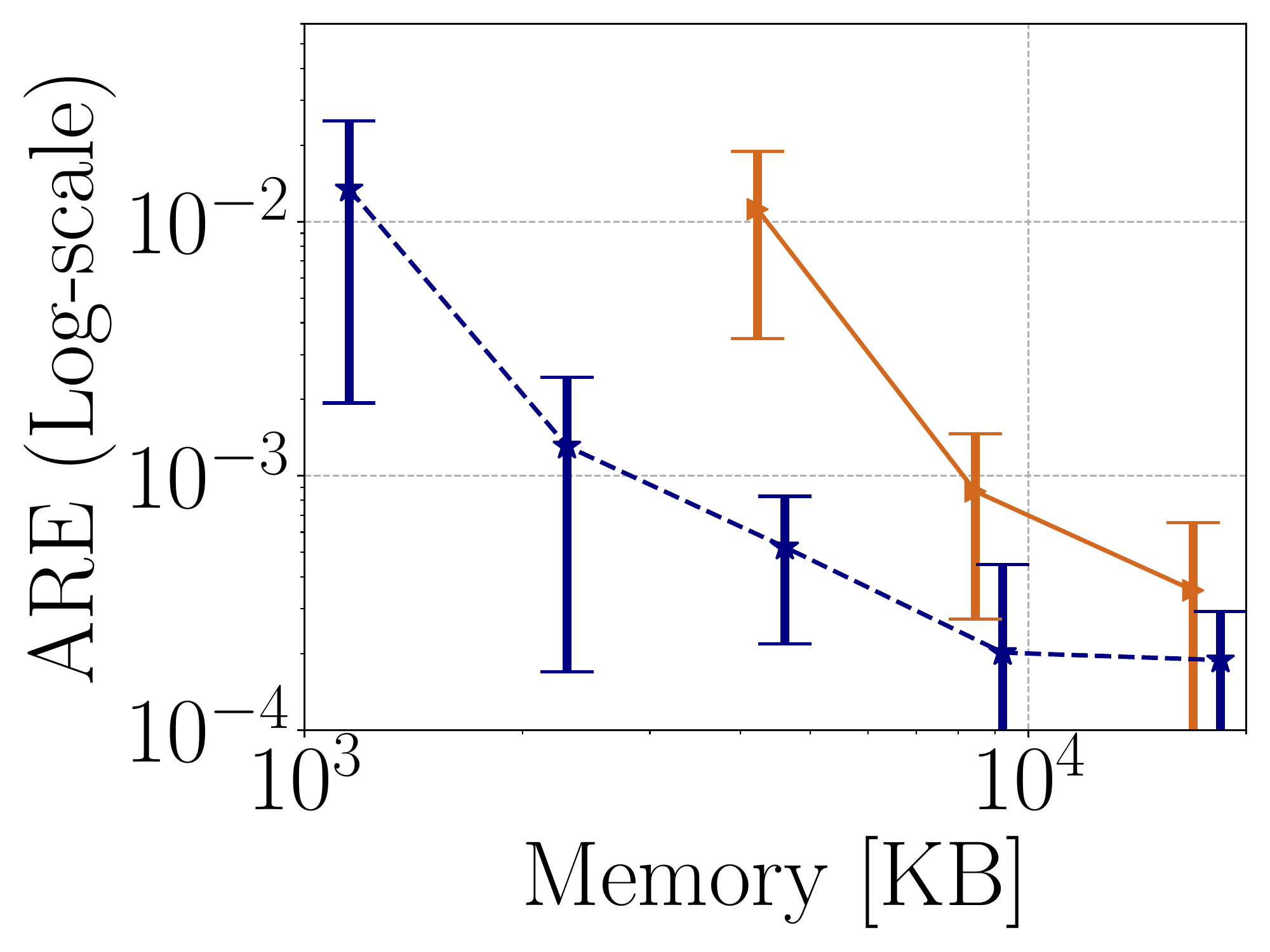}}
    \subfloat[Zipf (8MB)]
    { \includegraphics[width =0.16\textwidth]
    {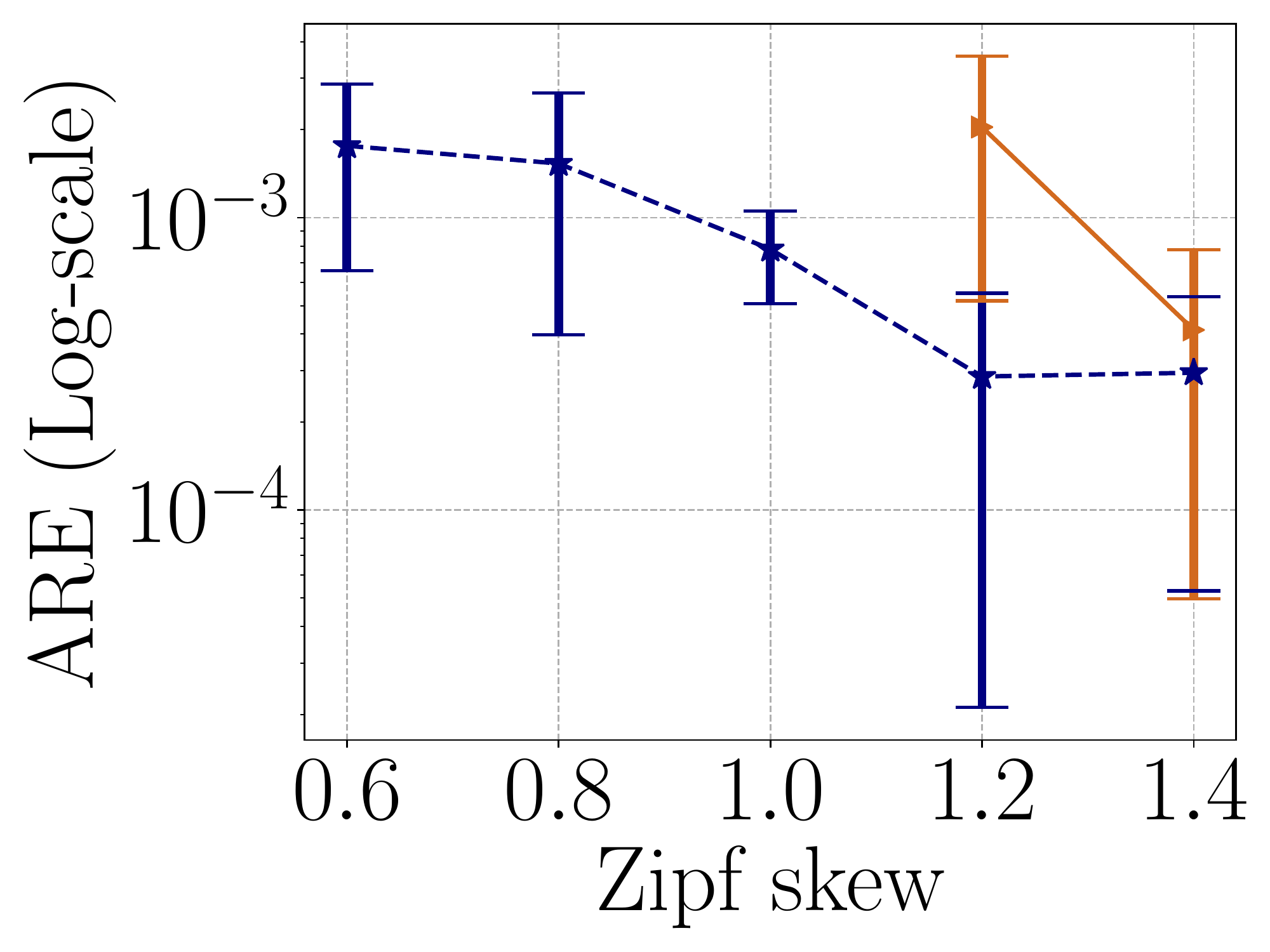}}   \\
    {\vspace*{-0mm}\includegraphics[width =.6\columnwidth]
    {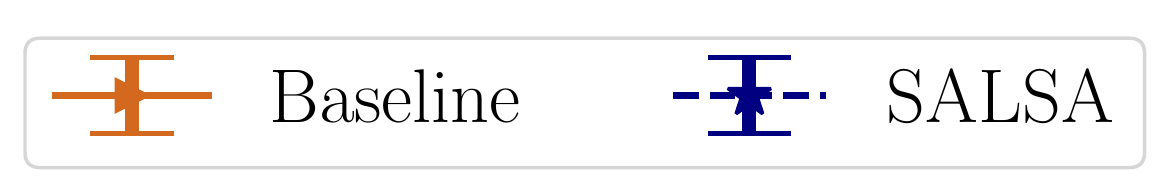}}
   
    \caption{\small Accuracy of SALSA CMS on counting distinct elements.}
    \label{fig:cmsCDandHHapp}
    \ifdefined\sigmodSubmission\vspace*{-3mm}\fi
\end{figure} 
\fi 

\ifdefined\icdeSubmission
\section{Should We Split Counters?}\label{sec:Split Counters}

\begin{figure}[t]
    \centering
    \hspace*{-2mm}
    \subfloat[Error, NY18]
    { \includegraphics[width =0.5\columnwidth]
    {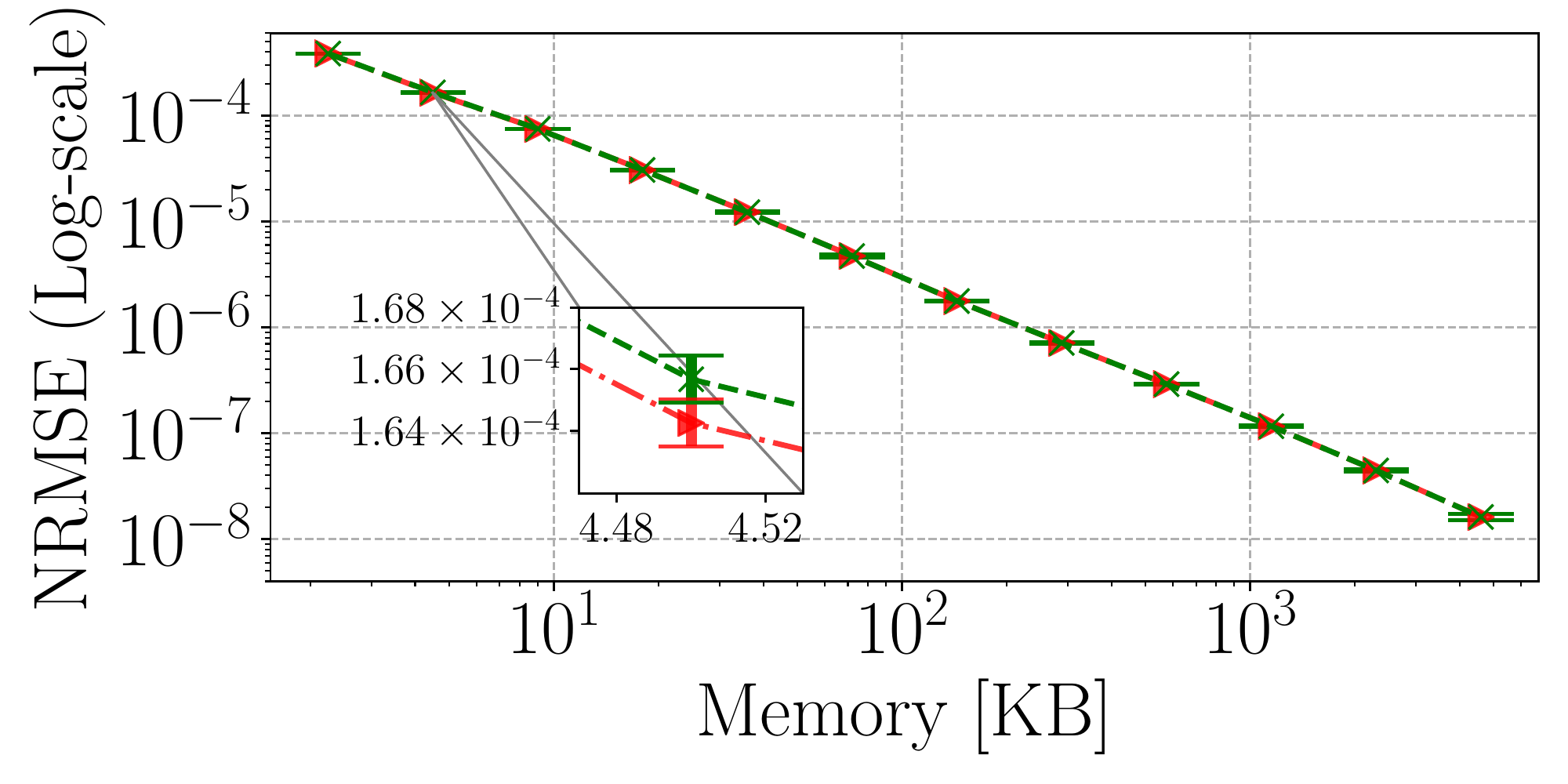}}
    \hspace*{-2mm}
    \subfloat[Error, CH16]
    { \includegraphics[width =0.5\columnwidth]
    {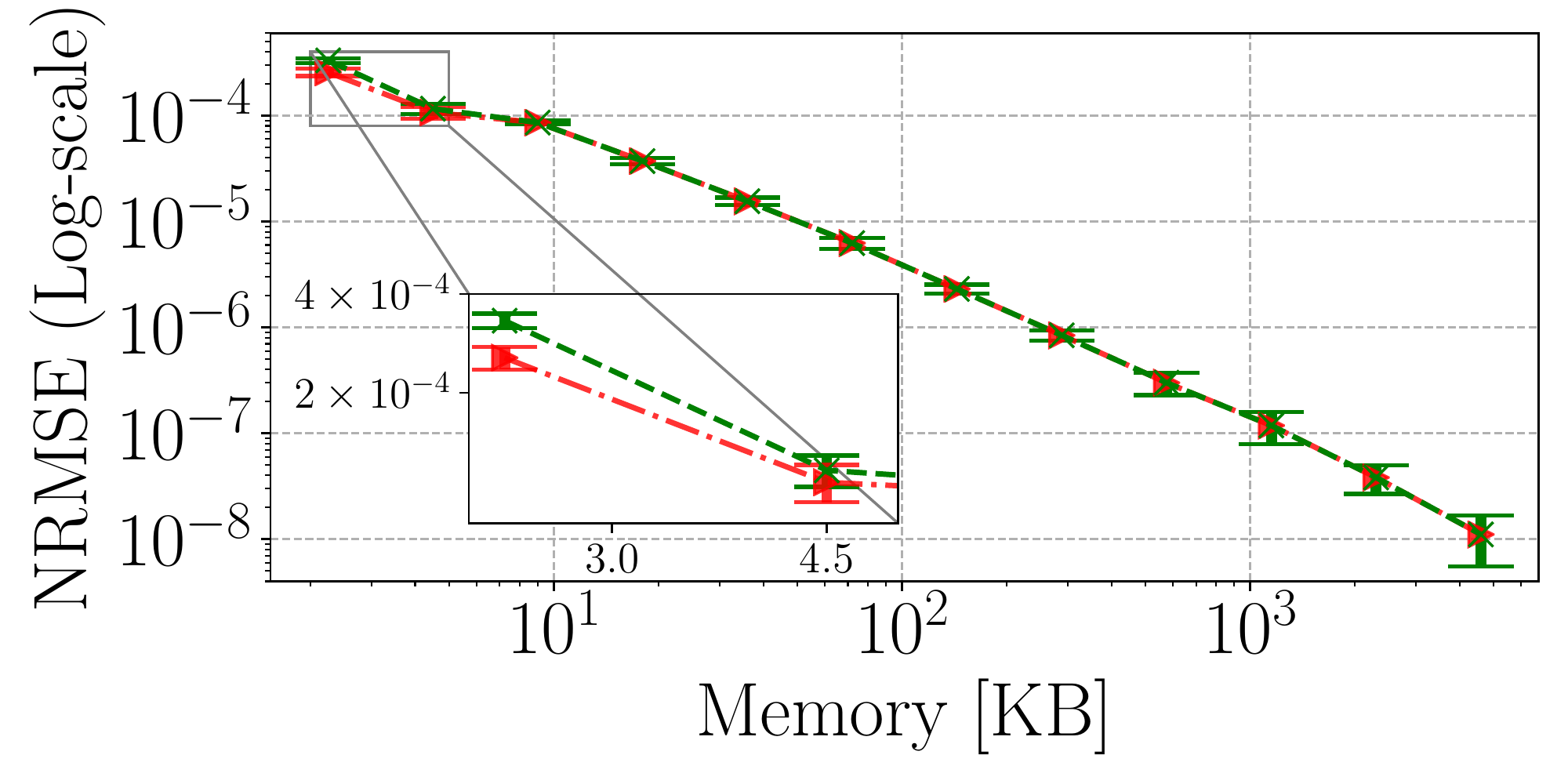}}   
    \hspace*{-1mm}\\
    {\includegraphics[width =0.7\columnwidth]
    {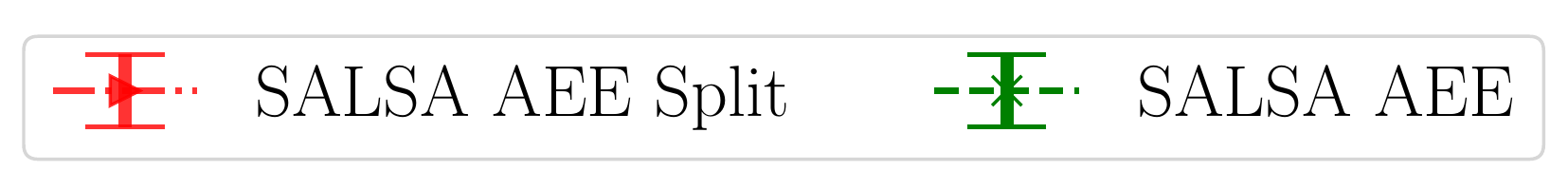}\ifdefined\sigmodSubmission\vspace*{-1mm}\fi}\\
    \caption{Affect of splitting counters in SALSA AEE (CM sketch).}\label{fig:splitting}
    \ifdefined\sigmodSubmission\vspace*{-2mm}\fi
\end{figure}

Finally, we check the accuracy gains obtainable by splitting counters. 
Intuitively, once a counter is downsampled, it may require fewer bits to represent. Therefore, if previously the counter had $s\cdot 2^{\ell}$ bits and the downsampled value is lower than $2^{s\cdot 2^{\ell-1}}-1$ (and $\ell\ge 1$), we can split the counter into two $s\cdot 2^{\ell-1}$-bit counters. As a result, there are now fewer collisions between elements, and SALSA AEE has better accuracy. However, as the results in Figure~\ref{fig:splitting} suggest, this effect is minor, and in most cases, the \mbox{accuracy gains are insignificant.}
\fi

\section{Improved Encoding}\label{app:improvedEncoding}
Here, we lower bound the space required to encode SALSA, and then suggest a near-optimal \mbox{encoding that has $O(1)$ time operations.}

\smallskip
\noindent\textbf{Lower bound.}
%
For $n\in \mathbb N$, we define by $a_n$ the number of possible layouts for a consecutive block of $2^n \cdot s$  bits (i.e., a block that started from $2^n$ counters of size $s$-bits each).
For example, we have $a_2=5$ since the possible combinations for $4$ consecutive counters are $\angles{\set{a},\set{b},\set{c},\set{d}},
\angles{\set{a,b},\set{c},\set{d}},
\angles{\set{a},\set{b},\set{c,d}},\allowbreak
\angles{\set{a,b},\set{c,d}},
\angles{\set{a,b,c,d}}.
$
Observe that either all $2^n$ counters are merged together, or it is enough to specify the layouts of the first $2^{n-1}$ counters and last $2^{n-1}$ counters. Therefore, we get the recursive relation $a_n=a_{n-1}^2+1$ and $a_0=1$.
Given that we start from $w$ counters, this implies that the number of possible layouts is \mbox{lower bounded by $a_{\floor{\log_2 w}}.$}

\begin{lemma}\label{lem:an-bound}
$\forall n\in\mathbb N: \floor{1.5^{2^n}} \le a_n < 1.51^{2^n}$.
\end{lemma}
\begin{proof}
The inequality is easy to verify for $n\le 3$. 
By induction, one can then easily prove \mbox{that $\forall n\ge 3:1.5^{2^n} + 1 < a_n < 1.51^{2^n} - 1.$ \hspace*{-4mm}\qedhere}
\end{proof}

\ifdefined\sigmodSubmission\vspace*{-1mm}\fi
This suggests a lower bound of $\ceil{\log_2 a_{\floor{\log_2 w}}} \le \ceil{2^{\floor{\log_2 w}}\log_2 1.5}$ bits. Specifically, for $w\ge 2^4$ values which are powers of $2$ any encoding must use at least $\log_2 1.5\approx 0.585$ bits per counter.

\smallskip
\noindent\textbf{Near-optimal encoding.}
Denote by $\mathfrak m$ the maximal number of merges a single counter may go through during the execution. 
We note that $\mathfrak m=O(1)$ as we assumed the final counters must fit into $O(1)$ machine words. For example, if we start from $s=2$ bits counters and we assume that counters grow up to $128$ bits, then $\mathfrak m=6$. Let $\overline m = \max\set{5,\mathfrak m}$. Intuitively, we encode every $2^{\overline m}$ counters separately, thereby allowing $O(\overline m) = O(1)$ time size computation. According to Lemma~\ref{lem:an-bound}, we have that 
$z_{\overline m}\triangleq\ceil{\log_2 a_{\overline m}}$ bits are enough to encode the counter-set layout; for example, $z_5=19$ bits are enough to encode the layout of $2^5=32$ counters. Specifically, for $n=\overline m,\overline m-1,\ldots,0$ we write a $z_{n}$-bits value $X_n$ such that $X_n=a_{n}-1$ means that all $2^{n}$ counters are encoded, and otherwise $X_{n-1}\triangleq \floor{X_n/a_{n-1}}$ encodes the layout of the first $2^{n-1}$ counters while $X_{n-1}'\triangleq X_n \mod a_{n-1}$ encodes the layout of the rest (i.e., they are the base-$(a_{n-1})$ digits of $X_n$). As a result, we use $z_{\overline m}$ bits for each consecutive set of $2^{\overline m}$ counters, giving an overhead of $z_{\overline m}/2^{\overline m}$. For $n\ge 5$, we have that $z_n/2^{n}< 0.594$, i.e., we require at most $0.594$ overhead bits per counter.
Computing the size of a counter then becomes simple: we start from $n=\overline m$ and every time check if the value is $a_n-1$, or recurse into either the left or right half depending on the counter index. An example of this process is illustrated in Figure~\ref{fig:complex-encoding}.
While this approach reduces the overhead, the decoding process involves division and modulo \mbox{operations that may reduce the speed.}
%
\begin{figure}[]\ifdefined\sigmodSubmission\vspace*{-2mm}\fi
\centering{
\includegraphics[width = 1\columnwidth]
		{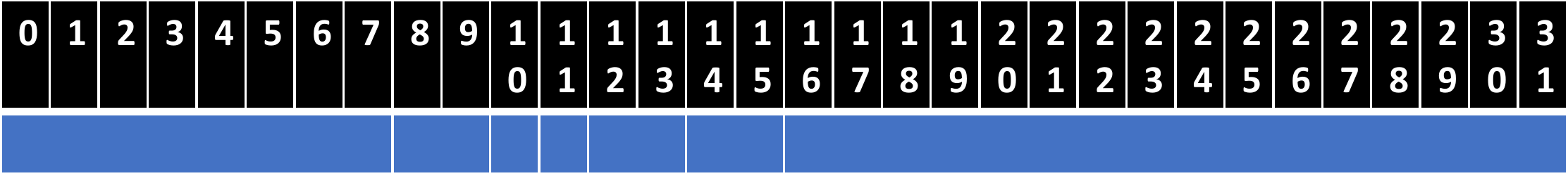}}
	\ifdefined\submissionVersion
	\ifdefined\sigmodSubmission\vskip -0.1cm\fi
	\fi
	\ifdefined\sigmodSubmission\vspace*{-5mm}\fi
    \caption{\label{fig:complex-encoding} \small An encoding example for $\overline m=5$. This layout is encoded by $X_5=449527< a_5$. To compute the size of counter $9$, we first check that $X_5<a_5-1$ and thus not all counters are merged. Next, we have that $X_4=\floor{X_5/a_4}=663<a_4-1$ and thus counters 0-15 are not all merged. Then, we check that $X_3'=X_4\mod a_3=13<a_3-1$ which means that counters 8-15 are not all merged. We continue with $X_2=\floor{X_3'/a_2}=2<a_2-1$ (thus 8-11 are not merged) and finally get $X_1=\floor{X_2/a_1}=1=a_1-1$ which \mbox{implies that $9$ is merged with $8$.}
    }
\end{figure}

\ifdefined\fullversion
\section{Understanding the Differences -- Extended~Results}\label{app:stupidMetrics}
For completeness, we repeat the experiment in Figure~\ref{fig:smallCounters} using even smaller ($4$-bit) counters.
The results are shown in figures~\ref{fig:zeroAlg} and~\ref{fig:zeroAlgAAE}. We measured the error on all heavy hitters -- elements larger than a $\phi$ fraction of the input. The leftmost point ($\phi=10^{-8}$) of Figure~\ref{fig:zeroAlg} corresponds to the ARE metric (i.e., all flows will be considered). As shown, in this case, the best algorithm is $\overline 0$, which corresponds to returning $0$ estimates for all element sizes. That is, according to this metric, one can reduce the error by not running measurements at all. A similar result was observed for AAE, in Figure~\ref{fig:zeroAlgAAE}), where the $\overline 0$ algorithm outperforms the baseline when considering all flows (leftmost point).
\begin{figure}[t]
    \centering
    \ifdefined\sigmodSubmission\vspace*{-1mm}\fi
    \hspace*{-2mm}
    { \includegraphics[width =1.0\columnwidth]
    {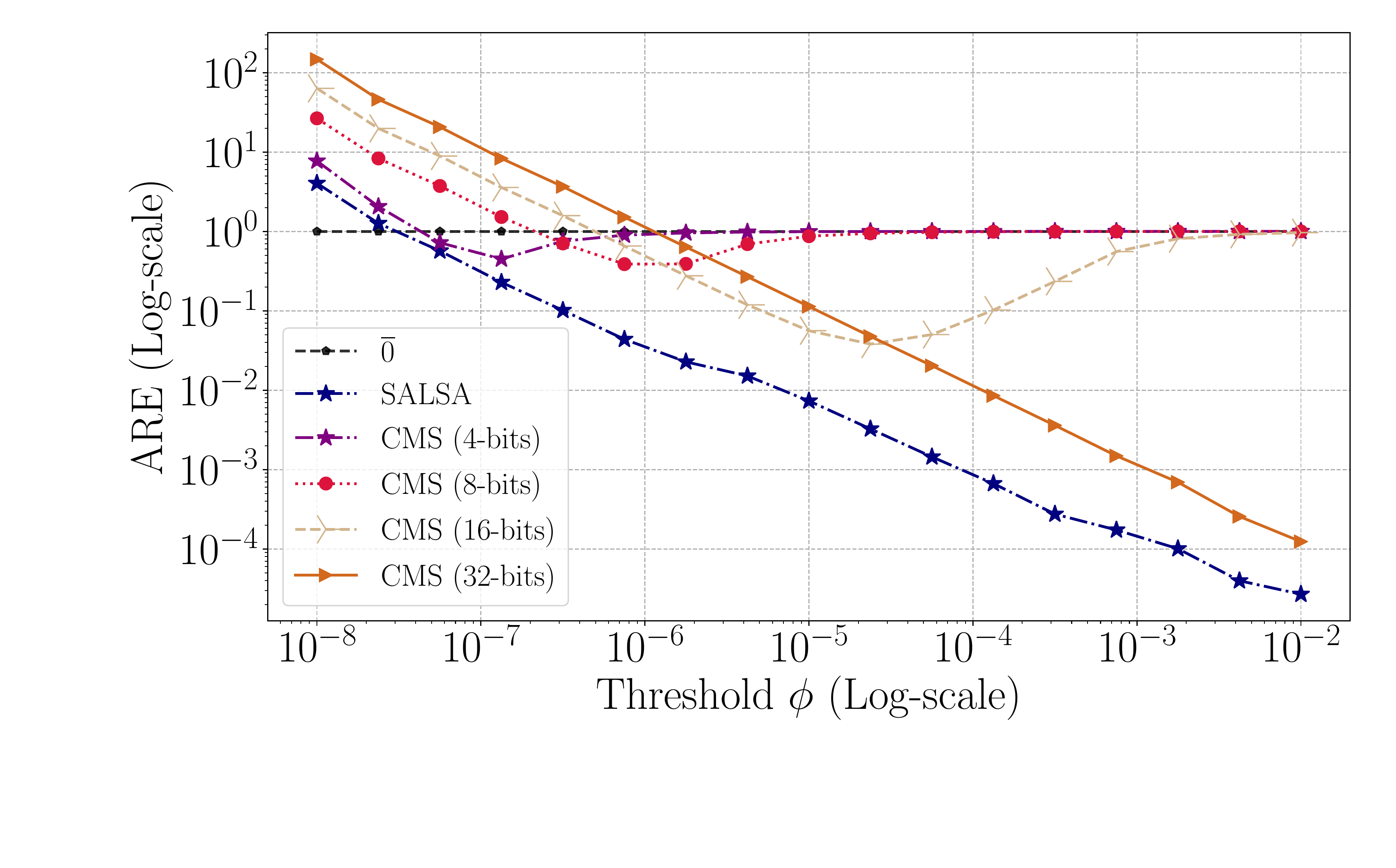}}
    \vspace*{-15mm}
    \caption{\rev{Running CMS with small number of bits and the ``$\overline 0$'' algorithm for estimating heavy hitter sizes (2MB) using average relative error metric. The leftmost point corresponds to the standard ARE metric (used in figures~\ref{fig:areNY18} and \ref{fig:areCH16}), which considers all flows.}}\label{fig:zeroAlg}
    \ifdefined\sigmodSubmission\vspace*{-2mm}\fi
\end{figure}
\begin{figure}[t]
    \centering
    \ifdefined\sigmodSubmission\vspace*{-3mm}\fi
    \hspace*{-2mm}
    { \includegraphics[width =1.0\columnwidth]
    {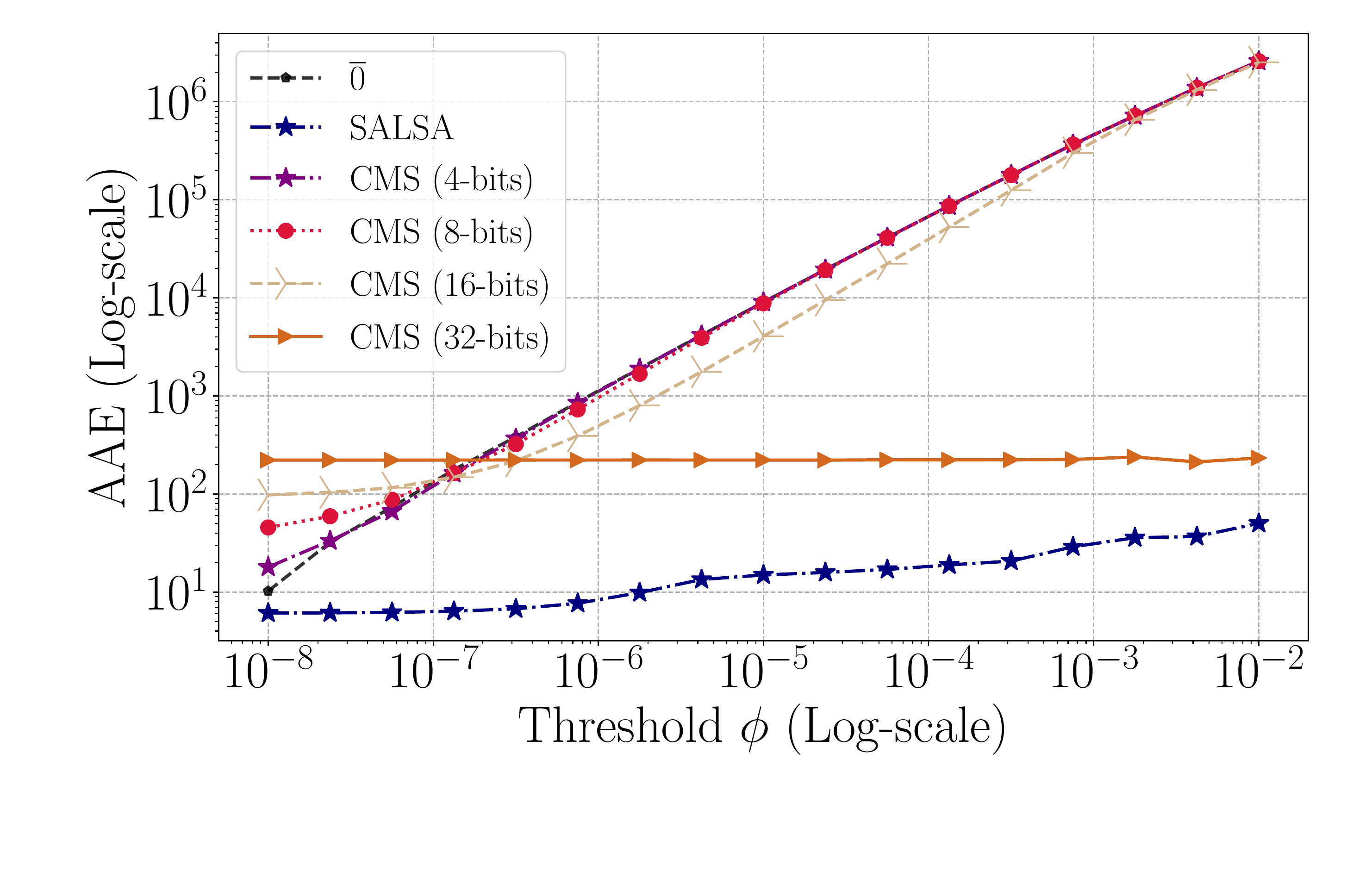}}
    \vspace*{-15mm}
    \caption{\rev{Running CMS with small number of bits and the ``$\overline 0$'' algorithm for estimating heavy hitter sizes (2MB) using average absolute error metric. The leftmost point corresponds to the standard AAE metric (used in figures~\ref{fig:aaeNY18} and \ref{fig:aaeCH16}), which considers all flows.}}\label{fig:zeroAlgAAE}
    \ifdefined\sigmodSubmission\vspace*{-2mm}\fi
\end{figure}

\fi

{ 

}\end{appendices}\end{document}\endinput
\end{appendices}
\end{document}\endinput


\begin{figure}[t]
    \centering
    \subfloat[Error, Random Order]
    { \includegraphics[width =0.5\columnwidth]
    {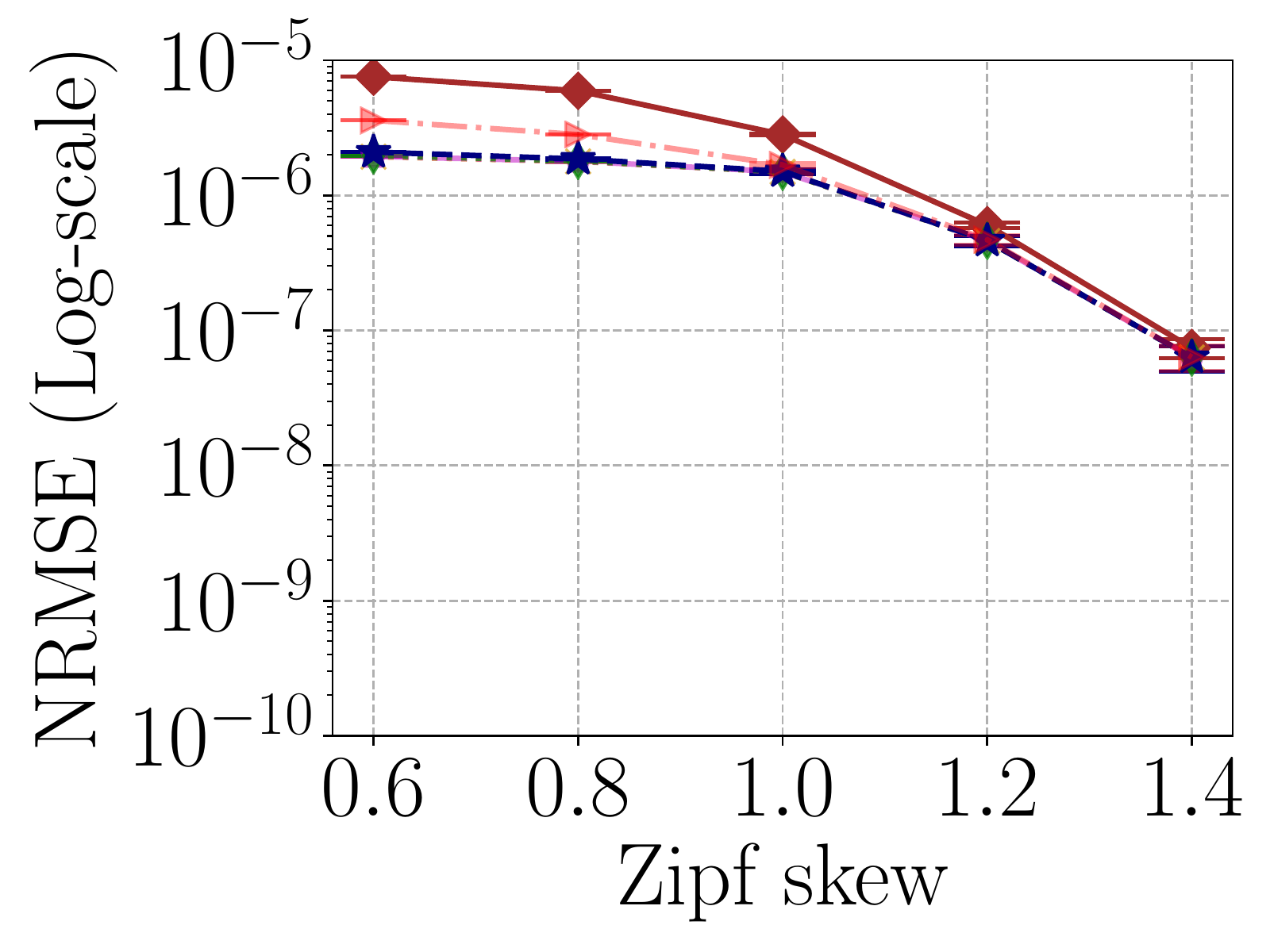}}
    \subfloat[Error, Sequential Order]
    {\includegraphics[width =0.5\columnwidth]
    {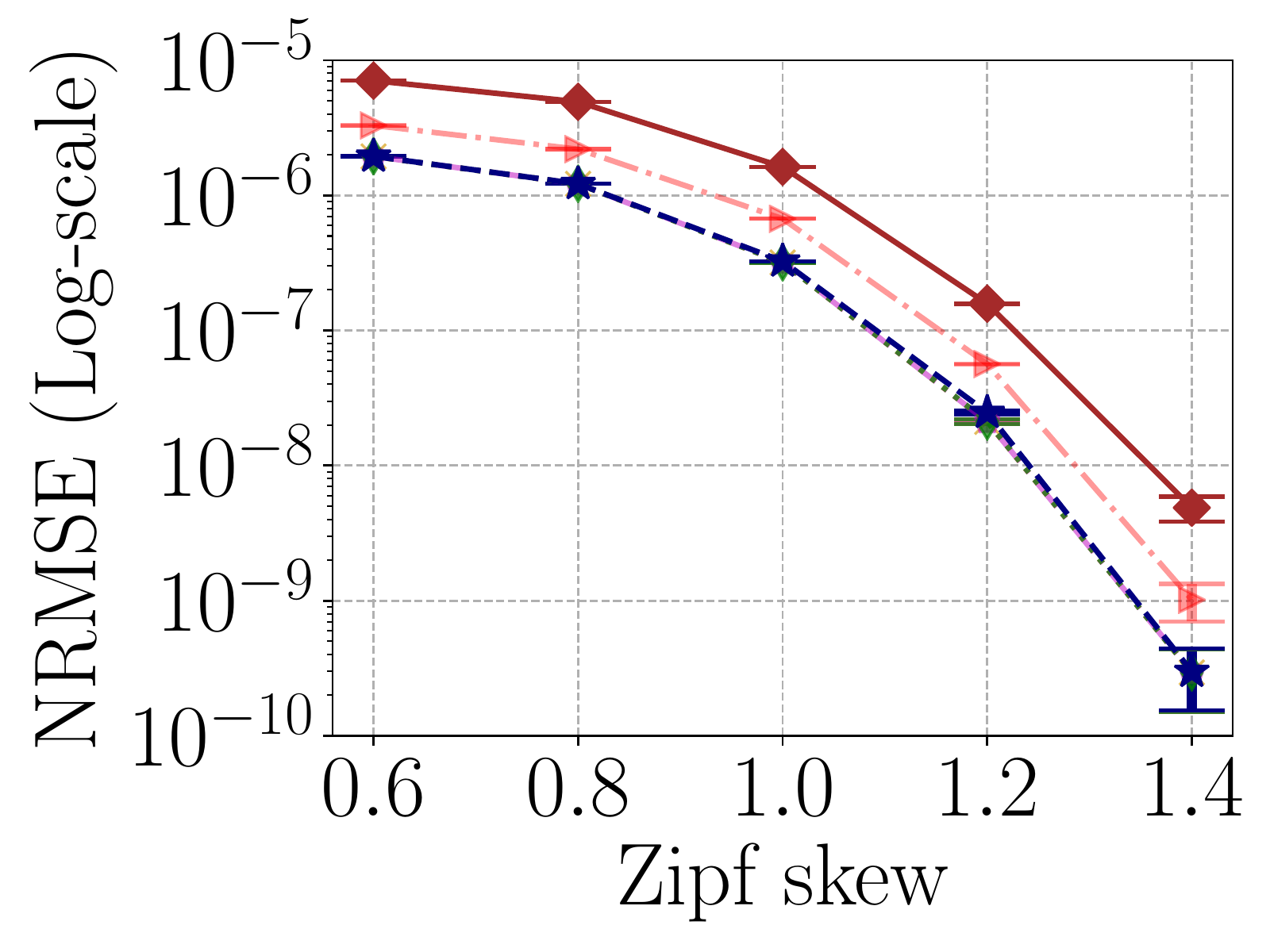}}\\
    {\includegraphics[width =1.04\columnwidth]
    {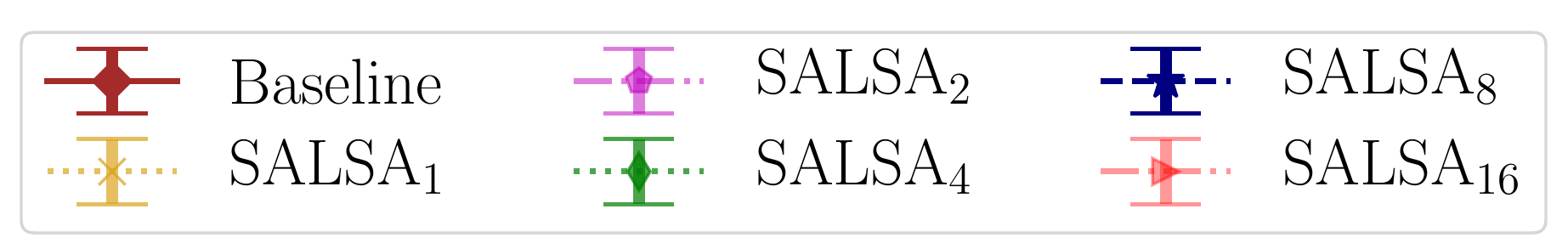}\ifdefined\sigmodSubmission\vspace*{-3mm}\fi}
    \subfloat[Speed, Random Order]
    {\label{5c} \includegraphics[width =0.5\columnwidth]
    {figs/Salsa/countMin/zipf_1024_KB_speed}}
    \subfloat[Speed, Sequential Order]
    {\label{5d}\includegraphics[width =0.5\columnwidth]
    {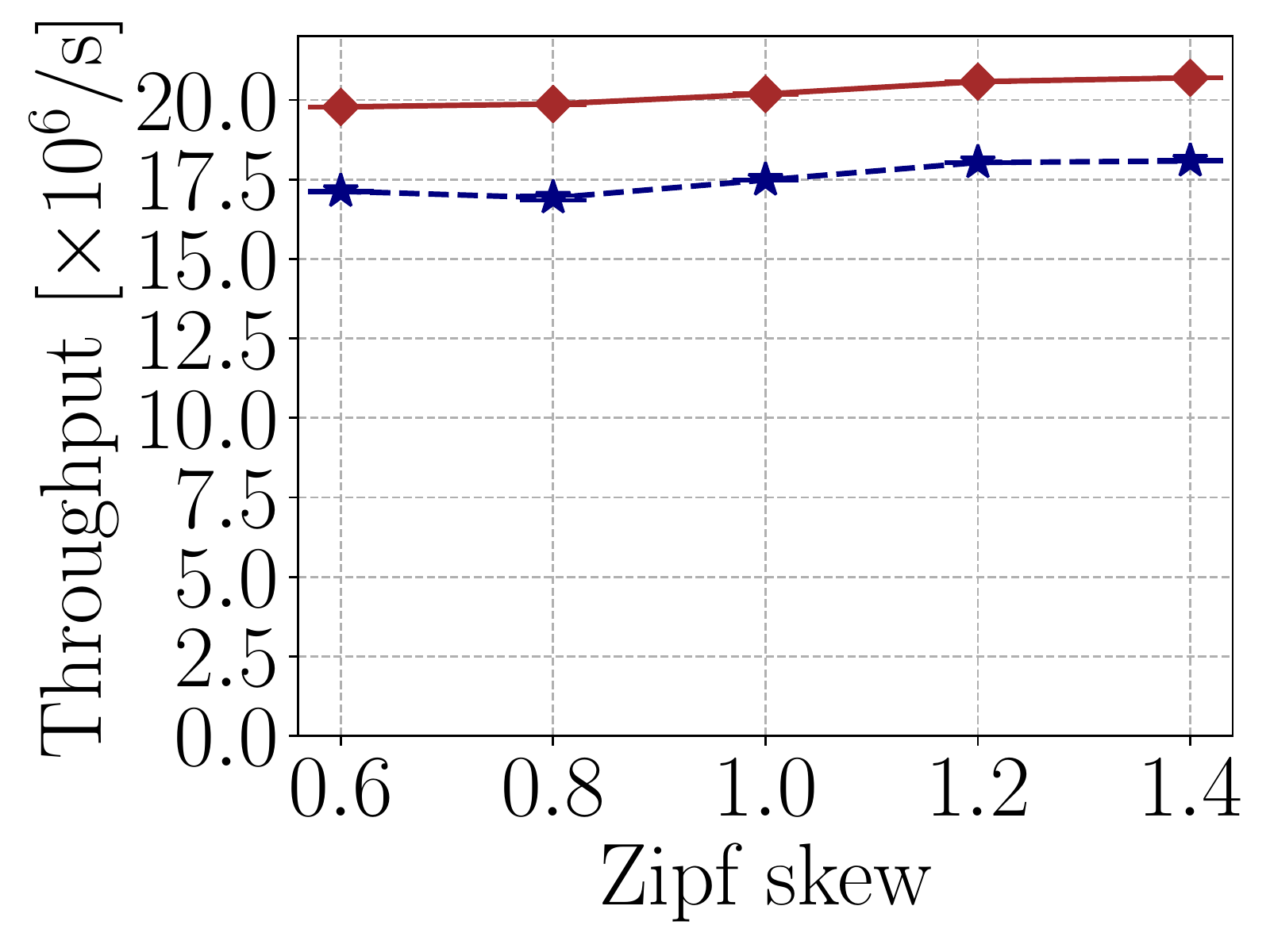}} 
    \ifdefined\sigmodSubmission\vspace*{-1mm}\fi
    \caption{\small Speed and accuracy of SALSA CMS for the synthetic datasets. The Baseline uses $w=2^{16}$ in each row for a total of 1MB of space. SALSA$_s$ is using $w=2^{16}\cdot 32/s$ sized rows for a total of $(1+1/s)$MB space.\ran{We should probably remove }
    }\label{fig:CMS_Zipf}
    \ifdefined\sigmodSubmission\vspace*{-3mm}\fi
\end{figure} 

\section{Proof of the SALSA Count Sketch Algorithm}
Consider a counter $AB$, and let $A$ and $B$ be its two sub-counters. We denote by $F_A$ the set of elements mapped into the sub-counter $A$ and by $F_B$ those mapped to $B$. Let $X_A, X_B$ and $X_{AB}$ denote the values of the counters.
We use $f_y$ to denote the size of a element $y$ and by $g_y$ to denote its direction. 
\begin{figure}[t]
\centering
\includegraphics[width = 0.95\columnwidth]
		{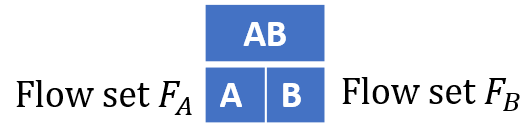}
	\ifdefined\sigmodSubmission\vskip -0.15cm\fi
	\ifdefined\submissionVersion
	\ifdefined\sigmodSubmission\vskip -0.1cm\fi
	\fi
    \caption{\label{fig:assoc_combined} Setting for the proof of Count Sketch SALSA.
    }
\end{figure}

Fix some queried element $x$, and assume without loss of generality that $g_x=1$ and that $x$ is mapped to counter $A$.
We denote by 
$$
\mathfrak E_A = f_x - X_A = f_x - \sum_{y\in F_A} f_yg_y 
$$
and 
$$
\mathfrak E_{AB} = f_x - X_{AB} = f_x - \sum_{y\in F_A\cup F_B} f_yg_y
$$
the errors we get from using the counters to estimate the size of $x$.
We denote the event of an overflow by $\mathfrak O$ and its indicator by $I_{\mathfrak O}$.
Under these notations, the error of SALSA is
$$
\mathfrak E = \mathfrak E_A  (1-I_{\mathfrak O}) + \mathfrak E_{AB} I_{\mathfrak O}.
$$

We first prove that the estimate of SALSA is unbiased.
\begin{lemma}
$$
\mathbb E[\mathfrak E] = 0.
$$
\end{lemma}
\begin{proof}
We have that
\begin{multline*}
\mathbb E[\mathfrak E] = \mathbb E\brackets{\mathfrak E_A  (1-I_{\mathfrak O}) + \mathfrak E_{AB}  I_{\mathfrak O}}\\
= \mathbb E\brackets{\mathfrak E_A | \neg\mathfrak O}\Pr[\neg\mathfrak O] + \mathbb E\brackets{\mathfrak E_{AB} | \mathfrak O}\Pr[\mathfrak O] \\
= \mathbb E\brackets{\mathfrak E_{A} | \neg\mathfrak O }\Pr[\neg\mathfrak O] + \mathbb E\brackets{\mathfrak E_{A} | \mathfrak O}\Pr[\mathfrak O] \\+ \mathbb E\brackets{\mathfrak E_{AB}-\mathfrak E_{A} | \mathfrak O} \Pr[\mathfrak O]\\
= \mathbb E\brackets{\mathfrak E_{A}} + \mathbb E\brackets{\mathfrak E_{AB}-\mathfrak E_{A} | \mathfrak O}\Pr[\mathfrak O] = 0.
\end{multline*}
Here, the final equation follows from the following facts: 

\begin{itemize}
    \item $\mathbb E[\mathfrak E_A]=0$ since $g$ is pairwise independent and $$\Pr(g_y=1)=\Pr(g_y=-1) \,\, \forall y \neq x.$$
    \item  $\mathbb E\brackets{\mathfrak E_{AB}-\mathfrak E_{A} | \mathfrak O}=0$ since $$\mathfrak E_{AB}-\mathfrak E_{A} = X_B$$ and $$\Pr[X_B = k| \mathfrak O]=\Pr[X_B = -k| \mathfrak O] \,\, \forall k.\qedhere$$
\end{itemize}
\end{proof}
As a result, we have that 
\begin{multline}
\Var[\mathfrak E] = \mathbb E[\mathfrak E^2] = \mathbb E\brackets{\parentheses{\mathfrak E_A  (1-I_{\mathfrak O}) + \mathfrak E_{AB} I_{\mathfrak O}}^2} \\
 = \mathbb E\brackets{\mathfrak E_A^2  (1-I_{\mathfrak O}) + \mathfrak E_{AB}^2 I_{\mathfrak O}}\\
 = \mathbb E\brackets{\mathfrak E_A^2|\neg\mathfrak O}\Pr[\neg\mathfrak O] + \mathbb E\brackets{\mathfrak E_{AB}^2 |\mathfrak O}  \Pr[\mathfrak O].\label{eq0}
\end{multline}
We now prove that SALSA can only decrease the estimate variance.
\begin{lemma}
$$
\Var[\mathfrak E]\le \Var[\mathfrak E_{AB}].
$$
\end{lemma}
\begin{proof}
According to~\eqref{eq0}, we have that 
\begin{multline*}
\Var[\mathfrak E] =  \mathbb E\brackets{\mathfrak E_A^2|\neg\mathfrak O}\Pr[\neg\mathfrak O] + \mathbb E\brackets{\mathfrak E_{AB}^2 |\mathfrak O}  \Pr[\mathfrak O] \\
=  \mathbb E\brackets{\mathfrak E_A^2|\neg\mathfrak O}\Pr[\neg\mathfrak O] + \\(\mathbb E\brackets{\mathfrak E_{AB}^2|\neg\mathfrak O}\Pr[\neg\mathfrak O] + \mathbb E\brackets{\mathfrak E_{AB}^2 |\mathfrak O}  \Pr[\mathfrak O])\\ - \mathbb E\brackets{\mathfrak E_{AB}^2|\neg\mathfrak O}\Pr[\neg\mathfrak O] \\
=  \mathbb E\brackets{\mathfrak E_A^2|\neg\mathfrak O}\Pr[\neg\mathfrak O] + Var[\mathfrak E_{AB} ] - \mathbb E\brackets{\mathfrak E_{AB}^2|\neg\mathfrak O}\Pr[\neg\mathfrak O] \\
= Var[\mathfrak E_{AB} ] -   \mathbb E\brackets{\mathfrak E_{AB}^2 -\mathfrak E_A^2 |\neg\mathfrak O}\Pr[\neg\mathfrak O].
\end{multline*}
We are left with proving that $\mathbb E\brackets{\mathfrak E_{AB}^2 -\mathfrak E_A^2 |\neg\mathfrak O} \ge 0$.
Plugging in the values of $\mathfrak E_{AB}$ and $\mathfrak E_{A}$:
\begin{multline*}
\mathbb E\brackets{\mathfrak E_{AB}^2 -\mathfrak E_A^2 |\neg\mathfrak O} = 
\mathbb E\brackets{(f_x - X_{AB})^2 -(f_x - X_{A})^2 |\neg\mathfrak O}  \\ 
=\mathbb E\brackets{(X_{AB}^2- X_{A}^2 - 2f_x(X_{AB} - X_{A}) |\neg\mathfrak O}  \\
=\mathbb E\brackets{X_{AB}^2- X_{A}^2|\neg\mathfrak O} -2f_x\mathbb E\brackets{(X_{AB} - X_{A}) |\neg\mathfrak O}.
\end{multline*}
Recall that $X_B = X_{AB} - X_A$ and that by symmetry $\mathbb E\brackets{X_B |\neg\mathfrak O} = 0$.
Therefore:
\begin{multline*}
\mathbb E\brackets{\mathfrak E_{AB}^2 -\mathfrak E_A^2 |\neg\mathfrak O} \\= \mathbb E\brackets{X_{AB}^2- X_{A}^2|\neg\mathfrak O} -2f_x\mathbb E\brackets{(X_{AB} - X_{A}) |\neg\mathfrak O} \\
= \mathbb E\brackets{X_{AB}^2- X_{A}^2|\neg\mathfrak O} 
= \mathbb E\brackets{(X_{A}+X_B)^2- X_{A}^2|\neg\mathfrak O}\\
= \mathbb E\brackets{X_B^2+ 2X_{A}X_B|\neg\mathfrak O} = \mathbb E\brackets{X_B^2|\neg\mathfrak O} \ge 0.
\end{multline*}
Here, we used the fact that $\mathbb E\brackets{X_{A}X_B|\neg\mathfrak O}=0$ which again follows from the sign-symmetry of $B$.\qedhere
\end{proof}

\begin{figure}[t]
\centering
\includegraphics[width = 0.95\columnwidth]
		{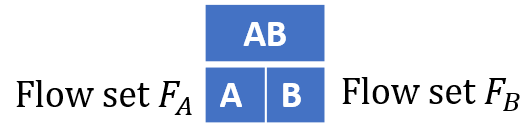}
	\ifdefined\sigmodSubmission\vskip -0.15cm\fi
	\ifdefined\submissionVersion
	\ifdefined\sigmodSubmission\vskip -0.1cm\fi
	\fi
    \caption{\label{fig:assoc_combined} Setting for the proof of Count Sketch SALSA.
    }
\end{figure}

$$
\widehat{F_{AB}}=X_{AB}^2+\sum_{j\neq AB} C[i,j]^2
$$
the $F_2$ estimation when $A$ and $B$ are merged and by 
$$
\widehat{F_{A,B}}=X_{A}^2+X_B^2+\sum_{j\neq AB} C[i,j]^2
$$
the estimation when they are not.

We define by
$$
\widehat{F_2} \triangleq \widehat{F_{A,B}}  (1-I_{\mathfrak O}) + \widehat{F_{AB}} I_{\mathfrak O}
$$
the estimation of SALSA.

We first prove that the estimation is unbiased.
\begin{lemma}
$$
\mathbb E[\widehat{F_2}] = F_2.
$$
\end{lemma}
\begin{proof}
We have that
\begin{multline*}
    \mathbb E[\widehat{F_2}] = \mathbb E[\widehat{F_{A,B}}  (1-I_{\mathfrak O}) + \widehat{F_{AB}} I_{\mathfrak O}]\\
    = \mathbb E[\widehat{F_{A,B}} | \neg \mathfrak O]\Pr[\neg \mathfrak O] + \mathbb E[\widehat{F_{AB}} | \mathfrak O]\Pr[\mathfrak O]\\
    = \mathbb E[\widehat{F_{A,B}} | \neg \mathfrak O]\Pr[\neg \mathfrak O] + \mathbb E[\widehat{F_{A,B}} | \mathfrak O]\Pr[\mathfrak O]\\
    + \mathbb E[\widehat{F_{AB}} - \widehat{F_{A,B}} | \mathfrak O]\Pr[\mathfrak O]\\
    = \mathbb E[\widehat{F_{A,B}}] + \mathbb E[\widehat{F_{AB}} - \widehat{F_{A,B}} | \mathfrak O]\Pr[\mathfrak O]\\
    = \mathbb E[\widehat{F_{A,B}}] + \mathbb E[X_{AB}^2 - (X_{A}^2+X_B^2) | \mathfrak O]\Pr[\mathfrak O]\\
    = \mathbb E[\widehat{F_{A,B}}] + \mathbb E[(X_{A}+X_B)^2 - (X_{A}^2+X_B^2) | \mathfrak O]\Pr[\mathfrak O]\\
    = \mathbb E[\widehat{F_{A,B}}] + 2\mathbb E[X_{A}X_B | \mathfrak O]\Pr[\mathfrak O] = \mathbb E[\widehat{F_{A,B}}]\\
    = \sum_{j}\mathbb E\brackets{\parentheses{\sum_{h(y)=j}f_yg_y}^2} \\
    = \sum_{j}\mathbb E\brackets{\sum_{h(y)=j}f_y^2 + \sum_{h(y)=j, h(z)=j, y\neq z}f_yf_zg_yg_z}\\
    = F_2 +\mathbb E\brackets{\sum_{h(y)=j, h(z)=j, y\neq z}f_yf_zg_yg_z} = F_2.
\end{multline*}
\end{proof}
\begin{lemma}
$$\Var[\widehat{F_2}]\le \Var[\widehat{F_{AB}}].$$
\end{lemma}
\begin{proof}

\end{proof}

\section{Proof sketch for the $k$-tail bounds of SALSA Count Min}
Assume that we have $w$ counters (in a row).
The probability that we hit one of the largest $k$ elements is at most $k/w$.
Conditioned on not hitting any of the $k$ elements, we have an expected noise $F_1^{res(k)}/w$.

Let $k = x/\epsilon$.

Therefore, except with probability $k/w+1/c$, the noise is at most $cF_1^{res(k)}/w$.

Set $w = 4/\epsilon, c = 4$

The noise is at most $\epsilon F_1^{res(k)}$ except with prob. $(1/4 + k/w) = (1/4 + k\epsilon/4)$.

If $k\le w/2$ then we succeed with probability $\ge1/2$.

---

Set $w = 2\sqrt[4]{\delta^{-1}}/\epsilon, c = 2\sqrt[4]{\delta^{-1}}$

The noise is at most $\epsilon F_1^{res(k)}$ except w. prob. $(\sqrt[4]{\delta}/2 + k/w).$

If $k\le w\sqrt[4]{\delta}$ then we succeed with probability $1-\sqrt[4]{\delta}$.

---
Set $w = 10/9\sqrt[4]{\delta^{-1}}/\epsilon, c = 10/9\sqrt[4]{\delta^{-1}}$

The noise is at most $\epsilon F_1^{res(k)}$ except w. prob. $(9\sqrt[4]{\delta}/10 + k/w).$

If $k\le (w\sqrt[4]{\delta})/10$ then we succeed with probability $1-\sqrt[4]{\delta}$.

\fi
\fi
\end{document}